%% file: Reversible_Truly_Concurrent_Process_Algebra.tex
\documentclass{fac}

\usepackage{graphicx}
\usepackage{amssymb}
\usepackage{amsfonts}
\usepackage{amsmath}
\usepackage{dsfont}
\usepackage{MnSymbol}
\usepackage{mathtools}
\usepackage{epsfig}
\usepackage{epstopdf}
\usepackage{tikz}
\usepackage{amsthm}
\ifprodtf \else \usepackage{latexsym}\fi

\usepackage{tikz}
\usetikzlibrary{calc,decorations.pathmorphing,shapes,graphs}

\newcounter{sarrow}

\newcounter{sarrow1}
\newcommand\xnrsquigarrow[1]{%
\stepcounter{sarrow1}%
\mathrel{\begin{tikzpicture}[baseline= {( $ (current bounding box.south) + (0,-0.5ex) $ )}]
\node[inner sep=.5ex] (\thesarrow) {$\scriptstyle #1$};
\path[draw,<-,decorate,
  decoration={zigzag,amplitude=0.7pt,segment length=1.2mm,pre=lineto,pre length=4pt}]
    (\thesarrow1.south east) -- (\thesarrow1.south west);
    $\slashedarrowfill@\relbar\relbar/$
    \end{tikzpicture}}%
}

\makeatletter
\def\slashedarrowfill@#1#2#3#4#5{%
  $\m@th\thickmuskip0mu\medmuskip\thickmuskip\thinmuskip\thickmuskip
   \relax#5#1\mkern-7mu%
   \cleaders\hbox{$#5\mkern-2mu#2\mkern-2mu$}\hfill
   \mathclap{#3}\mathclap{#2}%
   \cleaders\hbox{$#5\mkern-2mu#2\mkern-2mu$}\hfill
   \mkern-7mu#4$%
}
\def\rightslashedarrowfillb@{%
  \slashedarrowfill@\relbar\relbar/\rightarrow}
\newcommand\xnrightarrow[2][]{%
  \ext@arrow 0055{\rightslashedarrowfillb@}{#1}{#2}}

\def\rightslashedarrowfille@{%
  \slashedarrowfill@\relbar\relbar/\twoheadrightarrow}
\newcommand\xntworightarrow[2][]{%
  \ext@arrow 0055{\rightslashedarrowfille@}{#1}{#2}}

\def\rightslashedarrowfillg@{%
  \slashedarrowfill@\relbar\relbar{\raisebox{.12em}{}}\twoheadrightarrow}
\newcommand\xtworightarrow[2][]{%
  \ext@arrow 0055{\rightslashedarrowfillg@}{#1}{#2}}

\def\rightslashedarrowfillx@{%
  \slashedarrowfill@\Relbar\Relbar/\rightrightarrows}
\newcommand\xnTworightarrow[2][]{%
  \ext@arrow 0055{\rightslashedarrowfillx@}{#1}{#2}}

\def\rightslashedarrowfilly@{%
  \slashedarrowfill@\Relbar\Relbar{\raisebox{.12em}{}}\rightrightarrows}
\newcommand\xTworightarrow[2][]{%
  \ext@arrow 0055{\rightslashedarrowfilly@}{#1}{#2}}

\pgfdeclareshape{slash underlined}
{
  \inheritsavedanchors[from=rectangle] 
  \inheritanchorborder[from=rectangle]
  \inheritanchor[from=rectangle]{north}
  \inheritanchor[from=rectangle]{north west}
  \inheritanchor[from=rectangle]{north east}
  \inheritanchor[from=rectangle]{center}
  \inheritanchor[from=rectangle]{west}
  \inheritanchor[from=rectangle]{east}
  \inheritanchor[from=rectangle]{mid}
  \inheritanchor[from=rectangle]{mid west}
  \inheritanchor[from=rectangle]{mid east}
  \inheritanchor[from=rectangle]{base}
  \inheritanchor[from=rectangle]{base west}
  \inheritanchor[from=rectangle]{base east}
  \inheritanchor[from=rectangle]{south}
  \inheritanchor[from=rectangle]{south west}
  \inheritanchor[from=rectangle]{south east}
  \inheritanchorborder[from=rectangle]
  \foregroundpath{
    \southwest \pgf@xa=\pgf@x \pgf@ya=\pgf@y
    \northeast \pgf@xb=\pgf@x \pgf@yb=\pgf@y
    \pgf@xc=\pgf@xa
    \advance\pgf@xc by .5\pgf@xb
    \pgf@yc=\pgf@ya
    \advance\pgf@xc by -1.3pt
    \advance\pgf@yc by -1.8pt
    \pgfpathmoveto{\pgfqpoint{\pgf@xc}{\pgf@yc}}
    \advance\pgf@xc by  2.6pt
    \advance\pgf@yc by  3.6pt
    \pgfpathlineto{\pgfqpoint{\pgf@xc}{\pgf@yc}}
    \pgfpathmoveto{\pgfqpoint{\pgf@xa}{\pgf@ya}}
    \pgfpathlineto{\pgfqpoint{\pgf@xb}{\pgf@ya}}
 }
}
\tikzset{nomorepostaction/.code=\let\tikz@postactions\pgfutil@empty}

\makeatother

\newcommand\black{\ensuremath{\blacktriangleright}}
\newcommand\white{\ensuremath{\vartriangleright}}

  \newcommand\whbl{\white\kern-.1em--\kern-.1em\black}
  \newcommand\blwh{\black\kern-.1em--\kern-.1em\white}
  \newcommand\blbl{\black\kern-.1em--\kern-.1em\black}
  \newcommand\whwh{\white\kern-.1em--\kern-.1em\white}

\newtheorem{theorem}{Theorem}[section]
\newtheorem{definition}[theorem]{Definition}
\newtheorem{proposition}[theorem]{Proposition}

\title[Draft of Reversible Truly Concurrent Process Algebra]
      {Reversible Truly Concurrent Process Algebra}

\author[Yong Wang]
    {Yong Wang\\
     College of Computer Science and Technology,\\
     Faculty of Information Technology,\\
     Beijing University of Technology, Beijing, China\\
     }

\correspond{Yong Wang, Pingleyuan 100, Chaoyang District, Beijing, China.
            e-mail: wangy@bjut.edu.cn}

\pubyear{2018}
\pagerange{\pageref{firstpage}--\pageref{lastpage}}

\begin{document}
\label{firstpage}

\makecorrespond

\maketitle

\begin{abstract}
We design a reversible version of truly concurrent process algebra CTC which is called RCTC. It has good properties modulo several kinds of strongly forward-reverse truly concurrent bisimulations and weakly forward-reverse truly concurrent bisimulations. These properties include monoid laws, static laws, new expansion law for strongly forward-reverse truly concurrent bisimulations, $\tau$ laws for weakly forward-reverse truly concurrent bisimulations, and congruences for strongly and weakly forward-reverse truly concurrent bisimulations.
\end{abstract}


\input{section1}

\input{section2}

\input{section3}

\input{section4}

\input{section5}

\input{section6}

\input{section7}

\input{refs}
\newpage

%

\label{lastpage}

\end{document}

%% file: section1.tex
\section{Introduction}{\label{int}}

Process algebras are well-known formal theory based on the so-called interleaving bisimilarity, such as CCS \cite{CCS} \cite{CC} and ACP \cite{ALNC} \cite{ACP}. We did some works on truly concurrent process algebra, which is called CTC \cite{CTC}.

Reversible computation is another interesting topic, there are researches \cite{CR} \cite{RCCS2} \cite{TCSR} on reversible computation by use of communication key, based on the so-called forward-reverse bisimilarity.

In this paper, we introduce reversible computation in CTC, which is called RCTC. This paper is organized as follows. In section \ref{bg}, we introduce CTC and truly concurrent bisimilarities. In section \ref{ftc}, we give the so-called forward-reverse truly concurrent bisimilarities on which RCTC is based. We give the syntax and operational semantics of RCTC in section \ref{sos}. We discuss the properties of RCTC based on strongly forward-reverse truly concurrent bisimilarities in section \ref{sftcb}, and the properties of RCTC based on weakly forward-reverse truly concurrent bisimilarities in section \ref{wftcb}. Finally, we conclude this paper in section \ref{con}.

%% file: section2.tex
\section{Backgrounds}\label{bg}

In this subsection, we introduce the preliminaries on truly concurrent process algebra CTC \cite{CTC}, which is based on the truly concurrent bisimulation semantics.

\subsection{CTC}\label{CTC}

CTC\cite{CTC} is a calculus of truly concurrent systems. It includes syntax and semantics:

\begin{enumerate}
  \item Its syntax includes actions, process constant, and operators acting between actions, like Prefix, Summation, Composition, Restriction, Relabelling.
  \item Its semantics is based on labeled transition systems, Prefix, Summation, Composition, Restriction, Relabelling have their transition rules. CTC has good semantic properties based on the truly concurrent bisimulations. These properties include monoid laws, static laws, new expansion law for strongly truly concurrent bisimulations, $\tau$ laws for weakly truly concurrent bisimulations, and full congruences for strongly and weakly truly concurrent bisimulations, and also unique solution for recursion.
\end{enumerate}

CTC can be used widely in verification of computer systems with a truly concurrent flavor.

\subsection{Operational Semantics}\label{OS}

The semantics of CTC is based on truly concurrent bisimulation/rooted branching truly concurrent bisimulation equivalences, for the conveniences, we introduce some concepts and conclusions on them.

\begin{definition}[Prime event structure with silent event]\label{PES}
Let $\Lambda$ be a fixed set of labels, ranged over $a,b,c,\cdots$ and $\tau$. A ($\Lambda$-labelled) prime event structure with silent event $\tau$ is a tuple $\mathcal{E}=\langle \mathbb{E}, \leq, \sharp, \lambda\rangle$, where $\mathbb{E}$ is a denumerable set of events, including the silent event $\tau$. Let $\hat{\mathbb{E}}=\mathbb{E}\backslash\{\tau\}$, exactly excluding $\tau$, it is obvious that $\hat{\tau^*}=\epsilon$, where $\epsilon$ is the empty event. Let $\lambda:\mathbb{E}\rightarrow\Lambda$ be a labelling function and let $\lambda(\tau)=\tau$. And $\leq$, $\sharp$ are binary relations on $\mathbb{E}$, called causality and conflict respectively, such that:

\begin{enumerate}
  \item $\leq$ is a partial order and $\lceil e \rceil = \{e'\in \mathbb{E}|e'\leq e\}$ is finite for all $e\in \mathbb{E}$. It is easy to see that $e\leq\tau^*\leq e'=e\leq\tau\leq\cdots\leq\tau\leq e'$, then $e\leq e'$.
  \item $\sharp$ is irreflexive, symmetric and hereditary with respect to $\leq$, that is, for all $e,e',e''\in \mathbb{E}$, if $e\sharp e'\leq e''$, then $e\sharp e''$.
\end{enumerate}

Then, the concepts of consistency and concurrency can be drawn from the above definition:

\begin{enumerate}
  \item $e,e'\in \mathbb{E}$ are consistent, denoted as $e\frown e'$, if $\neg(e\sharp e')$. A subset $X\subseteq \mathbb{E}$ is called consistent, if $e\frown e'$ for all $e,e'\in X$.
  \item $e,e'\in \mathbb{E}$ are concurrent, denoted as $e\parallel e'$, if $\neg(e\leq e')$, $\neg(e'\leq e)$, and $\neg(e\sharp e')$.
\end{enumerate}
\end{definition}

\begin{definition}[Configuration]
Let $\mathcal{E}$ be a PES. A (finite) configuration in $\mathcal{E}$ is a (finite) consistent subset of events $C\subseteq \mathcal{E}$, closed with respect to causality (i.e. $\lceil C\rceil=C$). The set of finite configurations of $\mathcal{E}$ is denoted by $\mathcal{C}(\mathcal{E})$. We let $\hat{C}=C\backslash\{\tau\}$.
\end{definition}

A consistent subset of $X\subseteq \mathbb{E}$ of events can be seen as a pomset. Given $X, Y\subseteq \mathbb{E}$, $\hat{X}\sim \hat{Y}$ if $\hat{X}$ and $\hat{Y}$ are isomorphic as pomsets. In the following of the paper, we say $C_1\sim C_2$, we mean $\hat{C_1}\sim\hat{C_2}$.

\begin{definition}[Pomset transitions and step]
Let $\mathcal{E}$ be a PES and let $C\in\mathcal{C}(\mathcal{E})$, and $\emptyset\neq X\subseteq \mathbb{E}$, if $C\cap X=\emptyset$ and $C'=C\cup X\in\mathcal{C}(\mathcal{E})$, then $C\xrightarrow{X} C'$ is called a pomset transition from $C$ to $C'$. When the events in $X$ are pairwise concurrent, we say that $C\xrightarrow{X}C'$ is a step.
\end{definition}

\begin{definition}[Weak pomset transitions and weak step]
Let $\mathcal{E}$ be a PES and let $C\in\mathcal{C}(\mathcal{E})$, and $\emptyset\neq X\subseteq \hat{\mathbb{E}}$, if $C\cap X=\emptyset$ and $\hat{C'}=\hat{C}\cup X\in\mathcal{C}(\mathcal{E})$, then $C\xRightarrow{X} C'$ is called a weak pomset transition from $C$ to $C'$, where we define $\xRightarrow{e}\triangleq\xrightarrow{\tau^*}\xrightarrow{e}\xrightarrow{\tau^*}$. And $\xRightarrow{X}\triangleq\xrightarrow{\tau^*}\xrightarrow{e}\xrightarrow{\tau^*}$, for every $e\in X$. When the events in $X$ are pairwise concurrent, we say that $C\xRightarrow{X}C'$ is a weak step.
\end{definition}

We will also suppose that all the PESs in this paper are image finite, that is, for any PES $\mathcal{E}$ and $C\in \mathcal{C}(\mathcal{E})$ and $a\in \Lambda$, $\{e\in \mathbb{E}|C\xrightarrow{e} C'\wedge \lambda(e)=a\}$ and $\{e\in\hat{\mathbb{E}}|C\xRightarrow{e} C'\wedge \lambda(e)=a\}$ is finite.

\begin{definition}[Pomset, step bisimulation]\label{PSB}
Let $\mathcal{E}_1$, $\mathcal{E}_2$ be PESs. A pomset bisimulation is a relation $R\subseteq\mathcal{C}(\mathcal{E}_1)\times\mathcal{C}(\mathcal{E}_2)$, such that if $(C_1,C_2)\in R$, and $C_1\xrightarrow{X_1}C_1'$ then $C_2\xrightarrow{X_2}C_2'$, with $X_1\subseteq \mathbb{E}_1$, $X_2\subseteq \mathbb{E}_2$, $X_1\sim X_2$ and $(C_1',C_2')\in R$, and vice-versa. We say that $\mathcal{E}_1$, $\mathcal{E}_2$ are pomset bisimilar, written $\mathcal{E}_1\sim_p\mathcal{E}_2$, if there exists a pomset bisimulation $R$, such that $(\emptyset,\emptyset)\in R$. By replacing pomset transitions with steps, we can get the definition of step bisimulation. When PESs $\mathcal{E}_1$ and $\mathcal{E}_2$ are step bisimilar, we write $\mathcal{E}_1\sim_s\mathcal{E}_2$.
\end{definition}

\begin{definition}[Weak pomset, step bisimulation]\label{WPSB}
Let $\mathcal{E}_1$, $\mathcal{E}_2$ be PESs. A weak pomset bisimulation is a relation $R\subseteq\mathcal{C}(\mathcal{E}_1)\times\mathcal{C}(\mathcal{E}_2)$, such that if $(C_1,C_2)\in R$, and $C_1\xRightarrow{X_1}C_1'$ then $C_2\xRightarrow{X_2}C_2'$, with $X_1\subseteq \hat{\mathbb{E}_1}$, $X_2\subseteq \hat{\mathbb{E}_2}$, $X_1\sim X_2$ and $(C_1',C_2')\in R$, and vice-versa. We say that $\mathcal{E}_1$, $\mathcal{E}_2$ are weak pomset bisimilar, written $\mathcal{E}_1\approx_p\mathcal{E}_2$, if there exists a weak pomset bisimulation $R$, such that $(\emptyset,\emptyset)\in R$. By replacing weak pomset transitions with weak steps, we can get the definition of weak step bisimulation. When PESs $\mathcal{E}_1$ and $\mathcal{E}_2$ are weak step bisimilar, we write $\mathcal{E}_1\approx_s\mathcal{E}_2$.
\end{definition}

\begin{definition}[Posetal product]
Given two PESs $\mathcal{E}_1$, $\mathcal{E}_2$, the posetal product of their configurations, denoted $\mathcal{C}(\mathcal{E}_1)\overline{\times}\mathcal{C}(\mathcal{E}_2)$, is defined as

$$\{(C_1,f,C_2)|C_1\in\mathcal{C}(\mathcal{E}_1),C_2\in\mathcal{C}(\mathcal{E}_2),f:C_1\rightarrow C_2 \textrm{ isomorphism}\}.$$

A subset $R\subseteq\mathcal{C}(\mathcal{E}_1)\overline{\times}\mathcal{C}(\mathcal{E}_2)$ is called a posetal relation. We say that $R$ is downward closed when for any $(C_1,f,C_2),(C_1',f',C_2')\in \mathcal{C}(\mathcal{E}_1)\overline{\times}\mathcal{C}(\mathcal{E}_2)$, if $(C_1,f,C_2)\subseteq (C_1',f',C_2')$ pointwise and $(C_1',f',C_2')\in R$, then $(C_1,f,C_2)\in R$.

For $f:X_1\rightarrow X_2$, we define $f[x_1\mapsto x_2]:X_1\cup\{x_1\}\rightarrow X_2\cup\{x_2\}$, $z\in X_1\cup\{x_1\}$,(1)$f[x_1\mapsto x_2](z)=
x_2$,if $z=x_1$;(2)$f[x_1\mapsto x_2](z)=f(z)$, otherwise. Where $X_1\subseteq \mathbb{E}_1$, $X_2\subseteq \mathbb{E}_2$, $x_1\in \mathbb{E}_1$, $x_2\in \mathbb{E}_2$.
\end{definition}

\begin{definition}[Weakly posetal product]
Given two PESs $\mathcal{E}_1$, $\mathcal{E}_2$, the weakly posetal product of their configurations, denoted $\mathcal{C}(\mathcal{E}_1)\overline{\times}\mathcal{C}(\mathcal{E}_2)$, is defined as

$$\{(C_1,f,C_2)|C_1\in\mathcal{C}(\mathcal{E}_1),C_2\in\mathcal{C}(\mathcal{E}_2),f:\hat{C_1}\rightarrow \hat{C_2} \textrm{ isomorphism}\}.$$

A subset $R\subseteq\mathcal{C}(\mathcal{E}_1)\overline{\times}\mathcal{C}(\mathcal{E}_2)$ is called a weakly posetal relation. We say that $R$ is downward closed when for any $(C_1,f,C_2),(C_1',f,C_2')\in \mathcal{C}(\mathcal{E}_1)\overline{\times}\mathcal{C}(\mathcal{E}_2)$, if $(C_1,f,C_2)\subseteq (C_1',f',C_2')$ pointwise and $(C_1',f',C_2')\in R$, then $(C_1,f,C_2)\in R$.

For $f:X_1\rightarrow X_2$, we define $f[x_1\mapsto x_2]:X_1\cup\{x_1\}\rightarrow X_2\cup\{x_2\}$, $z\in X_1\cup\{x_1\}$,(1)$f[x_1\mapsto x_2](z)=
x_2$,if $z=x_1$;(2)$f[x_1\mapsto x_2](z)=f(z)$, otherwise. Where $X_1\subseteq \hat{\mathbb{E}_1}$, $X_2\subseteq \hat{\mathbb{E}_2}$, $x_1\in \hat{\mathbb{E}}_1$, $x_2\in \hat{\mathbb{E}}_2$. Also, we define $f(\tau^*)=f(\tau^*)$.
\end{definition}

\begin{definition}[(Hereditary) history-preserving bisimulation]\label{HHPB}
A history-preserving (hp-) bisimulation is a posetal relation $R\subseteq\mathcal{C}(\mathcal{E}_1)\overline{\times}\mathcal{C}(\mathcal{E}_2)$ such that if $(C_1,f,C_2)\in R$, and $C_1\xrightarrow{e_1} C_1'$, then $C_2\xrightarrow{e_2} C_2'$, with $(C_1',f[e_1\mapsto e_2],C_2')\in R$, and vice-versa. $\mathcal{E}_1,\mathcal{E}_2$ are history-preserving (hp-)bisimilar and are written $\mathcal{E}_1\sim_{hp}\mathcal{E}_2$ if there exists a hp-bisimulation $R$ such that $(\emptyset,\emptyset,\emptyset)\in R$.

A hereditary history-preserving (hhp-)bisimulation is a downward closed hp-bisimulation. $\mathcal{E}_1,\mathcal{E}_2$ are hereditary history-preserving (hhp-)bisimilar and are written $\mathcal{E}_1\sim_{hhp}\mathcal{E}_2$.
\end{definition}

\begin{definition}[Weak (hereditary) history-preserving bisimulation]\label{WHHPB}
A weak history-preserving (hp-) bisimulation is a weakly posetal relation $R\subseteq\mathcal{C}(\mathcal{E}_1)\overline{\times}\mathcal{C}(\mathcal{E}_2)$ such that if $(C_1,f,C_2)\in R$, and $C_1\xRightarrow{e_1} C_1'$, then $C_2\xRightarrow{e_2} C_2'$, with $(C_1',f[e_1\mapsto e_2],C_2')\in R$, and vice-versa. $\mathcal{E}_1,\mathcal{E}_2$ are weak history-preserving (hp-)bisimilar and are written $\mathcal{E}_1\approx_{hp}\mathcal{E}_2$ if there exists a hp-bisimulation $R$ such that $(\emptyset,\emptyset,\emptyset)\in R$.

A weakly hereditary history-preserving (hhp-)bisimulation is a downward closed weak hp-bisimulation. $\mathcal{E}_1,\mathcal{E}_2$ are weakly hereditary history-preserving (hhp-)bisimilar and are written $\mathcal{E}_1\approx_{hhp}\mathcal{E}_2$.
\end{definition}

\begin{definition}[Congruence]
Let $\Sigma$ be a signature. An equivalence relation $R$ on $\mathcal{T}(\Sigma)$ is a congruence if for each $f\in\Sigma$, if $s_i R t_i$ for $i\in\{1,\cdots,ar(f)\}$, then $f(s_1,\cdots,s_{ar(f)}) R f(t_1,\cdots,t_{ar(f)})$.
\end{definition}

%% file: section3.tex
\section{Forward-reverse Truly Concurrent Bisimulations}{\label{ftc}}

\begin{definition}[Forward-reverse (FR) pomset transitions and forward-reverse (FR) step]
Let $\mathcal{E}$ be a PES and let $C\in\mathcal{C}(\mathcal{E})$, $\emptyset\neq X\subseteq \mathbb{E}$, $\mathcal{K}\subseteq \mathbb{N}$, and $X[\mathcal{K}]$ denotes that for each $e\in X$, there is $e[m]\in X[\mathcal{K}]$ where $(m\in\mathcal{K})$, which is called the past of $e$. If $C\cap X[\mathcal{K}]=\emptyset$ and $C'=C\cup X[\mathcal{K}], X\in\mathcal{C}(\mathcal{E})$, then $C\xrightarrow{X} C'$ is called a forward pomset transition from $C$ to $C'$, and $C'\xtworightarrow{X[\mathcal{K}]} C$ is called a reverse pomset transition from $C'$ to $C$. When the events in $X$ are pairwise concurrent, we say that $C\xrightarrow{X}C'$ is a forward step and $C'\xtworightarrow{X[\mathcal{K}]} C$ is a reverse step.
\end{definition}

\begin{definition}[Weak forward-reverse (FR) pomset transitions and weak forward-reverse (FR) step]
Let $\mathcal{E}$ be a PES and let $C\in\mathcal{C}(\mathcal{E})$, and $\emptyset\neq X\subseteq \hat{\mathbb{E}}$, $\mathcal{K}\subseteq \mathbb{N}$, and $X[\mathcal{K}]$ denotes that for each $e\in X$, there is $e[m]\in X[\mathcal{K}]$ where $(m\in\mathcal{K})$, which is called the past of $e$. If $C\cap X[\mathcal{K}]=\emptyset$ and $\hat{C'}=\hat{C}\cup X[\mathcal{K}], X\in\mathcal{C}(\mathcal{E})$, then $C\xRightarrow{X} C'$ is called a weak forward pomset transition from $C$ to $C'$, where we define $\xRightarrow{e}\triangleq\xrightarrow{\tau^*}\xrightarrow{e}\xrightarrow{\tau^*}$ and $\xRightarrow{X}\triangleq\xrightarrow{\tau^*}\xrightarrow{e}\xrightarrow{\tau^*}$, for every $e\in X$. And $C'\xTworightarrow{X[\mathcal{K}]} C$ is called a weak reverse pomset transition from $C'$ to $C$, where we define $\xTworightarrow{e[m]}\triangleq\xtworightarrow{\tau^*}\xtworightarrow{e[m]}\xtworightarrow{\tau^*}$, $\xTworightarrow{X[\mathcal{K}]}\triangleq\xtworightarrow{\tau^*}\xtworightarrow{e[m]} \xtworightarrow{\tau^*}$, for every $e\in X$ and $m\in\mathcal{K}$. When the events in $X$ are pairwise concurrent, we say that $C\xRightarrow{X}C'$ is a weak forward step and $C'\xTworightarrow{X[\mathcal{K}]} C$ is a weak reverse step.
\end{definition}

We will also suppose that all the PESs in this paper are image finite, that is, for any PES $\mathcal{E}$ and $C\in \mathcal{C}(\mathcal{E})$, and $a\in \Lambda$, $\{e\in \mathbb{E}|C\xrightarrow{e} C'\wedge \lambda(e)=a\}$ and $\{e\in\hat{\mathbb{E}}|C\xRightarrow{e} C'\wedge \lambda(e)=a\}$, and $a\in \Lambda$, $\{e\in \mathbb{E}|C'\xtworightarrow{e[m]} C\wedge \lambda(e)=a\}$ and $\{e\in\hat{\mathbb{E}}|C'\xTworightarrow{e[m]} C\wedge \lambda(e)=a\}$ are finite.

\begin{definition}[Forward-reverse (FR) pomset, step bisimulation]\label{FRPSB}
Let $\mathcal{E}_1$, $\mathcal{E}_2$ be PESs. An FR pomset bisimulation is a relation $R\subseteq\mathcal{C}(\mathcal{E}_1)\times\mathcal{C}(\mathcal{E}_2)$, such that (1) if $(C_1,C_2)\in R$, and $C_1\xrightarrow{X_1}C_1'$ then $C_2\xrightarrow{X_2}C_2'$, with $X_1\subseteq \mathbb{E}_1$, $X_2\subseteq \mathbb{E}_2$, $X_1\sim X_2$ and $(C_1',C_2')\in R$, and vice-versa; (2) if $(C_1',C_2')\in R$, and $C_1'\xtworightarrow{X_1[\mathcal{K}_1]}C_1$ then $C_2'\xtworightarrow{X_2[\mathcal{K}_2]}C_2$, with $X_1\subseteq \mathbb{E}_1$, $X_2\subseteq \mathbb{E}_2$, $\mathcal{K}_1,\mathcal{K}_2\subseteq\mathbb{N}$, $X_1\sim X_2$ and $(C_1,C_2)\in R$, and vice-versa. We say that $\mathcal{E}_1$, $\mathcal{E}_2$ are FR pomset bisimilar, written $\mathcal{E}_1\sim_p^{fr}\mathcal{E}_2$, if there exists an FR pomset bisimulation $R$, such that $(\emptyset,\emptyset)\in R$. By replacing FR pomset transitions with FR steps, we can get the definition of FR step bisimulation. When PESs $\mathcal{E}_1$ and $\mathcal{E}_2$ are FR step bisimilar, we write $\mathcal{E}_1\sim_s^{fr}\mathcal{E}_2$.
\end{definition}

\begin{definition}[Weak forward-reverse (FR) pomset, step bisimulation]\label{FRWPSB}
Let $\mathcal{E}_1$, $\mathcal{E}_2$ be PESs. A weak FR pomset bisimulation is a relation $R\subseteq\mathcal{C}(\mathcal{E}_1)\times\mathcal{C}(\mathcal{E}_2)$, such that (1) if $(C_1,C_2)\in R$, and $C_1\xRightarrow{X_1}C_1'$ then $C_2\xRightarrow{X_2}C_2'$, with $X_1\subseteq \hat{\mathbb{E}_1}$, $X_2\subseteq \hat{\mathbb{E}_2}$, $X_1\sim X_2$ and $(C_1',C_2')\in R$, and vice-versa; (2) if $(C_1',C_2')\in R$, and $C_1'\xTworightarrow{X_1[\mathcal{K}_1]}C_1$ then $C_2'\xTworightarrow{X_2[\mathcal{K}_2]}C_2$, with $X_1\subseteq \hat{\mathbb{E}_1}$, $X_2\subseteq \hat{\mathbb{E}_2}$, $\mathcal{K}_1,\mathcal{K}_2\subseteq\mathbb{N}$, $X_1\sim X_2$ and $(C_1,C_2)\in R$, and vice-versa. We say that $\mathcal{E}_1$, $\mathcal{E}_2$ are weak FR pomset bisimilar, written $\mathcal{E}_1\approx_p^{fr}\mathcal{E}_2$, if there exists a weak FR pomset bisimulation $R$, such that $(\emptyset,\emptyset)\in R$. By replacing weak FR pomset transitions with weak FR steps, we can get the definition of weak FR step bisimulation. When PESs $\mathcal{E}_1$ and $\mathcal{E}_2$ are weak FR step bisimilar, we write $\mathcal{E}_1\approx_s^{fr}\mathcal{E}_2$.
\end{definition}

\begin{definition}[Forward-reverse (FR) (hereditary) history-preserving bisimulation]\label{FRHHPB}
An FR history-preserving (hp-) bisimulation is a posetal relation $R\subseteq\mathcal{C}(\mathcal{E}_1)\overline{\times}\mathcal{C}(\mathcal{E}_2)$ such that (1) if $(C_1,f,C_2)\in R$, and $C_1\xrightarrow{e_1} C_1'$, then $C_2\xrightarrow{e_2} C_2'$, with $(C_1',f[e_1\mapsto e_2],C_2')\in R$, and vice-versa, (2) if $(C_1',f',C_2')\in R$, and $C_1'\xtworightarrow{e_1[m]} C_1$, then $C_2'\xtworightarrow{e_2[n]} C_2$, with $(C_1,f'[e_1[m]\mapsto e_2[n]],C_2)\in R$, and vice-versa. $\mathcal{E}_1,\mathcal{E}_2$ are FR history-preserving (hp-) bisimilar and are written $\mathcal{E}_1\sim_{hp}^{fr}\mathcal{E}_2$ if there exists an FR hp-bisimulation $R$ such that $(\emptyset,\emptyset,\emptyset)\in R$.

An FR hereditary history-preserving (hhp-)bisimulation is a downward closed FR hp-bisimulation. $\mathcal{E}_1,\mathcal{E}_2$ are FR hereditary history-preserving (hhp-)bisimilar and are written $\mathcal{E}_1\sim_{hhp}^{fr}\mathcal{E}_2$.
\end{definition}

\begin{definition}[Weak forward-reverse (FR) (hereditary) history-preserving bisimulation]\label{FRWHHPB}
A weak FR history-preserving (hp-) bisimulation is a weakly posetal relation $R\subseteq\mathcal{C}(\mathcal{E}_1)\overline{\times}\mathcal{C}(\mathcal{E}_2)$ such that (1) if $(C_1,f,C_2)\in R$, and $C_1\xRightarrow{e_1} C_1'$, then $C_2\xRightarrow{e_2} C_2'$, with $(C_1',f[e_1\mapsto e_2],C_2')\in R$, and vice-versa, (2) if $(C_1',f',C_2')\in R$, and $C_1'\xTworightarrow{e_1[m]} C_1$, then $C_2'\xTworightarrow{e_2[n]} C_2$, with $(C_1,f'[e_1[m]\mapsto e_2[n]],C_2)\in R$, and vice-versa. $\mathcal{E}_1,\mathcal{E}_2$ are weak FR history-preserving (hp-) bisimilar and are written $\mathcal{E}_1\approx_{hp}^{fr}\mathcal{E}_2$ if there exists a weak FR hp-bisimulation $R$ such that $(\emptyset,\emptyset,\emptyset)\in R$.

A weak FR hereditary history-preserving (hhp-) bisimulation is a downward closed weak FR hp-bisimulation. $\mathcal{E}_1,\mathcal{E}_2$ are weak FR hereditary history-preserving (hhp-) bisimilar and are written $\mathcal{E}_1\approx_{hhp}^{fr}\mathcal{E}_2$.
\end{definition}

%% file: section4.tex
\section{Syntax and Operational Semantics}\label{sos}

We assume an infinite set $\mathcal{N}$ of (action or event) names, and use $a,b,c,\cdots$ to range over $\mathcal{N}$. We denote by $\overline{\mathcal{N}}$ the set of co-names and let $\overline{a},\overline{b},\overline{c},\cdots$ range over $\overline{\mathcal{N}}$. Then we set $\mathcal{L}=\mathcal{N}\cup\overline{\mathcal{N}}$ as the set of labels, and use $l,\overline{l}$ to range over $\mathcal{L}$. We extend complementation to $\mathcal{L}$ such that $\overline{\overline{a}}=a$. Let $\tau$ denote the silent step (internal action or event) and define $Act=\mathcal{L}\cup\{\tau\}\cup\mathcal{L}[\mathcal{K}]$ to be the set of actions, $\alpha,\beta$ range over $Act$. And $K,L$ are used to stand for subsets of $\mathcal{L}$ and $\overline{L}$ is used for the set of complements of labels in $L$. A relabelling function $f$ is a function from $\mathcal{L}$ to $\mathcal{L}$ such that $f(\overline{l})=\overline{f(l)}$. By defining $f(\tau)=\tau$, we extend $f$ to $Act$. We write $\mathcal{P}$ for the set of processes. Sometimes, we use $I,J$ to stand for an indexing set, and we write $E_i:i\in I$ for a family of expressions indexed by $I$. $Id_D$ is the identity function or relation over set $D$. 

For each process constant schema $A$, a defining equation of the form

$$A\overset{\text{def}}{=}P$$

is assumed, where $P$ is a process.

\subsection{Syntax}

We use the Prefix $.$ to model the causality relation $\leq$ in true concurrency, the Summation $+$ to model the conflict relation $\sharp$ in true concurrency, and the Composition $\parallel$ to explicitly model concurrent relation in true concurrency. And we follow the conventions of process algebra.

\begin{definition}[Syntax]\label{syntax}
Reversible truly concurrent processes RCTC are defined inductively by the following formation rules:

\begin{enumerate}
  \item $A\in\mathcal{P}$;
  \item $\textbf{nil}\in\mathcal{P}$;
  \item if $P\in\mathcal{P}$, then the Prefix $\alpha.P\in\mathcal{P}$ and $P.\alpha[m]\in\mathcal{P}$, for $\alpha\in Act$ and $m\in\mathcal{K}$;
  \item if $P,Q\in\mathcal{P}$, then the Summation $P+Q\in\mathcal{P}$;
  \item if $P,Q\in\mathcal{P}$, then the Composition $P\parallel Q\in\mathcal{P}$;
  \item if $P\in\mathcal{P}$, then the Prefix $(\alpha_1\parallel\cdots\parallel\alpha_n).P\in\mathcal{P}\quad(n\in I)$ and $P.(\alpha_1[m]\parallel\cdots\parallel\alpha_n[m])\in\mathcal{P}\quad(n\in I)$, for $\alpha_,\cdots,\alpha_n\in Act$ and $m\in\mathcal{K}$;
  \item if $P\in\mathcal{P}$, then the Restriction $P\setminus L\in\mathcal{P}$ with $L\in\mathcal{L}$;
  \item if $P\in\mathcal{P}$, then the Relabelling $P[f]\in\mathcal{P}$.
\end{enumerate}

The standard BNF grammar of syntax of RCTC can be summarized as follows:

$P::=A\quad|\quad\textbf{nil}\quad|\quad\alpha.P\quad|\quad P.\alpha[m]\quad|\quad P+P\quad |\quad P\parallel P\quad |\quad (\alpha_1\parallel\cdots\parallel\alpha_n).P| \quad P.(\alpha_1[m]\parallel\cdots\parallel\alpha_n[m]) \quad |\quad P\setminus L\quad |\quad P[f].$
\end{definition}

\subsection{Operational Semantics}

The operational semantics is defined by LTSs (labelled transition systems), and it is detailed by the following definition.

\begin{definition}[Semantics]\label{semantics}
The operational semantics of CTC corresponding to the syntax in Definition \ref{syntax} is defined by a series of transition rules, they are shown in Table \ref{FTRForPS}, \ref{RTRForPS}, \ref{FPTRForPS}, \ref{RPTRForPS}, \ref{FTRForCom}, \ref{RTRForCom}, \ref{FTRForRRC} and \ref{RTRForRRC}. And the predicate $\xrightarrow{\alpha}\alpha[m]$ represents successful forward termination after execution of the action $\alpha$, the predicate $\xtworightarrow{\alpha[m]}\alpha$ represents successful reverse termination after execution of the event $\alpha[m]$, the the predicate $\textrm{Std(P)}$ represents that $p$ is a standard process containing no past events, the the predicate $\textrm{NStd(P)}$ represents that $P$ is a process full of past events.
\end{definition}

The forward transition rules for Prefix and Summation are shown in Table \ref{FTRForPS}.

\begin{center}
    \begin{table}
        $$\frac{}{\alpha\xrightarrow{\alpha}\alpha[m]}$$

        $$\frac{P\xrightarrow{\alpha}\alpha[m] \quad \alpha\notin Q}{P+Q\xrightarrow{\alpha}\alpha[m]+Q}
        \quad\frac{P\xrightarrow{\alpha}P' \quad \alpha\notin Q}{P+Q\xrightarrow{\alpha}P'+Q}$$
        $$\frac{Q\xrightarrow{\alpha}\alpha[m] \quad \alpha\notin P}{P+Q\xrightarrow{\alpha}P+\alpha[m]}
        \quad\frac{Q\xrightarrow{\alpha}Q'\quad \alpha\notin P}{P+Q\xrightarrow{\alpha}P+Q'}$$

        $$\frac{P\xrightarrow{\alpha}\alpha[m]\quad Q\xrightarrow{\alpha}\alpha[m]}{P+Q\xrightarrow{\alpha}\alpha[m]+\alpha[m]}
        \quad\frac{P\xrightarrow{\alpha}P'\quad Q\xrightarrow{\alpha}\alpha[m]}{P+Q\xrightarrow{\alpha}P'+\alpha[m]}$$
        $$\frac{P\xrightarrow{\alpha}\alpha[m]\quad Q\xrightarrow{\alpha}Q'}{P+Q\xrightarrow{\alpha}\alpha[m]+Q'}
        \quad\frac{P\xrightarrow{\alpha}P'\quad Q\xrightarrow{\alpha}Q'}{P+Q\xrightarrow{\alpha}P'+Q'}$$

        $$\frac{P\xrightarrow{\alpha}\alpha[m]\quad\textrm{Std}(Q)}{P. Q\xrightarrow{\alpha} \alpha[m]. Q} \quad\frac{P\xrightarrow{\alpha}P' \quad \textrm{Std}(Q)}{P. Q\xrightarrow{\alpha}P'. Q}$$
        $$\frac{Q\xrightarrow{\beta}\beta[n]\quad \textrm{NStd}(P)}{P. Q\xrightarrow{\beta}P. \beta[n]} \quad\frac{Q\xrightarrow{\beta}Q'\quad \textrm{NStd}(P)}{P. Q\xrightarrow{\beta}P. Q'}$$
        \caption{Forward transition rules of Prefix and Summation}
        \label{FTRForPS}
    \end{table}
\end{center}

The reverse transition rules for Prefix and Summation are shown in Table \ref{RTRForPS}.

\begin{center}
    \begin{table}
        $$\frac{}{\alpha[m]\xtworightarrow{\alpha[m]}\alpha}$$

        $$\frac{P\xtworightarrow{\alpha[m]}\alpha\quad \alpha\notin Q}{P+Q\xtworightarrow{\alpha[m]}\alpha+Q}
        \quad\frac{P\xtworightarrow{\alpha[m]}P' \quad \alpha\notin Q}{P+Q\xtworightarrow{\alpha[m]}P'+Q}$$
        $$\frac{Q\xtworightarrow{\alpha[m]}\alpha \quad \alpha\notin P}{P+Q\xtworightarrow{\alpha[m]}P+\alpha}
        \quad\frac{Q\xtworightarrow{\alpha[m]}Q' \quad \alpha\notin P}{P+Q\xtworightarrow{\alpha[m]}P+Q'}$$

        $$\frac{P\xtworightarrow{\alpha[m]}\alpha\quad Q\xtworightarrow{\alpha[m]}\alpha}{P+Q\xtworightarrow{\alpha[m]}\alpha+\alpha}
        \quad\frac{P\xtworightarrow{\alpha[m]}P'\quad Q\xtworightarrow{\alpha[m]}\alpha}{P+Q\xtworightarrow{\alpha[m]}P'+\alpha}$$
        $$\frac{P\xtworightarrow{\alpha[m]}\alpha\quad Q\xtworightarrow{\alpha[m]}Q'}{P+Q\xtworightarrow{\alpha[m]}\alpha+Q'}
        \quad\frac{P\xtworightarrow{\alpha[m]}P'\quad Q\xtworightarrow{\alpha[m]}Q'}{P+Q\xtworightarrow{\alpha[m]}P'+Q'}$$

        $$\frac{P\xtworightarrow{\alpha[m]}\alpha \quad \textrm{Std}(Q)}{P. Q\xtworightarrow{\alpha[m]} \alpha. Q} \quad\frac{P\xtworightarrow{\alpha[m]}P'\quad \textrm{Std}(Q)}{P. Q\xtworightarrow{\alpha[m]}P'. Q}$$
        $$\frac{Q\xtworightarrow{\beta[n]}\beta \quad \textrm{NStd}(P)}{P. Q\xtworightarrow{\beta[n]}P. \beta}\quad \frac{Q\xtworightarrow{\beta[n]}Q' \quad \textrm{NStd}(P)}{P. Q\xtworightarrow{\beta[n]}P. Q'}$$
        \caption{Reverse transition rules of Prefix and Summation}
        \label{RTRForPS}
    \end{table}
\end{center}

The forward and reverse pomset transition rules of Prefix and Summation are shown in Table \ref{FPTRForPS} and Table \ref{RPTRForPS}, different to single event transition rules in Table \ref{FTRForPS} and Table \ref{RTRForPS}, the forward and reverse pomset transition rules are labeled by pomsets, which are defined by causality $.$ and conflict $+$.

\begin{center}
    \begin{table}
        $$\frac{}{p\xrightarrow{p}p[\mathcal{K}]}$$

        $$\frac{P\xrightarrow{p}p[\mathcal{K}] \quad p\nsubseteq Q}{P+Q\xrightarrow{p}p[\mathcal{K}]+Q}
        \quad\frac{P\xrightarrow{p}P' \quad p\nsubseteq Q}{P+Q\xrightarrow{p}P'+Q}$$
        $$\frac{Q\xrightarrow{q}q[\mathcal{J}] \quad q\nsubseteq P}{P+Q\xrightarrow{q}P+q[\mathcal{J}]}
        \quad\frac{Q\xrightarrow{q}Q'\quad q\nsubseteq P}{P+Q\xrightarrow{q}P+Q'}$$

        $$\frac{P\xrightarrow{p}p[\mathcal{K}]\quad Q\xrightarrow{p}p[\mathcal{K}]}{P+Q\xrightarrow{p}p[\mathcal{K}]+p[\mathcal{K}]}
        \quad\frac{P\xrightarrow{p}P'\quad Q\xrightarrow{p}p[\mathcal{K}]}{P+Q\xrightarrow{p}P'+p[\mathcal{K}]}$$
        $$\frac{P\xrightarrow{p}p[\mathcal{K}]\quad Q\xrightarrow{p}Q'}{P+Q\xrightarrow{p}p[\mathcal{K}]+Q'}
        \quad\frac{P\xrightarrow{p}P'\quad Q\xrightarrow{p}Q'}{P+Q\xrightarrow{p}P'+Q'}$$

        $$\frac{P\xrightarrow{p}p[\mathcal{K}]\quad\textrm{Std}(Q)}{P. Q\xrightarrow{p} p[\mathcal{K}]. Q}(p\subseteq P) \quad\frac{P\xrightarrow{p}P' \quad \textrm{Std}(Q)}{P. Q\xrightarrow{p}P'. Q}(p\subseteq P)$$
        $$\frac{Q\xrightarrow{q}q[\mathcal{J}]\quad \textrm{NStd}(P)}{P. Q\xrightarrow{q}P. q[\mathcal{J}]}(q\subseteq Q) \quad\frac{Q\xrightarrow{q}Q'\quad \textrm{NStd}(P)}{P. Q\xrightarrow{q}P. Q'}(q\subseteq Q)$$
        \caption{Forward pomset transition rules of Prefix and Summation}
        \label{FPTRForPS}
    \end{table}
\end{center}

\begin{center}
    \begin{table}
        $$\frac{}{p[\mathcal{K}]\xtworightarrow{p[\mathcal{K}]}p}$$

        $$\frac{P\xtworightarrow{p[\mathcal{K}]}p\quad p\nsubseteq Q}{P+Q\xtworightarrow{p[\mathcal{K}]}p+Q}
        \quad\frac{P\xtworightarrow{p[\mathcal{K}]}P' \quad p\nsubseteq Q}{P+Q\xtworightarrow{p[\mathcal{K}]}P'+Q}$$
        $$\frac{Q\xtworightarrow{q[\mathcal{J}]}q \quad q\nsubseteq P}{P+Q\xtworightarrow{q[\mathcal{J}]}P+q}
        \quad\frac{Q\xtworightarrow{q[\mathcal{J}]}Q' \quad q\nsubseteq P}{P+Q\xtworightarrow{q[\mathcal{J}]}P+Q'}$$

        $$\frac{P\xtworightarrow{p[\mathcal{K}]}p\quad Q\xtworightarrow{p[\mathcal{K}]}p}{P+Q\xtworightarrow{p[\mathcal{K}]}p+p}
        \quad\frac{P\xtworightarrow{p[\mathcal{K}]}P'\quad Q\xtworightarrow{p[\mathcal{K}]}p}{P+Q\xtworightarrow{p[\mathcal{K}]}P'+p}$$
        $$\frac{P\xtworightarrow{p[\mathcal{K}]}p\quad Q\xtworightarrow{p[\mathcal{K}]}Q'}{P+Q\xtworightarrow{p[\mathcal{K}]}p+Q'}
        \quad\frac{P\xtworightarrow{p[\mathcal{K}]}P'\quad Q\xtworightarrow{p[\mathcal{K}]}Q'}{P+Q\xtworightarrow{p[\mathcal{K}]}P'+Q'}$$

        $$\frac{P\xtworightarrow{p[\mathcal{K}]}p \quad \textrm{Std}(Q)}{P. Q\xtworightarrow{p[\mathcal{K}]} p. Q}(p\subseteq P) \quad\frac{P\xtworightarrow{p[\mathcal{K}]}P'\quad \textrm{Std}(Q)}{P. Q\xtworightarrow{p[\mathcal{K}]}P'. Q}(p\subseteq P)$$
        $$\frac{Q\xtworightarrow{q[\mathcal{J}]}q \quad \textrm{NStd}(P)}{P. Q\xtworightarrow{q[\mathcal{J}]}P. q}(q\subseteq Q)\quad \frac{Q\xtworightarrow{q[\mathcal{J}]}Q' \quad \textrm{NStd}(P)}{P\cdot Q\xtworightarrow{q[\mathcal{J}]}P. Q'}(q\subseteq Q)$$
        \caption{Reverse pomset transition rules of Prefix and Summation}
        \label{RPTRForPS}
    \end{table}
\end{center}

The forward transition rules for Composition are shown in Table \ref{FTRForCom}.

\begin{center}
    \begin{table}
        $$\frac{P\xrightarrow{\alpha}P'\quad Q\nrightarrow}{P\parallel Q\xrightarrow{\alpha}P'\parallel Q}$$
        $$\frac{Q\xrightarrow{\alpha}Q'\quad P\nrightarrow}{P\parallel Q\xrightarrow{\alpha}P\parallel Q'}$$
        $$\frac{P\xrightarrow{\alpha}P'\quad Q\xrightarrow{\beta}Q'}{P\parallel Q\xrightarrow{\{\alpha,\beta\}}P'\parallel Q'}\quad (\beta\neq\overline{\alpha})$$
        $$\frac{P\xrightarrow{l}P'\quad Q\xrightarrow{\overline{l}}Q'}{P\parallel Q\xrightarrow{\tau}P'\parallel Q'}$$
        \caption{Forward transition rules of Composition}
        \label{FTRForCom}
    \end{table}
\end{center}

The reverse transition rules for Composition are shown in Table \ref{RTRForCom}.

\begin{center}
    \begin{table}
        $$\frac{P\xtworightarrow{\alpha[m]}P'\quad Q\xntworightarrow{}}{P\parallel Q\xtworightarrow{\alpha[m]}P'\parallel Q}$$
        $$\frac{Q\xtworightarrow{\alpha[m]}Q'\quad P\xntworightarrow{}}{P\parallel Q\xtworightarrow{\alpha[m]}P\parallel Q'}$$
        $$\frac{P\xtworightarrow{\alpha[m]}P'\quad Q\xtworightarrow{\beta[m]}Q'}{P\parallel Q\xtworightarrow{\{\alpha[m],\beta[m]\}}P'\parallel Q'}\quad (\beta\neq\overline{\alpha})$$
        $$\frac{P\xtworightarrow{l[m]}P'\quad Q\xtworightarrow{\overline{l}[m]}Q'}{P\parallel Q\xtworightarrow{\tau}P'\parallel Q'}$$
        \caption{Reverse transition rules of Composition}
        \label{RTRForCom}
    \end{table}
\end{center}

The forward transition rules for Restriction, Relabelling and Constants are shown in Table \ref{FTRForRRC}.

\begin{center}
    \begin{table}
        $$\frac{P\xrightarrow{\alpha}P'}{P\setminus L\xrightarrow{\alpha}P'\setminus L}\quad (\alpha,\overline{\alpha}\notin L)$$
        $$\frac{P\xrightarrow{\{\alpha_1,\cdots,\alpha_n\}}P'}{P\setminus L\xrightarrow{\{\alpha_1,\cdots,\alpha_n\}}P'\setminus L}\quad (\alpha_1,\overline{\alpha_1},\cdots,\alpha_n,\overline{\alpha_n}\notin L)$$
        $$\frac{P\xrightarrow{\alpha}P'}{P[f]\xrightarrow{f(\alpha)}P'[f]}$$
        $$\frac{P\xrightarrow{\{\alpha_1,\cdots,\alpha_n\}}P'}{P[f]\xrightarrow{\{f(\alpha_1),\cdots,f(\alpha_n)\}}P'[f]}$$
        $$\frac{P\xrightarrow{\alpha}P'}{A\xrightarrow{\alpha}P'}\quad (A\overset{\text{def}}{=}P)$$
        $$\frac{P\xrightarrow{\{\alpha_1,\cdots,\alpha_n\}}P'}{A\xrightarrow{\{\alpha_1,\cdots,\alpha_n\}}P'}\quad (A\overset{\text{def}}{=}P)$$
        \caption{Forward transition rules of Restriction, Relabelling and Constants}
        \label{FTRForRRC}
    \end{table}
\end{center}

The reverse transition rules for Restriction, Relabelling and Constants are shown in Table \ref{RTRForRRC}.

\begin{center}
    \begin{table}
        $$\frac{P\xtworightarrow{\alpha[m]}P'}{P\setminus L\xtworightarrow{\alpha[m]}P'\setminus L}\quad (\alpha,\overline{\alpha}\notin L)$$
        $$\frac{P\xtworightarrow{\{\alpha_1[m],\cdots,\alpha_n[m]\}}P'}{P\setminus L\xtworightarrow{\{\alpha_1[m],\cdots,\alpha_n[m]\}}P'\setminus L}\quad (\alpha_1,\overline{\alpha_1},\cdots,\alpha_n,\overline{\alpha_n}\notin L)$$
        $$\frac{P\xtworightarrow{\alpha[m]}P'}{P[f]\xtworightarrow{f(\alpha[m])}P'[f]}$$
        $$\frac{P\xtworightarrow{\{\alpha_1[m],\cdots,\alpha_n[m]\}}P'}{P[f]\xtworightarrow{\{f(\alpha_1)[m],\cdots,f(\alpha_n)[m]\}}P'[f]}$$
        $$\frac{P\xtworightarrow{\alpha[m]}P'}{A\xtworightarrow{\alpha[m]}P'}\quad (A\overset{\text{def}}{=}P)$$
        $$\frac{P\xtworightarrow{\{\alpha_1[m],\cdots,\alpha_n[m]\}}P'}{A\xtworightarrow{\{\alpha_1[m],\cdots,\alpha_n[m]\}}P'}\quad (A\overset{\text{def}}{=}P)$$
        \caption{Reverse transition rules of Restriction, Relabelling and Constants}
        \label{RTRForRRC}
    \end{table}
\end{center}

\subsection{Properties of Transitions}

\begin{definition}[Sorts]\label{sorts}
Given the sorts $\mathcal{L}(A)$ and $\mathcal{L}(X)$ of constants and variables, we define $\mathcal{L}(P)$ inductively as follows.

\begin{enumerate}
  \item $\mathcal{L}(l.P)=\{l\}\cup\mathcal{L}(P)$;
  \item $\mathcal{L}(P.l[m])=\{l\}\cup\mathcal{L}(P)$;
  \item $\mathcal{L}((l_1\parallel \cdots\parallel l_n).P)=\{l_1,\cdots,l_n\}\cup\mathcal{L}(P)$;
  \item $\mathcal{L}(P.(l_1[m]\parallel \cdots\parallel l_n[m]))=\{l_1,\cdots,l_n\}\cup\mathcal{L}(P)$;
  \item $\mathcal{L}(\tau.P)=\mathcal{L}(P)$;
  \item $\mathcal{L}(P+Q)=\mathcal{L}(P)\cup\mathcal{L}(Q)$;
  \item $\mathcal{L}(P\parallel Q)=\mathcal{L}(P)\cup\mathcal{L}(Q)$;
  \item $\mathcal{L}(P\setminus L)=\mathcal{L}(P)-(L\cup\overline{L})$;
  \item $\mathcal{L}(P[f])=\{f(l):l\in\mathcal{L}(P)\}$;
  \item for $A\overset{\text{def}}{=}P$, $\mathcal{L}(P)\subseteq\mathcal{L}(A)$.
\end{enumerate}
\end{definition}

Now, we present some properties of the transition rules defined in Definition \ref{semantics}.

\begin{proposition}
If $P\xrightarrow{\alpha}P'$, then
\begin{enumerate}
  \item $\alpha\in\mathcal{L}(P)\cup\{\tau\}$;
  \item $\mathcal{L}(P')\subseteq\mathcal{L}(P)$.
\end{enumerate}

If $P\xrightarrow{\{\alpha_1,\cdots,\alpha_n\}}P'$, then
\begin{enumerate}
  \item $\alpha_1,\cdots,\alpha_n\in\mathcal{L}(P)\cup\{\tau\}$;
  \item $\mathcal{L}(P')\subseteq\mathcal{L}(P)$.
\end{enumerate}
\end{proposition}

\begin{proof}
By induction on the inference of $P\xrightarrow{\alpha}P'$ and $P\xrightarrow{\{\alpha_1,\cdots,\alpha_n\}}P'$, there are several cases corresponding to the forward transition rules in Definition \ref{semantics}, we omit them.
\end{proof}

\begin{proposition}
If $P\xtworightarrow{\alpha[m]}P'$, then
\begin{enumerate}
  \item $\alpha\in\mathcal{L}(P)\cup\{\tau\}$;
  \item $\mathcal{L}(P')\subseteq\mathcal{L}(P)$.
\end{enumerate}

If $P\xtworightarrow{\{\alpha_1[m],\cdots,\alpha_n[m]\}}P'$, then
\begin{enumerate}
  \item $\alpha_1,\cdots,\alpha_n\in\mathcal{L}(P)\cup\{\tau\}$;
  \item $\mathcal{L}(P')\subseteq\mathcal{L}(P)$.
\end{enumerate}
\end{proposition}

\begin{proof}
By induction on the inference of $P\xtworightarrow{\alpha}P'$ and $P\xtworightarrow{\{\alpha_1,\cdots,\alpha_n\}}P'$, there are several cases corresponding to the forward transition rules in Definition \ref{semantics}, we omit them.
\end{proof}

%% file: section5.tex
\section{Strongly Forward-reverse Truly Concurrent Bisimulations}\label{sftcb}

Based on the concepts of strongly FR truly concurrent bisimulation equivalences, we get the following laws.

\begin{proposition}[Monoid laws for strongly FR pomset bisimulation] The monoid laws for strongly FR pomset bisimulation are as follows.

\begin{enumerate}
  \item $P+Q\sim_p^{fr} Q+P$;
  \item $P+(Q+R)\sim_p^{fr} (P+Q)+R$;
  \item $P+P\sim_p^{fr} P$;
  \item $P+\textbf{nil}\sim_p^{fr} P$.
\end{enumerate}

\end{proposition}

\begin{proof}
\begin{enumerate}
  \item $P+Q\sim_p^{fr} Q+P$. There are several cases, we will not enumerate all. By the forward transition rules of Summation in Table \ref{FPTRForPS}, we get

      $$\frac{P\xrightarrow{p}P'\quad p\nsubseteq Q}{P+ Q\xrightarrow{p}P'+Q} (p\subseteq P) \quad \frac{P\xrightarrow{p}P'\quad p\nsubseteq Q}{Q+ P\xrightarrow{p}Q+ P'}(p\subseteq P)$$

      $$\frac{Q\xrightarrow{q}Q'\quad q\nsupseteq P}{P+ Q\xrightarrow{q}P+Q'}(q\subseteq Q) \quad \frac{Q\xrightarrow{q}Q'\quad q\nsubseteq P}{Q+ P\xrightarrow{q}Q'+P }(q\subseteq Q)$$

      By the reverse transition rules of Summation in Table \ref{RPTRForPS}, we get

      $$\frac{P\xtworightarrow{p[\mathcal{K}]}P'\quad p\nsubseteq Q}{P+ Q\xtworightarrow{p[\mathcal{K}]}P'+Q} (p\subseteq P) \quad \frac{P\xtworightarrow{p[\mathcal{K}]}P'\quad p\nsubseteq Q}{Q+ P\xtworightarrow{p[\mathcal{K}]}Q+ P'}(p\subseteq P)$$

      $$\frac{Q\xtworightarrow{q[\mathcal{J}]}Q'\quad q\nsubseteq P}{P+ Q\xtworightarrow{q[\mathcal{J}]}P+Q'}(q\subseteq Q) \quad \frac{Q\xtworightarrow{q[\mathcal{J}]}Q'\quad q\nsubseteq P}{Q+ P\xtworightarrow{q[\mathcal{J}]}Q'+P }(q\subseteq Q)$$

      With the assumptions $P'+Q\sim_p^{fr} Q+P'$ and $P+Q'\sim_p^{fr} Q'+P$, so $P+ Q\sim_p^{fr} Q+ P$, as desired.
  \item $P+(Q+R)\sim_p^{fr} (P+Q)+R$. There are several cases, we will not enumerate all. By the forward transition rules of Summation in Table \ref{FPTRForPS}, we get

      $$\frac{P\xrightarrow{p}P'\quad p\nsubseteq Q\quad p\nsubseteq R}{P+(Q+R)\xrightarrow{p}P'+(Q+R)}(p\subseteq P) \quad \frac{P\xrightarrow{p}P'\quad p\nsubseteq Q\quad p\nsubseteq R}{(P+Q)+R\xrightarrow{p}(P'+Q)+R}(p\subseteq P)$$

      $$\frac{Q\xrightarrow{q}Q'\quad q\nsubseteq P\quad q\nsubseteq R}{P+(Q+R)\xrightarrow{q}P+(Q'+R)}(q\subseteq Q) \quad \frac{Q\xrightarrow{q}Q'\quad q\nsubseteq P\quad q\nsubseteq R}{(P+Q)+R\xrightarrow{q}(P+Q')+R}(q\subseteq Q)$$

      $$\frac{R\xrightarrow{r}R'\quad r\nsubseteq P\quad r\nsubseteq Q}{P+(Q+R)\xrightarrow{r}P+(Q+R')}(r\subseteq R) \quad \frac{R\xrightarrow{r}R'\quad r\nsubseteq P\quad r\nsubseteq Q}{(P+Q)+R\xrightarrow{r}(P+Q)+R'}(r\subseteq R)$$

      By the reverse transition rules of Summation in Table \ref{RPTRForPS}, we get

      $$\frac{P\xtworightarrow{p[\mathcal{K}]}P'\quad p\nsubseteq Q\quad p\nsubseteq R}{P+(Q+R)\xtworightarrow{p[\mathcal{K}]}P'+(Q+R)}(p\subseteq P) \quad \frac{P\xtworightarrow{p[\mathcal{K}]}P'\quad p\nsubseteq Q\quad p\nsubseteq R}{(P+Q)+R\xtworightarrow{p[\mathcal{K}]}(P'+Q)+R}(p\subseteq P)$$

      $$\frac{Q\xtworightarrow{q[\mathcal{J}]}Q'\quad q\nsubseteq P\quad q\nsubseteq R}{P+(Q+R)\xtworightarrow{q[\mathcal{J}]}P+(Q'+R)}(q\subseteq Q) \quad \frac{Q\xtworightarrow{q[\mathcal{J}]}Q'\quad q\nsubseteq P\quad q\nsubseteq R}{(P+Q)+R\xtworightarrow{q[\mathcal{J}]}(P+Q')+R}(q\subseteq Q)$$

      $$\frac{R\xtworightarrow{r[\mathcal{I}]}R'\quad r\nsubseteq P\quad r\nsubseteq Q}{P+(Q+R)\xtworightarrow{r[\mathcal{I}]}P+(Q+R')}(r\subseteq R) \quad \frac{R\xtworightarrow{r[\mathcal{I}]}R'\quad r\nsubseteq P\quad r\nsubseteq Q}{(P+Q)+R\xtworightarrow{r[\mathcal{I}]}(P+Q)+R'}(r\subseteq R)$$

      With the assumptions $P'+(Q+R)\sim_p^{fr} (P'+Q)+R$, $P+(Q'+R)\sim_p^{fr} (P+Q')+R$ and $P+(Q+R')\sim_p^{fr} (P+Q)+R'$, so $P+(Q+R)\sim_p^{fr} (P+Q)+R$, as desired.
  \item $P+P\sim_p^{fr} P$. By the forward transition rules of Summation, we get

      $$\frac{P\xrightarrow{p}p[\mathcal{K}]}{P+ P\xrightarrow{p}p[\mathcal{K}]+p[\mathcal{K}]}(p\subseteq P) \quad \frac{P\xrightarrow{p}p[\mathcal{K}]}{P\xrightarrow{p}p[\mathcal{K}]}(p\subseteq P)$$

      $$\frac{P\xrightarrow{p}P'}{P+ P\xrightarrow{p}P'+P'}(p\subseteq P) \quad \frac{P\xrightarrow{p}P'}{P\xrightarrow{p}P'}(p\subseteq P)$$

      By the reverse transition rules of Summation, we get

      $$\frac{P\xtworightarrow{p[\mathcal{K}]}p}{P+ P\xtworightarrow{p[\mathcal{K}]}p+p}(p\subseteq P) \quad \frac{P\xtworightarrow{p[\mathcal{K}]}p}{P\xtworightarrow{p[\mathcal{K}]}p}(p\subseteq P)$$

      $$\frac{P\xtworightarrow{p[\mathcal{K}]}P'}{P+ P\xtworightarrow{p[\mathcal{K}]}P'+P'}(p\subseteq P) \quad \frac{P\xtworightarrow{p[\mathcal{K}]}P'}{P\xtworightarrow{p[\mathcal{K}]}P'}(p\subseteq P)$$

      With the assumptions $p[\mathcal{K}]+p[\mathcal{K}]\sim_p^{fr} p[\mathcal{K}]$, $p+p\sim_p^{fr} p$ and $P'+P'\sim_p^{fr} P'$,so $P+ P\sim_p^{fr} P$, as desired.
  \item $P+\textbf{nil}\sim_p^{fr} P$. There are several cases, we will not enumerate all. By the forward transition rules of Summation in Table \ref{FPTRForPS}, we get

      $$\frac{P\xrightarrow{p}P'}{P+ \textbf{nil}\xrightarrow{p}P'}(p\subseteq P) \quad \frac{P\xrightarrow{p}P'}{P\xrightarrow{p}P'}(p\subseteq P)$$

      By the reverse transition rules of Summation in Table \ref{RPTRForPS}, we get

      $$\frac{P\xtworightarrow{p[\mathcal{K}]}P'}{P+ \textbf{nil}\xtworightarrow{p[\mathcal{K}]}P'}(p\subseteq P) \quad \frac{P\xtworightarrow{p[\mathcal{K}]}P'}{P\xtworightarrow{p[\mathcal{K}]}P'}(p\subseteq P)$$

      Since $P'\sim_p^{fr} P'$, $P+ \textbf{nil}\sim_p^{fr} P$, as desired.
\end{enumerate}
\end{proof}

\begin{proposition}[Monoid laws for strongly FR step bisimulation] The monoid laws for strongly FR step bisimulation are as follows.
\begin{enumerate}
  \item $P+Q\sim_s^{fr} Q+P$;
  \item $P+(Q+R)\sim_s^{fr} (P+Q)+R$;
  \item $P+P\sim_s^{fr} P$;
  \item $P+\textbf{nil}\sim_s^{fr} P$.
\end{enumerate}
\end{proposition}

\begin{proof}
\begin{enumerate}
  \item $P+Q\sim_s^{fr} Q+P$. There are several cases, we will not enumerate all. By the forward transition rules of Summation in Table \ref{FPTRForPS}, we get

      $$\frac{P\xrightarrow{p}P'\quad p\nsubseteq Q}{P+ Q\xrightarrow{p}P'+Q} (p\subseteq P,\forall\alpha,\beta \in p,\textrm{ are pairwise concurrent})$$
      $$\frac{P\xrightarrow{p}P'\quad p\nsubseteq Q}{Q+ P\xrightarrow{p}Q+ P'}(p\subseteq P,\forall\alpha,\beta \in p,\textrm{ are pairwise concurrent})$$

      $$\frac{Q\xrightarrow{q}Q'\quad q\nsupseteq P}{P+ Q\xrightarrow{q}P+Q'}(q\subseteq Q,\forall\alpha,\beta \in q,\textrm{ are pairwise concurrent})$$
      $$\frac{Q\xrightarrow{q}Q'\quad q\nsubseteq P}{Q+ P\xrightarrow{q}Q'+P }(q\subseteq Q,\forall\alpha,\beta \in q,\textrm{ are pairwise concurrent})$$

      By the reverse transition rules of Summation in Table \ref{RPTRForPS}, we get

      $$\frac{P\xtworightarrow{p[\mathcal{K}]}P'\quad p\nsubseteq Q}{P+ Q\xtworightarrow{p[\mathcal{K}]}P'+Q} (p\subseteq P,\forall\alpha,\beta \in p,\textrm{ are pairwise concurrent})$$
      $$\frac{P\xtworightarrow{p[\mathcal{K}]}P'\quad p\nsubseteq Q}{Q+ P\xtworightarrow{p[\mathcal{K}]}Q+ P'}(p\subseteq P,\forall\alpha,\beta \in p,\textrm{ are pairwise concurrent})$$

      $$\frac{Q\xtworightarrow{q[\mathcal{J}]}Q'\quad q\nsubseteq P}{P+ Q\xtworightarrow{q[\mathcal{J}]}P+Q'}(q\subseteq Q,\forall\alpha,\beta \in q,\textrm{ are pairwise concurrent})$$
      $$\frac{Q\xtworightarrow{q[\mathcal{J}]}Q'\quad q\nsubseteq P}{Q+ P\xtworightarrow{q[\mathcal{J}]}Q'+P }(q\subseteq Q,\forall\alpha,\beta \in q,\textrm{ are pairwise concurrent})$$

      With the assumptions $P'+Q\sim_s^{fr} Q+P'$ and $P+Q'\sim_s^{fr} Q'+P$, so $P+ Q\sim_s^{fr} Q+ P$, as desired.
  \item $P+(Q+R)\sim_s^{fr} (P+Q)+R$. There are several cases, we will not enumerate all. By the forward transition rules of Summation in Table \ref{FPTRForPS}, we get

      $$\frac{P\xrightarrow{p}P'\quad p\nsubseteq Q\quad p\nsubseteq R}{P+(Q+R)\xrightarrow{p}P'+(Q+R)}(p\subseteq P,\forall\alpha,\beta \in p,\textrm{ are pairwise concurrent})$$
      $$\frac{P\xrightarrow{p}P'\quad p\nsubseteq Q\quad p\nsubseteq R}{(P+Q)+R\xrightarrow{p}(P'+Q)+R}(p\subseteq P,\forall\alpha,\beta \in p,\textrm{ are pairwise concurrent})$$

      $$\frac{Q\xrightarrow{q}Q'\quad q\nsubseteq P\quad q\nsubseteq R}{P+(Q+R)\xrightarrow{q}P+(Q'+R)}(q\subseteq Q,\forall\alpha,\beta \in q,\textrm{ are pairwise concurrent})$$
      $$\frac{Q\xrightarrow{q}Q'\quad q\nsubseteq P\quad q\nsubseteq R}{(P+Q)+R\xrightarrow{q}(P+Q')+R}(q\subseteq Q,\forall\alpha,\beta \in q,\textrm{ are pairwise concurrent})$$

      $$\frac{R\xrightarrow{r}R'\quad r\nsubseteq P\quad r\nsubseteq Q}{P+(Q+R)\xrightarrow{r}P+(Q+R')}(r\subseteq R,\forall\alpha,\beta \in r,\textrm{ are pairwise concurrent})$$
      $$\frac{R\xrightarrow{r}R'\quad r\nsubseteq P\quad r\nsubseteq Q}{(P+Q)+R\xrightarrow{r}(P+Q)+R'}(r\subseteq R,\forall\alpha,\beta \in r,\textrm{ are pairwise concurrent})$$

      By the reverse transition rules of Summation in Table \ref{RPTRForPS}, we get

      $$\frac{P\xtworightarrow{p[\mathcal{K}]}P'\quad p\nsubseteq Q\quad p\nsubseteq R}{P+(Q+R)\xtworightarrow{p[\mathcal{K}]}P'+(Q+R)}(p\subseteq P,\forall\alpha,\beta \in p,\textrm{ are pairwise concurrent})$$
      $$\frac{P\xtworightarrow{p[\mathcal{K}]}P'\quad p\nsubseteq Q\quad p\nsubseteq R}{(P+Q)+R\xtworightarrow{p[\mathcal{K}]}(P'+Q)+R}(p\subseteq P,\forall\alpha,\beta \in p,\textrm{ are pairwise concurrent})$$

      $$\frac{Q\xtworightarrow{q[\mathcal{J}]}Q'\quad q\nsubseteq P\quad q\nsubseteq R}{P+(Q+R)\xtworightarrow{q[\mathcal{J}]}P+(Q'+R)}(q\subseteq Q,\forall\alpha,\beta \in q,\textrm{ are pairwise concurrent})$$
      $$\frac{Q\xtworightarrow{q[\mathcal{J}]}Q'\quad q\nsubseteq P\quad q\nsubseteq R}{(P+Q)+R\xtworightarrow{q[\mathcal{J}]}(P+Q')+R}(q\subseteq Q,\forall\alpha,\beta \in q,\textrm{ are pairwise concurrent})$$

      $$\frac{R\xtworightarrow{r[\mathcal{I}]}R'\quad r\nsubseteq P\quad r\nsubseteq Q}{P+(Q+R)\xtworightarrow{r[\mathcal{I}]}P+(Q+R')}(r\subseteq R,\forall\alpha,\beta \in r,\textrm{ are pairwise concurrent})$$
      $$\frac{R\xtworightarrow{r[\mathcal{I}]}R'\quad r\nsubseteq P\quad r\nsubseteq Q}{(P+Q)+R\xtworightarrow{r[\mathcal{I}]}(P+Q)+R'}(r\subseteq R,\forall\alpha,\beta \in r,\textrm{ are pairwise concurrent})$$

      With the assumptions $P'+(Q+R)\sim_s^{fr} (P'+Q)+R$, $P+(Q'+R)\sim_s^{fr} (P+Q')+R$ and $P+(Q+R')\sim_s^{fr} (P+Q)+R'$, so $P+(Q+R)\sim_s^{fr} (P+Q)+R$, as desired.
  \item $P+P\sim_s^{fr} P$. By the forward transition rules of Summation, we get

      $$\frac{P\xrightarrow{p}p[\mathcal{K}]}{P+ P\xrightarrow{p}p[\mathcal{K}]+p[\mathcal{K}]}(p\subseteq P,\forall\alpha,\beta \in p,\textrm{ are pairwise concurrent})$$ $$\frac{P\xrightarrow{p}p[\mathcal{K}]}{P\xrightarrow{p}p[\mathcal{K}]}(p\subseteq P,\forall\alpha,\beta \in p,\textrm{ are pairwise concurrent})$$

      $$\frac{P\xrightarrow{p}P'}{P+ P\xrightarrow{p}P'+P'}(p\subseteq P,\forall\alpha,\beta \in p,\textrm{ are pairwise concurrent})$$
      $$\frac{P\xrightarrow{p}P'}{P\xrightarrow{p}P'}(p\subseteq P,\forall\alpha,\beta \in p,\textrm{ are pairwise concurrent})$$

      By the reverse transition rules of Summation, we get

      $$\frac{P\xtworightarrow{p[\mathcal{K}]}p}{P+ P\xtworightarrow{p[\mathcal{K}]}p+p}(p\subseteq P,\forall\alpha,\beta \in p,\textrm{ are pairwise concurrent})$$ $$\frac{P\xtworightarrow{p[\mathcal{K}]}p}{P\xtworightarrow{p[\mathcal{K}]}p}(p\subseteq P,\forall\alpha,\beta \in p,\textrm{ are pairwise concurrent})$$

      $$\frac{P\xtworightarrow{p[\mathcal{K}]}P'}{P+ P\xtworightarrow{p[\mathcal{K}]}P'+P'}(p\subseteq P,\forall\alpha,\beta \in p,\textrm{ are pairwise concurrent})$$ $$\frac{P\xtworightarrow{p[\mathcal{K}]}P'}{P\xtworightarrow{p[\mathcal{K}]}P'}(p\subseteq P,\forall\alpha,\beta \in p,\textrm{ are pairwise concurrent})$$

      With the assumptions $p[\mathcal{K}]+p[\mathcal{K}]\sim_s^{fr} p[\mathcal{K}]$, $p+p\sim_s^{fr} p$ and $P'+P'\sim_s^{fr} P'$,so $P+ P\sim_s^{fr} P$, as desired.
  \item $P+\textbf{nil}\sim_s^{fr} P$. There are several cases, we will not enumerate all. By the forward transition rules of Summation in Table \ref{FPTRForPS}, we get

      $$\frac{P\xrightarrow{p}P'}{P+ \textbf{nil}\xrightarrow{p}P'}(p\subseteq P,\forall\alpha,\beta \in p,\textrm{ are pairwise concurrent})$$
      $$\frac{P\xrightarrow{p}P'}{P\xrightarrow{p}P'}(p\subseteq P,\forall\alpha,\beta \in p,\textrm{ are pairwise concurrent})$$

      By the reverse transition rules of Summation in Table \ref{RPTRForPS}, we get

      $$\frac{P\xtworightarrow{p[\mathcal{K}]}P'}{P+ \textbf{nil}\xtworightarrow{p[\mathcal{K}]}P'}(p\subseteq P,\forall\alpha,\beta \in p,\textrm{ are pairwise concurrent})$$ $$\frac{P\xtworightarrow{p[\mathcal{K}]}P'}{P\xtworightarrow{p[\mathcal{K}]}P'}(p\subseteq P,\forall\alpha,\beta \in p,\textrm{ are pairwise concurrent})$$

      Since $P'\sim_s^{fr} P'$, $P+ \textbf{nil}\sim_s^{fr} P$, as desired.
\end{enumerate}
\end{proof}

\begin{proposition}[Monoid laws for strongly FR hp-bisimulation] The monoid laws for strongly FR hp-bisimulation are as follows.
\begin{enumerate}
  \item $P+Q\sim_{hp}^{fr} Q+P$;
  \item $P+(Q+R)\sim_{hp}^{fr} (P+Q)+R$;
  \item $P+P\sim_{hp}^{fr} P$;
  \item $P+\textbf{nil}\sim_{hp}^{fr} P$.
\end{enumerate}
\end{proposition}

\begin{proof}
\begin{enumerate}
  \item $P+Q\sim_{hp}^{fr} Q+P$. There are several cases, we will not enumerate all. By the forward transition rules of Summation in Table \ref{FTRForPS}, we get

      $$\frac{P\xrightarrow{\alpha}P'\quad \alpha\notin Q}{P+ Q\xrightarrow{\alpha}P'+Q} (\alpha\in P) \quad \frac{P\xrightarrow{\alpha}P'\quad \alpha\notin Q}{Q+ P\xrightarrow{\alpha}Q+ P'}(\alpha\in P)$$

      $$\frac{Q\xrightarrow{\beta}Q'\quad\beta\notin P}{P+ Q\xrightarrow{\beta}P+Q'}(\beta\in Q) \quad \frac{Q\xrightarrow{\beta}Q'\quad\beta\notin P}{Q+ P\xrightarrow{\beta}Q'+P }(\beta\in Q)$$

      By the reverse transition rules of Summation in Table \ref{RTRForPS}, we get

      $$\frac{P\xtworightarrow{\alpha[m]}P'\quad\alpha\notin Q}{P+ Q\xtworightarrow{\alpha[m]}P'+Q} (\alpha\in P) \quad \frac{P\xtworightarrow{\alpha[m]}P'\quad\alpha\notin Q}{Q+ P\xtworightarrow{\alpha[m]}Q+ P'}(\alpha P)$$

      $$\frac{Q\xtworightarrow{\beta[n]}Q'\quad\beta\notin P}{P+ Q\xtworightarrow{\beta[n]}P+Q'}(\beta\in Q) \quad \frac{Q\xtworightarrow{\beta[n]}Q'\quad\beta\notin P}{Q+ P\xtworightarrow{\beta[n]}Q'+P }(\beta\in Q)$$

      Since $(C(P+ Q),f,C(Q+ P))\in\sim_{hp}^{fr}$, $(C((P+ Q)'),f[\alpha\mapsto \alpha],C((Q+ P)'))\in\sim_{hp}^{fr}$ and $(C((P+ Q)'),f[\beta\mapsto \beta],C((Q+ P)'))\in\sim_{hp}^{fr}$, $P+ Q\sim_{hp}^{fr} Q+ P$, as desired.
  \item $P+(Q+R)\sim_{hp}^{fr} (P+Q)+R$. There are several cases, we will not enumerate all. By the forward transition rules of Summation in Table \ref{FTRForPS}, we get

      $$\frac{P\xrightarrow{\alpha}P'\quad\alpha\notin Q\quad\alpha\notin R}{P+(Q+R)\xrightarrow{\alpha}P'+(Q+R)}(\alpha\in P) \quad \frac{P\xrightarrow{\alpha}P'\quad\alpha\notin Q\quad\alpha\notin R}{(P+Q)+R\xrightarrow{\alpha}(P'+Q)+R}(\alpha\in P)$$

      $$\frac{Q\xrightarrow{\beta}Q'\quad\beta\notin P\quad\beta\notin R}{P+(Q+R)\xrightarrow{\beta}P+(Q'+R)}(\beta\in Q) \quad \frac{Q\xrightarrow{\beta}Q'\quad\beta\notin P\quad\beta\notin R}{(P+Q)+R\xrightarrow{\beta}(P+Q')+R}(\beta\in Q)$$

      $$\frac{R\xrightarrow{\gamma}R'\quad\gamma\notin P\quad\gamma\notin Q}{P+(Q+R)\xrightarrow{\gamma}P+(Q+R')}(\gamma\in R) \quad \frac{R\xrightarrow{\gamma}R'\quad\gamma\notin P\quad\gamma\notin Q}{(P+Q)+R\xrightarrow{\gamma}(P+Q)+R'}(\gamma\in R)$$

      By the reverse transition rules of Summation in Table \ref{RTRForPS}, we get

      $$\frac{P\xtworightarrow{\alpha[m]}P'\quad\alpha\notin Q\quad\alpha\notin R}{P+(Q+R)\xtworightarrow{\alpha[m]}P'+(Q+R)}(\alpha\in P) \quad \frac{P\xtworightarrow{\alpha[m]}P'\quad\alpha\notin Q\quad\alpha\notin R}{(P+Q)+R\xtworightarrow{\alpha[m]}(P'+Q)+R}(\alpha\in P)$$

      $$\frac{Q\xtworightarrow{\beta[n]}Q'\quad\beta\notin P\quad\beta\notin R}{P+(Q+R)\xtworightarrow{\beta[n]}P+(Q'+R)}(\beta\in Q) \quad \frac{Q\xtworightarrow{\beta[n]}Q'\quad\beta\notin P\quad\beta\notin R}{(P+Q)+R\xtworightarrow{\beta[n]}(P+Q')+R}(\beta Q)$$

      $$\frac{R\xtworightarrow{\gamma[k]}R'\quad\gamma\notin P\quad\gamma\notin Q}{P+(Q+R)\xtworightarrow{\gamma[k]}P+(Q+R')}(\gamma\in R) \quad \frac{R\xtworightarrow{\gamma[k]}R'\quad\gamma\notin P\quad\gamma\notin Q}{(P+Q)+R\xtworightarrow{\gamma[k]}(P+Q)+R'}(\gamma\in R)$$

      Since $(C(P+ (Q+R)),f,C((P+Q)+R))\in\sim_{hp}^{fr}$, $(C((P+ (Q+R))'),f[\alpha\mapsto \alpha],C((P+Q)+R)'))\in\sim_{hp}^{fr}$, $(C((P+ (Q+R))'),f[\beta\mapsto \beta],C((P+Q)+R)'))\in\sim_{hp}^{fr}$ and $(C((P+ (Q+R))'),f[\gamma\mapsto \gamma],C((P+Q)+R)'))\in\sim_{hp}^{fr}$, $P+(Q+R)\sim_{hp}^{fr} (P+Q)+R$, as desired.
  \item $P+P\sim_{hp}^{fr} P$. By the forward transition rules of Summation, we get

      $$\frac{P\xrightarrow{\alpha}\alpha[m]}{P+ P\xrightarrow{\alpha}\alpha[m]+\alpha[m]}(\alpha\in P) \quad \frac{P\xrightarrow{\alpha}\alpha[m]}{P\xrightarrow{\alpha}\alpha[m]}(\alpha\in P)$$

      $$\frac{P\xrightarrow{\alpha}P'}{P+ P\xrightarrow{\alpha}P'+P'}(p\subseteq P) \quad \frac{P\xrightarrow{\alpha}P'}{P\xrightarrow{\alpha}P'}(\alpha\in P)$$

      By the reverse transition rules of Summation, we get

      $$\frac{P\xtworightarrow{\alpha[m]}\alpha}{P+ P\xtworightarrow{\alpha[m]}\alpha+\alpha}(\alpha\in P) \quad \frac{P\xtworightarrow{\alpha[m]}\alpha}{P\xtworightarrow{\alpha[m]}\alpha}(\alpha\in P)$$

      $$\frac{P\xtworightarrow{\alpha[m]}P'}{P+ P\xtworightarrow{\alpha[m]}P'+P'}(\alpha\in P) \quad \frac{P\xtworightarrow{\alpha[m]}P'}{P\xtworightarrow{\alpha[m}P'}(\alpha\in P)$$

      Since $(C(P+P),f,C(P))\in\sim_{hp}^{fr}$, $(C((P+ P)'),f[\alpha\mapsto \alpha],C((P)'))\in\sim_{hp}^{fr}$, $P+ P\sim_{hp}^{fr} P$, as desired.
  \item $P+\textbf{nil}\sim_{hp}^{fr} P$. There are several cases, we will not enumerate all. By the forward transition rules of Summation in Table \ref{FTRForPS}, we get

      $$\frac{P\xrightarrow{\alpha}P'}{P+ \textbf{nil}\xrightarrow{\alpha}P'}(\alpha P) \quad \frac{P\xrightarrow{\alpha}P'}{P\xrightarrow{\alpha}P'}(\alpha\in P)$$

      By the reverse transition rules of Summation in Table \ref{RTRForPS}, we get

      $$\frac{P\xtworightarrow{\alpha[m]}P'}{P+ \textbf{nil}\xtworightarrow{\alpha[m]}P'}(\alpha\in P) \quad \frac{P\xtworightarrow{\alpha[m]}P'}{P\xtworightarrow{\alpha[m]}P'}(\alpha\in P)$$

      Since $(C(P+\textbf{nil}),f,C(P))\in\sim_{hp}^{fr}$, $(C((P+ \textbf{nil})'),f[\alpha\mapsto \alpha],C((P)'))\in\sim_{hp}^{fr}$, $P+ \textbf{nil}\sim_{hp}^{fr} P$, as desired.
\end{enumerate}
\end{proof}

\begin{proposition}[Monoid laws for strongly FR hhp-bisimulation] The monoid laws for strongly FR hhp-bisimulation are as follows.
\begin{enumerate}
  \item $P+Q\sim_{hhp}^{fr} Q+P$;
  \item $P+(Q+R)\sim_{hhp}^{fr} (P+Q)+R$;
  \item $P+P\sim_{hhp}^{fr} P$;
  \item $P+\textbf{nil}\sim_{hhp}^{fr} P$.
\end{enumerate}
\end{proposition}

\begin{proof}
\begin{enumerate}
  \item $P+Q\sim_{hhp}^{fr} Q+P$. There are several cases, we will not enumerate all. By the forward transition rules of Summation in Table \ref{FTRForPS}, we get

      $$\frac{P\xrightarrow{\alpha}P'\quad \alpha\notin Q}{P+ Q\xrightarrow{\alpha}P'+Q} (\alpha\in P) \quad \frac{P\xrightarrow{\alpha}P'\quad \alpha\notin Q}{Q+ P\xrightarrow{\alpha}Q+ P'}(\alpha\in P)$$

      $$\frac{Q\xrightarrow{\beta}Q'\quad\beta\notin P}{P+ Q\xrightarrow{\beta}P+Q'}(\beta\in Q) \quad \frac{Q\xrightarrow{\beta}Q'\quad\beta\notin P}{Q+ P\xrightarrow{\beta}Q'+P }(\beta\in Q)$$

      By the reverse transition rules of Summation in Table \ref{RTRForPS}, we get

      $$\frac{P\xtworightarrow{\alpha[m]}P'\quad\alpha\notin Q}{P+ Q\xtworightarrow{\alpha[m]}P'+Q} (\alpha\in P) \quad \frac{P\xtworightarrow{\alpha[m]}P'\quad\alpha\notin Q}{Q+ P\xtworightarrow{\alpha[m]}Q+ P'}(\alpha P)$$

      $$\frac{Q\xtworightarrow{\beta[n]}Q'\quad\beta\notin P}{P+ Q\xtworightarrow{\beta[n]}P+Q'}(\beta\in Q) \quad \frac{Q\xtworightarrow{\beta[n]}Q'\quad\beta\notin P}{Q+ P\xtworightarrow{\beta[n]}Q'+P }(\beta\in Q)$$

      Since $(C(P+ Q),f,C(Q+ P))\in\sim_{hhp}^{fr}$, $(C((P+ Q)'),f[\alpha\mapsto \alpha],C((Q+ P)'))\in\sim_{hhp}^{fr}$ and $(C((P+ Q)'),f[\beta\mapsto \beta],C((Q+ P)'))\in\sim_{hhp}^{fr}$, $P+ Q\sim_{hhp}^{fr} Q+ P$, as desired.
  \item $P+(Q+R)\sim_{hhp}^{fr} (P+Q)+R$. There are several cases, we will not enumerate all. By the forward transition rules of Summation in Table \ref{FTRForPS}, we get

      $$\frac{P\xrightarrow{\alpha}P'\quad\alpha\notin Q\quad\alpha\notin R}{P+(Q+R)\xrightarrow{\alpha}P'+(Q+R)}(\alpha\in P) \quad \frac{P\xrightarrow{\alpha}P'\quad\alpha\notin Q\quad\alpha\notin R}{(P+Q)+R\xrightarrow{\alpha}(P'+Q)+R}(\alpha\in P)$$

      $$\frac{Q\xrightarrow{\beta}Q'\quad\beta\notin P\quad\beta\notin R}{P+(Q+R)\xrightarrow{\beta}P+(Q'+R)}(\beta\in Q) \quad \frac{Q\xrightarrow{\beta}Q'\quad\beta\notin P\quad\beta\notin R}{(P+Q)+R\xrightarrow{\beta}(P+Q')+R}(\beta\in Q)$$

      $$\frac{R\xrightarrow{\gamma}R'\quad\gamma\notin P\quad\gamma\notin Q}{P+(Q+R)\xrightarrow{\gamma}P+(Q+R')}(\gamma\in R) \quad \frac{R\xrightarrow{\gamma}R'\quad\gamma\notin P\quad\gamma\notin Q}{(P+Q)+R\xrightarrow{\gamma}(P+Q)+R'}(\gamma\in R)$$

      By the reverse transition rules of Summation in Table \ref{RTRForPS}, we get

      $$\frac{P\xtworightarrow{\alpha[m]}P'\quad\alpha\notin Q\quad\alpha\notin R}{P+(Q+R)\xtworightarrow{\alpha[m]}P'+(Q+R)}(\alpha\in P) \quad \frac{P\xtworightarrow{\alpha[m]}P'\quad\alpha\notin Q\quad\alpha\notin R}{(P+Q)+R\xtworightarrow{\alpha[m]}(P'+Q)+R}(\alpha\in P)$$

      $$\frac{Q\xtworightarrow{\beta[n]}Q'\quad\beta\notin P\quad\beta\notin R}{P+(Q+R)\xtworightarrow{\beta[n]}P+(Q'+R)}(\beta\in Q) \quad \frac{Q\xtworightarrow{\beta[n]}Q'\quad\beta\notin P\quad\beta\notin R}{(P+Q)+R\xtworightarrow{\beta[n]}(P+Q')+R}(\beta Q)$$

      $$\frac{R\xtworightarrow{\gamma[k]}R'\quad\gamma\notin P\quad\gamma\notin Q}{P+(Q+R)\xtworightarrow{\gamma[k]}P+(Q+R')}(\gamma\in R) \quad \frac{R\xtworightarrow{\gamma[k]}R'\quad\gamma\notin P\quad\gamma\notin Q}{(P+Q)+R\xtworightarrow{\gamma[k]}(P+Q)+R'}(\gamma\in R)$$

      Since $(C(P+ (Q+R)),f,C((P+Q)+R))\in\sim_{hhp}^{fr}$, $(C((P+ (Q+R))'),f[\alpha\mapsto \alpha],C((P+Q)+R)'))\in\sim_{hhp}^{fr}$, $(C((P+ (Q+R))'),f[\beta\mapsto \beta],C((P+Q)+R)'))\in\sim_{hhp}^{fr}$ and $(C((P+ (Q+R))'),f[\gamma\mapsto \gamma],C((P+Q)+R)'))\in\sim_{hhp}^{fr}$, $P+(Q+R)\sim_{hhp}^{fr} (P+Q)+R$, as desired.
  \item $P+P\sim_{hhp}^{fr} P$. By the forward transition rules of Summation, we get

      $$\frac{P\xrightarrow{\alpha}\alpha[m]}{P+ P\xrightarrow{\alpha}\alpha[m]+\alpha[m]}(\alpha\in P) \quad \frac{P\xrightarrow{\alpha}\alpha[m]}{P\xrightarrow{\alpha}\alpha[m]}(\alpha\in P)$$

      $$\frac{P\xrightarrow{\alpha}P'}{P+ P\xrightarrow{\alpha}P'+P'}(p\subseteq P) \quad \frac{P\xrightarrow{\alpha}P'}{P\xrightarrow{\alpha}P'}(\alpha\in P)$$

      By the reverse transition rules of Summation, we get

      $$\frac{P\xtworightarrow{\alpha[m]}\alpha}{P+ P\xtworightarrow{\alpha[m]}\alpha+\alpha}(\alpha\in P) \quad \frac{P\xtworightarrow{\alpha[m]}\alpha}{P\xtworightarrow{\alpha[m]}\alpha}(\alpha\in P)$$

      $$\frac{P\xtworightarrow{\alpha[m]}P'}{P+ P\xtworightarrow{\alpha[m]}P'+P'}(\alpha\in P) \quad \frac{P\xtworightarrow{\alpha[m]}P'}{P\xtworightarrow{\alpha[m}P'}(\alpha\in P)$$

      Since $(C(P+P),f,C(P))\in\sim_{hhp}^{fr}$, $(C((P+ P)'),f[\alpha\mapsto \alpha],C((P)'))\in\sim_{hhp}^{fr}$, $P+ P\sim_{hhp}^{fr} P$, as desired.
  \item $P+\textbf{nil}\sim_{hhp}^{fr} P$. There are several cases, we will not enumerate all. By the forward transition rules of Summation in Table \ref{FTRForPS}, we get

      $$\frac{P\xrightarrow{\alpha}P'}{P+ \textbf{nil}\xrightarrow{\alpha}P'}(\alpha P) \quad \frac{P\xrightarrow{\alpha}P'}{P\xrightarrow{\alpha}P'}(\alpha\in P)$$

      By the reverse transition rules of Summation in Table \ref{RTRForPS}, we get

      $$\frac{P\xtworightarrow{\alpha[m]}P'}{P+ \textbf{nil}\xtworightarrow{\alpha[m]}P'}(\alpha\in P) \quad \frac{P\xtworightarrow{\alpha[m]}P'}{P\xtworightarrow{\alpha[m]}P'}(\alpha\in P)$$

      Since $(C(P+\textbf{nil}),f,C(P))\in\sim_{hhp}^{fr}$, $(C((P+ \textbf{nil})'),f[\alpha\mapsto \alpha],C((P)'))\in\sim_{hhp}^{fr}$, $P+ \textbf{nil}\sim_{hhp}^{fr} P$, as desired.
\end{enumerate}
\end{proof}

\begin{proposition}[Static laws for strongly FR step bisimulation] \label{SLSSB}
The static laws for strongly FR step bisimulation are as follows.
\begin{enumerate}
  \item $P\parallel Q\sim_s^{fr} Q\parallel P$;
  \item $P\parallel(Q\parallel R)\sim_s^{fr} (P\parallel Q)\parallel R$;
  \item $P\parallel \textbf{nil}\sim_s^{fr} P$;
  \item $P\setminus L\sim_s^{fr} P$, if $\mathcal{L}(P)\cap(L\cup\overline{L})=\emptyset$;
  \item $P\setminus K\setminus L\sim_s^{fr} P\setminus(K\cup L)$;
  \item $P[f]\setminus L\sim_s^{fr} P\setminus f^{-1}(L)[f]$;
  \item $(P\parallel Q)\setminus L\sim_s^{fr} P\setminus L\parallel Q\setminus L$, if $\mathcal{L}(P)\cap\overline{\mathcal{L}(Q)}\cap(L\cup\overline{L})=\emptyset$;
  \item $P[Id]\sim_s^{fr} P$;
  \item $P[f]\sim_s^{fr} P[f']$, if $f\upharpoonright\mathcal{L}(P)=f'\upharpoonright\mathcal{L}(P)$;
  \item $P[f][f']\sim_s^{fr} P[f'\circ f]$;
  \item $(P\parallel Q)[f]\sim_s^{fr} P[f]\parallel Q[f]$, if $f\upharpoonright(L\cup\overline{L})$ is one-to-one, where $L=\mathcal{L}(P)\cup\mathcal{L}(Q)$.
\end{enumerate}
\end{proposition}

\begin{proof}
Though transition rules in Table \ref{FTRForCom}, \ref{RTRForCom}, \ref{FTRForRRC}, \ref{RTRForRRC} are defined in the flavor of single event, they can be modified into a step (a set of events within which each event is pairwise concurrent), we omit them. If we treat a single event as a step containing just one event, the proof of the static laws does not exist any problem, so we use this way and still use the transition rules in Table \ref{FTRForCom}, \ref{RTRForCom}, \ref{FTRForRRC}, \ref{RTRForRRC}.

\begin{enumerate}
  \item $P\parallel Q\sim_s^{fr} Q\parallel P$. By the forward transition rules of Composition, we get

      $$\frac{P\xrightarrow{\alpha}P'\quad Q\nrightarrow}{P\parallel Q\xrightarrow{\alpha}P'\parallel Q}
      \quad\frac{P\xrightarrow{\alpha}P'\quad Q\nrightarrow}{Q\parallel P\xrightarrow{\alpha}Q\parallel P'}$$

      $$\frac{Q\xrightarrow{\beta}Q'\quad P\nrightarrow}{P\parallel Q\xrightarrow{\beta}P\parallel Q'}
      \quad\frac{Q\xrightarrow{\beta}Q'\quad P\nrightarrow}{Q\parallel P\xrightarrow{\beta}Q'\parallel P}$$

      $$\frac{P\xrightarrow{\alpha}P'\quad Q\xrightarrow{\beta}Q'}{P\parallel Q\xrightarrow{\{\alpha,\beta\}}P'\parallel Q'}(\beta\neq\overline{\alpha})
      \quad\frac{P\xrightarrow{\alpha}P'\quad Q\xrightarrow{\beta}Q'}{Q\parallel P\xrightarrow{\{\alpha,\beta\}}Q'\parallel P'}(\beta\neq\overline{\alpha})$$

      $$\frac{P\xrightarrow{l}P'\quad Q\xrightarrow{\overline{l}}Q'}{P\parallel Q\xrightarrow{\tau}P'\parallel Q'}
      \quad\frac{P\xrightarrow{l}P'\quad Q\xrightarrow{\overline{l}}Q'}{Q\parallel P\xrightarrow{\tau}Q'\parallel P'}$$

      By the reverse transition rules of Composition, we get

      $$\frac{P\xtworightarrow{\alpha[m]}P'\quad Q\xntworightarrow{}}{P\parallel Q\xtworightarrow{\alpha[m]}P'\parallel Q}
      \quad\frac{P\xtworightarrow{\alpha[m]}P'\quad Q\xntworightarrow{}}{Q\parallel P\xtworightarrow{\alpha[m]}Q\parallel P'}$$

      $$\frac{Q\xtworightarrow{\beta[n]}Q'\quad P\xntworightarrow{}}{P\parallel Q\xtworightarrow{\beta[n]}P\parallel Q'}
      \quad\frac{Q\xtworightarrow{\beta[n]}Q'\quad P\xntworightarrow{}}{Q\parallel P\xtworightarrow{\beta[n]}Q'\parallel P}$$

      $$\frac{P\xtworightarrow{\alpha[m]}P'\quad Q\xtworightarrow{\beta[m]}Q'}{P\parallel Q\xtworightarrow{\{\alpha[m],\beta[m]\}}P'\parallel Q'}(\beta\neq\overline{\alpha})
      \quad\frac{P\xtworightarrow{\alpha[m]}P'\quad Q\xtworightarrow{\beta[m]}Q'}{Q\parallel P\xrightarrow{\{\alpha[m],\beta[m]\}}Q'\parallel P'}(\beta\neq\overline{\alpha})$$

      $$\frac{P\xtworightarrow{l[m]}P'\quad Q\xtworightarrow{\overline{l}[m]}Q'}{P\parallel Q\xtworightarrow{\tau}P'\parallel Q'}
      \quad\frac{P\xtworightarrow{l[m]}P'\quad Q\xtworightarrow{\overline{l}[m]}Q'}{Q\parallel P\xtworightarrow{\tau}Q'\parallel P'}$$

      So, with the assumptions $P'\parallel Q \sim_s^{fr} Q\parallel P'$, $P\parallel Q' \sim_s^{fr} Q'\parallel P$ and $P'\parallel Q' \sim_s^{fr} Q'\parallel P'$, $P\parallel Q\sim_s^{fr} Q\parallel P$, as desired.
  \item $P\parallel(Q\parallel R)\sim_s^{fr} (P\parallel Q)\parallel R$. By the forward transition rules of Composition, we get

      $$\frac{P\xrightarrow{\alpha}P'\quad Q\nrightarrow\quad R\nrightarrow}{P\parallel (Q\parallel R)\xrightarrow{\alpha}P'\parallel (Q\parallel R)}
      \quad\frac{P\xrightarrow{\alpha}P'\quad Q\nrightarrow\quad R\nrightarrow}{(P\parallel Q)\parallel R\xrightarrow{\alpha}(P'\parallel Q)\parallel R}$$

      $$\frac{Q\xrightarrow{\beta}Q'\quad P\nrightarrow\quad R\nrightarrow}{P\parallel (Q\parallel R)\xrightarrow{\beta}P\parallel (Q'\parallel R)}
      \quad\frac{Q\xrightarrow{\beta}Q'\quad P\nrightarrow\quad R\nrightarrow}{(P\parallel Q)\parallel R\xrightarrow{\beta}(P\parallel Q')\parallel R}$$

      $$\frac{R\xrightarrow{\gamma}R'\quad P\nrightarrow\quad Q\nrightarrow}{P\parallel (Q\parallel R)\xrightarrow{\gamma}P\parallel (Q\parallel R')}
      \quad\frac{R\xrightarrow{\gamma}R'\quad P\nrightarrow\quad Q\nrightarrow}{(P\parallel Q)\parallel R\xrightarrow{\gamma}(P\parallel Q)\parallel R'}$$

      $$\frac{P\xrightarrow{\alpha}P'\quad Q\xrightarrow{\beta}Q'\quad R\nrightarrow}{P\parallel (Q\parallel R)\xrightarrow{\{\alpha,\beta\}}P'\parallel (Q'\parallel R)}(\beta\neq\overline{\alpha})
      \quad\frac{P\xrightarrow{\alpha}P'\quad Q\xrightarrow{\beta}Q'\quad R\nrightarrow}{(P\parallel Q)\parallel R\xrightarrow{\{\alpha,\beta\}}(P'\parallel Q')\parallel R}(\beta\neq\overline{\alpha})$$

      $$\frac{P\xrightarrow{\alpha}P'\quad R\xrightarrow{\gamma}R'\quad Q\nrightarrow}{P\parallel (Q\parallel R)\xrightarrow{\{\alpha,\gamma\}}P'\parallel (Q\parallel R')}(\gamma\neq\overline{\alpha})
      \quad\frac{P\xrightarrow{\alpha}P'\quad R\xrightarrow{\gamma}R'\quad Q\nrightarrow}{(P\parallel Q)\parallel R\xrightarrow{\{\alpha,\gamma\}}(P'\parallel Q)\parallel R]}(\gamma\neq\overline{\alpha})$$

      $$\frac{Q\xrightarrow{\beta}P'\quad R\xrightarrow{\gamma}R'\quad P\nrightarrow}{P\parallel (Q\parallel R)\xrightarrow{\{\beta,\gamma\}}P\parallel (Q'\parallel R')}(\gamma\neq\overline{\beta})
      \quad\frac{Q\xrightarrow{\beta}Q'\quad R\xrightarrow{\gamma}R'\quad P\nrightarrow}{(P\parallel Q)\parallel R\xrightarrow{\{\beta,\gamma\}}(P\parallel Q')\parallel R'}(\gamma\neq\overline{\beta})$$

      $$\frac{P\xrightarrow{\alpha}P'\quad Q\xrightarrow{\beta}Q'\quad R\xrightarrow{\gamma}R'}{P\parallel (Q\parallel R)\xrightarrow{\{\alpha,\beta,\gamma\}}P'\parallel (Q'\parallel R')}(\beta\neq\overline{\alpha},\gamma\neq\overline{\alpha},\gamma\neq\overline{\beta})
      \quad\frac{P\xrightarrow{\alpha}P'\quad Q\xrightarrow{\beta}Q'\quad R\xrightarrow{\gamma}R'}{(P\parallel Q)\parallel R\xrightarrow{\{\alpha,\beta,\gamma\}}(P'\parallel Q')\parallel R'}(\beta\neq\overline{\alpha},\gamma\neq\overline{\alpha},\gamma\neq\overline{\beta})$$

      $$\frac{P\xrightarrow{l}P'\quad Q\xrightarrow{\overline{l}}Q'\quad R\nrightarrow}{P\parallel (Q\parallel R)\xrightarrow{\tau}P'\parallel (Q'\parallel R)}
      \quad\frac{P\xrightarrow{l}P'\quad Q\xrightarrow{\overline{l}}Q'\quad R\nrightarrow}{(P\parallel Q)\parallel R\xrightarrow{\tau}(P'\parallel Q')\parallel R}$$

      $$\frac{P\xrightarrow{l}P'\quad R\xrightarrow{\overline{l}}R'\quad Q\nrightarrow}{P\parallel (Q\parallel R)\xrightarrow{\tau}P'\parallel (Q\parallel R')}
      \quad\frac{P\xrightarrow{l}P'\quad R\xrightarrow{\overline{l}}R'\quad Q\nrightarrow}{(P\parallel Q)\parallel R\xrightarrow{\tau}(P'\parallel Q)\parallel R]}$$

      $$\frac{Q\xrightarrow{l}P'\quad R\xrightarrow{\overline{l}}R'\quad P\nrightarrow}{P\parallel (Q\parallel R)\xrightarrow{\tau}P\parallel (Q'\parallel R')}
      \quad\frac{Q\xrightarrow{l}Q'\quad R\xrightarrow{\overline{l}}R'\quad P\nrightarrow}{(P\parallel Q)\parallel R\xrightarrow{\tau}(P\parallel Q')\parallel R'}$$

      $$\frac{P\xrightarrow{l}P'\quad Q\xrightarrow{\overline{l}}Q'\quad R\xrightarrow{\gamma}R'}{P\parallel (Q\parallel R)\xrightarrow{\tau,\gamma}P'\parallel (Q'\parallel R')}
      \quad\frac{P\xrightarrow{l}P'\quad Q\xrightarrow{\overline{l}}Q'\quad R\xrightarrow{\gamma}R'}{(P\parallel Q)\parallel R\xrightarrow{\tau,\gamma}(P'\parallel Q')\parallel R'}$$

      $$\frac{P\xrightarrow{l}P'\quad R\xrightarrow{\overline{l}}R'\quad Q\xrightarrow{\beta}Q'}{P\parallel (Q\parallel R)\xrightarrow{\tau,\beta}P'\parallel (Q'\parallel R')}
      \quad\frac{P\xrightarrow{l}P'\quad R\xrightarrow{\overline{l}}R'\quad Q\xrightarrow{\beta}Q'}{(P\parallel Q)\parallel R\xrightarrow{\tau,\beta}(P'\parallel Q')\parallel R]}$$

      $$\frac{Q\xrightarrow{l}Q'\quad R\xrightarrow{\overline{l}}R'\quad P\xrightarrow{\alpha}P'}{P\parallel (Q\parallel R)\xrightarrow{\tau,\alpha}P'\parallel (Q'\parallel R')}
      \quad\frac{Q\xrightarrow{l}Q'\quad R\xrightarrow{\overline{l}}R'\quad P\xrightarrow{\alpha}P'}{(P\parallel Q)\parallel R\xrightarrow{\tau,\alpha}(P'\parallel Q')\parallel R'}$$

      By the reverse transition rules of Composition, we get

      $$\frac{P\xtworightarrow{\alpha[m]}P'\quad Q\xntworightarrow{}\quad R\xntworightarrow{}}{P\parallel (Q\parallel R)\xtworightarrow{\alpha[m]}P'\parallel (Q\parallel R)}
      \quad\frac{P\xtworightarrow{\alpha[m]}P'\quad Q\xntworightarrow{}\quad R\xntworightarrow{}}{(P\parallel Q)\parallel R\xtworightarrow{\alpha[m]}(P'\parallel Q)\parallel R}$$

      $$\frac{Q\xtworightarrow{\beta[n]}Q'\quad P\xntworightarrow{}\quad R\xntworightarrow{}}{P\parallel (Q\parallel R)\xtworightarrow{\beta[n]}P\parallel (Q'\parallel R)}
      \quad\frac{Q\xtworightarrow{\beta[n]}Q'\quad P\xntworightarrow{}\quad R\xntworightarrow{}}{(P\parallel Q)\parallel R\xtworightarrow{\beta[n]}(P\parallel Q')\parallel R}$$

      $$\frac{R\xtworightarrow{\gamma[k]}R'\quad P\xntworightarrow{}\quad Q\xntworightarrow{}}{P\parallel (Q\parallel R)\xtworightarrow{\gamma[k]}P\parallel (Q\parallel R')}
      \quad\frac{R\xtworightarrow{\gamma[k]}R'\quad P\xntworightarrow{}\quad Q\xntworightarrow{}}{(P\parallel Q)\parallel R\xtworightarrow{\gamma[k]}(P\parallel Q)\parallel R'}$$

      $$\frac{P\xtworightarrow{\alpha[m]}P'\quad Q\xtworightarrow{\beta[m]}Q'\quad R\xntworightarrow{}}{P\parallel (Q\parallel R)\xtworightarrow{\{\alpha[m],\beta[m]\}}P'\parallel (Q'\parallel R)}(\beta\neq\overline{\alpha})
      \quad\frac{P\xtworightarrow{\alpha[m]}P'\quad Q\xtworightarrow{\beta[m]}Q'\quad R\xntworightarrow{}}{(P\parallel Q)\parallel R\xtworightarrow{\{\alpha[m],\beta[m]\}}(P'\parallel Q')\parallel R}(\beta\neq\overline{\alpha})$$

      $$\frac{P\xtworightarrow{\alpha[m]}P'\quad R\xtworightarrow{\gamma[m]}R'\quad Q\xntworightarrow{}}{P\parallel (Q\parallel R)\xtworightarrow{\{\alpha[m],\gamma[m]\}}P'\parallel (Q\parallel R')}(\gamma\neq\overline{\alpha})
      \quad\frac{P\xtworightarrow{\alpha[m]}P'\quad R\xtworightarrow{\gamma[m]}R'\quad Q\xntworightarrow{}}{(P\parallel Q)\parallel R\xtworightarrow{\{\alpha[m],\gamma[m]\}}(P'\parallel Q)\parallel R]}(\gamma\neq\overline{\alpha})$$

      $$\frac{Q\xtworightarrow{\beta[m]}P'\quad R\xtworightarrow{\gamma[m]}R'\quad P\xntworightarrow{}}{P\parallel (Q\parallel R)\xtworightarrow{\{\beta[m],\gamma[m]\}}P\parallel (Q'\parallel R')}(\gamma\neq\overline{\beta})
      \quad\frac{Q\xtworightarrow{\beta[m]}Q'\quad R\xtworightarrow{\gamma[m]}R'\quad P\xntworightarrow{}}{(P\parallel Q)\parallel R\xtworightarrow{\{\beta[m],\gamma[m]\}}(P\parallel Q')\parallel R'}(\gamma\neq\overline{\beta})$$

      $$\frac{P\xtworightarrow{\alpha[m]}P'\quad Q\xtworightarrow{\beta[m]}Q'\quad R\xtworightarrow{\gamma[m]}R'}{P\parallel (Q\parallel R)\xtworightarrow{\{\alpha[m],\beta[m],\gamma[m]\}}P'\parallel (Q'\parallel R')}(\beta\neq\overline{\alpha},\gamma\neq\overline{\alpha},\gamma\neq\overline{\beta})$$
      $$\frac{P\xtworightarrow{\alpha[m]}P'\quad Q\xtworightarrow{\beta[m]}Q'\quad R\xtworightarrow{\gamma[m]}R'}{(P\parallel Q)\parallel R\xtworightarrow{\{\alpha[m],\beta[m],\gamma[m]\}}(P'\parallel Q')\parallel R'}(\beta\neq\overline{\alpha},\gamma\neq\overline{\alpha},\gamma\neq\overline{\beta})$$

      $$\frac{P\xtworightarrow{l[m]}P'\quad Q\xtworightarrow{\overline{l}[m]}Q'\quad R\xntworightarrow{}}{P\parallel (Q\parallel R)\xtworightarrow{\tau}P'\parallel (Q'\parallel R)}
      \quad\frac{P\xtworightarrow{l[m]}P'\quad Q\xtworightarrow{\overline{l}[m]}Q'\quad R\xntworightarrow{}}{(P\parallel Q)\parallel R\xtworightarrow{\tau}(P'\parallel Q')\parallel R}$$

      $$\frac{P\xtworightarrow{l[m]}P'\quad R\xtworightarrow{\overline{l}[m]}R'\quad Q\xntworightarrow{}}{P\parallel (Q\parallel R)\xtworightarrow{\tau}P'\parallel (Q\parallel R')}
      \quad\frac{P\xtworightarrow{l[m]}P'\quad R\xtworightarrow{\overline{l}[m]}R'\quad Q\xntworightarrow{}}{(P\parallel Q)\parallel R\xtworightarrow{\tau}(P'\parallel Q)\parallel R]}$$

      $$\frac{Q\xtworightarrow{l[m]}P'\quad R\xtworightarrow{\overline{l}[m]}R'\quad P\xntworightarrow{}}{P\parallel (Q\parallel R)\xtworightarrow{\tau}P\parallel (Q'\parallel R')}
      \quad\frac{Q\xtworightarrow{l[m]}Q'\quad R\xtworightarrow{\overline{l}[m]}R'\quad P\xntworightarrow{}}{(P\parallel Q)\parallel R\xtworightarrow{\tau}(P\parallel Q')\parallel R'}$$

      $$\frac{P\xtworightarrow{l[m]}P'\quad Q\xtworightarrow{\overline{l}[m]}Q'\quad R\xtworightarrow{\gamma[m]}R'}{P\parallel (Q\parallel R)\xtworightarrow{\tau,\gamma[m]}P'\parallel (Q'\parallel R')}
      \quad\frac{P\xtworightarrow{l[m]}P'\quad Q\xtworightarrow{\overline{l}[m]}Q'\quad R\xtworightarrow{\gamma[m]}R'}{(P\parallel Q)\parallel R\xtworightarrow{\tau,\gamma[m]}(P'\parallel Q')\parallel R'}$$

      $$\frac{P\xtworightarrow{l[m]}P'\quad R\xtworightarrow{\overline{l}[m]}R'\quad Q\xtworightarrow{\beta[m]}Q'}{P\parallel (Q\parallel R)\xtworightarrow{\tau,\beta[m]}P'\parallel (Q'\parallel R')}
      \quad\frac{P\xtworightarrow{l[m]}P'\quad R\xtworightarrow{\overline{l}[m]}R'\quad Q\xtworightarrow{\beta[m]}Q'}{(P\parallel Q)\parallel R\xtworightarrow{\tau,\beta[m]}(P'\parallel Q')\parallel R]}$$

      $$\frac{Q\xtworightarrow{l[m]}Q'\quad R\xtworightarrow{\overline{l}[m]}R'\quad P\xtworightarrow{\alpha[m]}P'}{P\parallel (Q\parallel R)\xtworightarrow{\tau,\alpha[m]}P'\parallel (Q'\parallel R')}
      \quad\frac{Q\xtworightarrow{l[m]}Q'\quad R\xtworightarrow{\overline{l}[m]}R'\quad P\xtworightarrow{\alpha[m]}P'}{(P\parallel Q)\parallel R\xtworightarrow{\tau,\alpha[m]}(P'\parallel Q')\parallel R'}$$

      So, with the assumptions $P'\parallel (Q\parallel R) \sim_s^{fr} (P'\parallel Q)\parallel R$, $P\parallel (Q'\parallel R) \sim_s^{fr} (P\parallel Q')\parallel R$, $P\parallel (Q\parallel R') \sim_s^{fr} (P\parallel Q)\parallel R'$, $P'\parallel (Q'\parallel R) \sim_s^{fr} (P'\parallel Q')\parallel R$, $P'\parallel (Q\parallel R') \sim_s^{fr} (P'\parallel Q)\parallel R'$, $P\parallel (Q'\parallel R') \sim_s^{fr} (P\parallel Q')\parallel R'$ and $P'\parallel (Q'\parallel R') \sim_s^{fr} (P'\parallel Q')\parallel R'$, $P\parallel (Q\parallel R) \sim_s^{fr} (P\parallel Q)\parallel R$, as desired.
  \item $P\parallel \textbf{nil}\sim_s^{fr} P$. By the forward transition rules of Composition, we get

      $$\frac{P\xrightarrow{\alpha}P'}{P\parallel \textbf{nil}\xrightarrow{\alpha}P'} \quad \frac{P\xrightarrow{\alpha}P'}{P\xrightarrow{\alpha}P'}$$

      By the reverse transition rules of Composition, we get

      $$\frac{P\xtworightarrow{\alpha[m]}P'}{P\parallel \textbf{nil}\xtworightarrow{\alpha[m]}P'} \quad \frac{P\xtworightarrow{\alpha[m]}P'}{P\xtworightarrow{\alpha[m]}P'}$$

      Since $P'\sim_s^{fr} P'$, $P\parallel \textbf{nil}\sim_s^{fr} P$, as desired.
  \item $P\setminus L\sim_s^{fr} P$, if $\mathcal{L}(P)\cap(L\cup\overline{L})=\emptyset$. By the forward transition rules of Restriction, we get

      $$\frac{P\xrightarrow{\alpha}P'}{P\setminus L\xrightarrow{\alpha}P'\setminus L} (\mathcal{L}(P)\cap(L\cup\overline{L})=\emptyset)\quad \frac{P\xrightarrow{\alpha}P'}{P\xrightarrow{\alpha}P'}$$

      By the reverse transition rules of Restriction, we get

      $$\frac{P\xtworightarrow{\alpha[m]}P'}{P\setminus L\xtworightarrow{\alpha[m]}P'\setminus L} (\mathcal{L}(P)\cap(L\cup\overline{L})=\emptyset)\quad \frac{P\xtworightarrow{\alpha[m]}P'}{P\xtworightarrow{\alpha[m]}P'}$$

      Since $P'\sim_s^{fr} P'$, and with the assumption $P'\setminus L\sim_s^{fr} P'$, $P\setminus L\sim_s^{fr} P$, if $\mathcal{L}(P)\cap(L\cup\overline{L})=\emptyset$, as desired.
  \item $P\setminus K\setminus L\sim_s^{fr} P\setminus(K\cup L)$. By the forward transition rules of Restriction, we get

      $$\frac{P\xrightarrow{\alpha}P'}{P\setminus K\setminus L\xrightarrow{\alpha}P'\setminus K\setminus L} \quad \frac{P\xrightarrow{\alpha}P'}{P\setminus (K\cup L)\xrightarrow{\alpha}P'\setminus (K\cup L)}$$

      By the reverse transition rules of Restriction, we get

      $$\frac{P\xtworightarrow{\alpha[m]}P'}{P\setminus K\setminus L\xtworightarrow{\alpha[m]}P'\setminus K\setminus L} \quad \frac{P\xtworightarrow{\alpha[m]}P'}{P\setminus (K\cup L)\xtworightarrow{\alpha[m]}P'\setminus (K\cup L)}$$

      Since $P'\sim_s^{fr} P'$, and with the assumption $P'\setminus K\setminus L\sim_s^{fr} P'\setminus(K\cup L)$, $P\setminus K\setminus L\sim_s^{fr} P\setminus(K\cup L)$, as desired.
  \item $P[f]\setminus L\sim_s^{fr} P\setminus f^{-1}(L)[f]$. By the forward transition rules of Restriction and Relabelling, we get

      $$\frac{P\xrightarrow{\alpha}P'}{P[f]\setminus L\xrightarrow{f(\alpha)}P'[f]\setminus L}\quad \frac{P\xrightarrow{\alpha}P'}{P\setminus f^{-1}(L)[f]\xrightarrow{f(\alpha)}P'\setminus f^{-1}(L)[f]}$$

      By the reverse transition rules of Restriction and Relabelling, we get

      $$\frac{P\xtworightarrow{\alpha[m]}P'}{P[f]\setminus L\xtworightarrow{f(\alpha)[m]}P'[f]\setminus L}\quad \frac{P\xtworightarrow{\alpha[m]}P'}{P\setminus f^{-1}(L)[f]\xtworightarrow{f(\alpha)[m]}P'\setminus f^{-1}(L)[f]}$$

      So, with the assumption $P'[f]\setminus L\sim_s^{fr} P'\setminus f^{-1}(L)[f]$, $P[f]\setminus L\sim_s^{fr} P\setminus f^{-1}(L)[f]$, as desired.
  \item $(P\parallel Q)\setminus L\sim_s^{fr} P\setminus L\parallel Q\setminus L$, if $\mathcal{L}(P)\cap\overline{\mathcal{L}(Q)}\cap(L\cup\overline{L})=\emptyset$. By the forward transition rules of Composition and Restriction, we get

      $$\frac{P\xrightarrow{\alpha}P'\quad Q\nrightarrow}{(P\parallel Q)\setminus L\xrightarrow{\alpha}(P'\parallel Q)\setminus L}(\mathcal{L}(P)\cap\overline{\mathcal{L}(Q)}\cap(L\cup\overline{L})=\emptyset)$$
      $$\frac{P\xrightarrow{\alpha}P'\quad Q\nrightarrow}{P\setminus L\parallel Q\setminus L\xrightarrow{\alpha}P'\setminus L\parallel Q\setminus L}(\mathcal{L}(P)\cap\overline{\mathcal{L}(Q)}\cap(L\cup\overline{L})=\emptyset)$$

      $$\frac{Q\xrightarrow{\beta}Q'\quad P\nrightarrow}{(P\parallel Q)\setminus L\xrightarrow{\beta}(P\parallel Q')\setminus L}(\mathcal{L}(P)\cap\overline{\mathcal{L}(Q)}\cap(L\cup\overline{L})=\emptyset)$$
      $$\frac{Q\xrightarrow{\beta}Q'\quad P\nrightarrow}{P\setminus L\parallel Q\setminus L\xrightarrow{\beta}P\setminus L\parallel Q'\setminus L}(\mathcal{L}(P)\cap\overline{\mathcal{L}(Q)}\cap(L\cup\overline{L})=\emptyset)$$

      $$\frac{P\xrightarrow{\alpha}P'\quad Q\xrightarrow{\beta}Q'}{(P\parallel Q)\setminus L\xrightarrow{\{\alpha,\beta\}}(P'\parallel Q')\setminus L}(\mathcal{L}(P)\cap\overline{\mathcal{L}(Q)}\cap(L\cup\overline{L})=\emptyset)$$
      $$\frac{P\xrightarrow{\alpha}P'\quad Q\xrightarrow{\beta}Q'}{P\setminus L\parallel Q\setminus L\xrightarrow{\{\alpha,\beta\}}(P'\parallel Q')\setminus L}(\mathcal{L}(P)\cap\overline{\mathcal{L}(Q)}\cap(L\cup\overline{L})=\emptyset)$$

      $$\frac{P\xrightarrow{l}P'\quad Q\xrightarrow{\overline{l}}Q'}{(P\parallel Q)\setminus L\xrightarrow{\tau}(P'\parallel Q')\setminus L}(\mathcal{L}(P)\cap\overline{\mathcal{L}(Q)}\cap(L\cup\overline{L})=\emptyset)$$
      $$\frac{P\xrightarrow{l}P'\quad Q\xrightarrow{\overline{l}}Q'}{(P\setminus L\parallel Q\setminus L\xrightarrow{\tau}P'\setminus L\parallel Q'\setminus L}(\mathcal{L}(P)\cap\overline{\mathcal{L}(Q)}\cap(L\cup\overline{L})=\emptyset)$$

      By the reverse transition rules of Composition and Restriction, we get

      $$\frac{P\xtworightarrow{\alpha[m]}P'\quad Q\xntworightarrow{}}{(P\parallel Q)\setminus L\xtworightarrow{\alpha[m]}(P'\parallel Q)\setminus L}(\mathcal{L}(P)\cap\overline{\mathcal{L}(Q)}\cap(L\cup\overline{L})=\emptyset)$$
      $$\frac{P\xtworightarrow{\alpha[m]}P'\quad Q\xntworightarrow{}}{P\setminus L\parallel Q\setminus L\xtworightarrow{\alpha[m]}P'\setminus L\parallel Q\setminus L}(\mathcal{L}(P)\cap\overline{\mathcal{L}(Q)}\cap(L\cup\overline{L})=\emptyset)$$

      $$\frac{Q\xtworightarrow{\beta[n]}Q'\quad P\xntworightarrow{}}{(P\parallel Q)\setminus L\xtworightarrow{\beta[n]}(P\parallel Q')\setminus L}(\mathcal{L}(P)\cap\overline{\mathcal{L}(Q)}\cap(L\cup\overline{L})=\emptyset)$$
      $$\frac{Q\xtworightarrow{\beta[n]}Q'\quad P\xntworightarrow{}}{P\setminus L\parallel Q\setminus L\xtworightarrow{\beta[n]}P\setminus L\parallel Q'\setminus L}(\mathcal{L}(P)\cap\overline{\mathcal{L}(Q)}\cap(L\cup\overline{L})=\emptyset)$$

      $$\frac{P\xtworightarrow{\alpha[m]}P'\quad Q\xtworightarrow{\beta[m]}Q'}{(P\parallel Q)\setminus L\xtworightarrow{\{\alpha[m],\beta[m]\}}(P'\parallel Q')\setminus L}(\mathcal{L}(P)\cap\overline{\mathcal{L}(Q)}\cap(L\cup\overline{L})=\emptyset)$$
      $$\frac{P\xtworightarrow{\alpha[m]}P'\quad Q\xtworightarrow{\beta[m]}Q'}{P\setminus L\parallel Q\setminus L\xtworightarrow{\{\alpha[m],\beta[m]\}}(P'\parallel Q')\setminus L}(\mathcal{L}(P)\cap\overline{\mathcal{L}(Q)}\cap(L\cup\overline{L})=\emptyset)$$

      $$\frac{P\xtworightarrow{l[m]}P'\quad Q\xtworightarrow{\overline{l}[m]}Q'}{(P\parallel Q)\setminus L\xtworightarrow{\tau}(P'\parallel Q')\setminus L}(\mathcal{L}(P)\cap\overline{\mathcal{L}(Q)}\cap(L\cup\overline{L})=\emptyset)$$
      $$\frac{P\xtworightarrow{l[m]}P'\quad Q\xtworightarrow{\overline{l}[m]}Q'}{(P\setminus L\parallel Q\setminus L\xtworightarrow{\tau}P'\setminus L\parallel Q'\setminus L}(\mathcal{L}(P)\cap\overline{\mathcal{L}(Q)}\cap(L\cup\overline{L})=\emptyset)$$

      Since $(P'\parallel Q)\setminus L\sim_s^{fr}P'\setminus L\parallel Q\setminus L$, $(P\parallel Q')\setminus L\sim_s^{fr} P\setminus L\parallel Q'\setminus L$ and $(P'\parallel Q')\setminus L\sim_s^{fr} P'\setminus L\parallel Q'\setminus L$, $(P\parallel Q)\setminus L\sim_s^{fr} P\setminus L\parallel Q\setminus L$, if $\mathcal{L}(P)\cap\overline{\mathcal{L}(Q)}\cap(L\cup\overline{L})=\emptyset$, as desired.
  \item $P[Id]\sim_s^{fr} P$. By the forward transition rules Relabelling, we get

      $$\frac{P\xrightarrow{\alpha}P'}{P[Id]\xrightarrow{Id(\alpha)}P'[Id]}\quad \frac{P\xrightarrow{\alpha}P'}{P\xrightarrow{\alpha}P'}$$

      By the reverse transition rules Relabelling, we get

      $$\frac{P\xtworightarrow{\alpha[m]}P'}{P[Id]\xtworightarrow{Id(\alpha[m])}P'[Id]}\quad \frac{P\xtworightarrow{\alpha[m]}P'}{P\xtworightarrow{\alpha[m]}P'}$$

      So, with the assumption $P'[Id]\sim_s^{fr} P'$ and $Id(\alpha)=\alpha$, $P[Id]\sim_s^{fr} P$, as desired.
  \item $P[f]\sim_s^{fr} P[f']$, if $f\upharpoonright\mathcal{L}(P)=f'\upharpoonright\mathcal{L}(P)$. By the forward transition rules of Relabelling, we get

      $$\frac{P\xrightarrow{\alpha}P'}{P[f]\xrightarrow{f(\alpha)}P'[f]}\quad \frac{P\xrightarrow{\alpha}P'}{P[f']\xrightarrow{f'(\alpha)}P'[f']}$$

      By the reverse transition rules of Relabelling, we get

      $$\frac{P\xtworightarrow{\alpha[m]}P'}{P[f]\xtworightarrow{f(\alpha)[m]}P'[f]}\quad \frac{P\xtworightarrow{\alpha[m]}P'}{P[f']\xtworightarrow{f'(\alpha)[m]}P'[f']}$$

      So, with the assumption $P'[f]\sim_s^{fr} P'[f']$ and $f(\alpha)=f'(\alpha)$, if $f\upharpoonright\mathcal{L}(P)=f'\upharpoonright\mathcal{L}(P)$, $P[f]\sim_s^{fr} P[f']$, as desired.
  \item $P[f][f']\sim_s^{fr} P[f'\circ f]$. By the forward transition rules of Relabelling, we get

      $$\frac{P\xrightarrow{\alpha}P'}{P[f][f']\xrightarrow{f'(f(\alpha))}P'[f][f']}\quad \frac{P\xrightarrow{\alpha}P'}{P[f'\circ f]\xrightarrow{f'(f(\alpha))}P'[f'\circ f]}$$

      By the reverse transition rules of Relabelling, we get

      $$\frac{P\xtworightarrow{\alpha[m]}P'}{P[f][f']\xrightarrow{f'(f(\alpha))[m]}P'[f][f']}\quad \frac{P\xtworightarrow{\alpha[m]}P'}{P[f'\circ f]\xtworightarrow{f'(f(\alpha))[m]}P'[f'\circ f]}$$

      So, with the assumption $P'[f][f']\sim_s^{fr} P'[f'\circ f]$, $P[f][f']\sim_s^{fr} P[f'\circ f]$, as desired.
  \item $(P\parallel Q)[f]\sim_s^{fr} P[f]\parallel Q[f]$, if $f\upharpoonright(L\cup\overline{L})$ is one-to-one, where $L=\mathcal{L}(P)\cup\mathcal{L}(Q)$. By the forward transition rules of Composition and Relabelling, we get

      $$\frac{P\xrightarrow{\alpha}P'\quad Q\nrightarrow}{(P\parallel Q)[f]\xrightarrow{f(\alpha)}(P'\parallel Q)[f]}(\textrm{if } f\upharpoonright(L\cup\overline{L}) \textrm{ is one-to-one, where }L=\mathcal{L}(P)\cup\mathcal{L}(Q))$$
      $$\frac{P\xrightarrow{\alpha}P'\quad Q\nrightarrow}{P[f]\parallel Q[f]\xrightarrow{\alpha}P'[f]\parallel Q[f]}(\textrm{if } f\upharpoonright(L\cup\overline{L}) \textrm{ is one-to-one, where }L=\mathcal{L}(P)\cup\mathcal{L}(Q))$$

      $$\frac{Q\xrightarrow{\beta}Q'\quad P\nrightarrow}{(P\parallel Q)[f]\xrightarrow{f(\beta)}(P\parallel Q')[f]}(\textrm{if } f\upharpoonright(L\cup\overline{L}) \textrm{ is one-to-one, where }L=\mathcal{L}(P)\cup\mathcal{L}(Q))$$
      $$\frac{Q\xrightarrow{\beta}Q'\quad P\nrightarrow}{P[f]\parallel Q[f]\xrightarrow{\beta}P[f]\parallel Q'[f]}(\textrm{if } f\upharpoonright(L\cup\overline{L}) \textrm{ is one-to-one, where }L=\mathcal{L}(P)\cup\mathcal{L}(Q))$$

      $$\frac{P\xrightarrow{\alpha}P'\quad Q\xrightarrow{\beta}Q'}{(P\parallel Q)[f]\xrightarrow{\{f(\alpha),f(\beta)\}}(P'\parallel Q')[f]}(\textrm{if } f\upharpoonright(L\cup\overline{L}) \textrm{ is one-to-one, where }L=\mathcal{L}(P)\cup\mathcal{L}(Q))$$
      $$\frac{P\xrightarrow{\alpha}P'\quad Q\xrightarrow{\beta}Q'}{P[f]\parallel Q[f]\xrightarrow{\{f(\alpha),f(\beta)\}}P'[f]\parallel Q'[f]}(\textrm{if } f\upharpoonright(L\cup\overline{L}) \textrm{ is one-to-one, where }L=\mathcal{L}(P)\cup\mathcal{L}(Q))$$

      $$\frac{P\xrightarrow{l}P'\quad Q\xrightarrow{\overline{l}}Q'}{(P\parallel Q)[f]\xrightarrow{\tau}(P'\parallel Q')[f]}(\textrm{if } f\upharpoonright(L\cup\overline{L}) \textrm{ is one-to-one, where }L=\mathcal{L}(P)\cup\mathcal{L}(Q))$$
      $$\frac{P\xrightarrow{l}P'\quad Q\xrightarrow{\overline{l}}Q'}{(P[f]\parallel Q[f]\xrightarrow{\tau}P'[f]\parallel Q'[f]}(\textrm{if } f\upharpoonright(L\cup\overline{L}) \textrm{ is one-to-one, where }L=\mathcal{L}(P)\cup\mathcal{L}(Q))$$

      By the reverse transition rules of Composition and Relabelling, we get

      $$\frac{P\xtworightarrow{\alpha[m]}P'\quad Q\xntworightarrow{}}{(P\parallel Q)[f]\xtworightarrow{f(\alpha)[m]}(P'\parallel Q)[f]}(\textrm{if } f\upharpoonright(L\cup\overline{L}) \textrm{ is one-to-one, where }L=\mathcal{L}(P)\cup\mathcal{L}(Q))$$
      $$\frac{P\xtworightarrow{\alpha[m]}P'\quad Q\xntworightarrow{}}{P[f]\parallel Q[f]\xtworightarrow{f(\alpha)[m]}P'[f]\parallel Q[f]}(\textrm{if } f\upharpoonright(L\cup\overline{L}) \textrm{ is one-to-one, where }L=\mathcal{L}(P)\cup\mathcal{L}(Q))$$

      $$\frac{Q\xtworightarrow{\beta[n]}Q'\quad P\xntworightarrow{}}{(P\parallel Q)[f]\xtworightarrow{f(\beta)[n]}(P\parallel Q')[f]}(\textrm{if } f\upharpoonright(L\cup\overline{L}) \textrm{ is one-to-one, where }L=\mathcal{L}(P)\cup\mathcal{L}(Q))$$
      $$\frac{Q\xtworightarrow{\beta[n]}Q'\quad P\xntworightarrow{}}{P[f]\parallel Q[f]\xtworightarrow{f(\beta)[n]}P[f]\parallel Q'[f]}(\textrm{if } f\upharpoonright(L\cup\overline{L}) \textrm{ is one-to-one, where }L=\mathcal{L}(P)\cup\mathcal{L}(Q))$$

      $$\frac{P\xtworightarrow{\alpha[m]}P'\quad Q\xtworightarrow{\beta[m]}Q'}{(P\parallel Q)[f]\xtworightarrow{\{f(\alpha)[m],f(\beta)[m]\}}(P'\parallel Q')[f]}(\textrm{if } f\upharpoonright(L\cup\overline{L}) \textrm{ is one-to-one, where }L=\mathcal{L}(P)\cup\mathcal{L}(Q))$$
      $$\frac{P\xtworightarrow{\alpha[m]}P'\quad Q\xtworightarrow{\beta[m]}Q'}{P[f]\parallel Q[f]\xtworightarrow{\{f(\alpha)[m],f(\beta)[m]\}}P'[f]\parallel Q'[f]}(\textrm{if } f\upharpoonright(L\cup\overline{L}) \textrm{ is one-to-one, where }L=\mathcal{L}(P)\cup\mathcal{L}(Q))$$

      $$\frac{P\xtworightarrow{l[m]}P'\quad Q\xtworightarrow{\overline{l}[m]}Q'}{(P\parallel Q)[f]\xtworightarrow{\tau}(P'\parallel Q')[f]}(\textrm{if } f\upharpoonright(L\cup\overline{L}) \textrm{ is one-to-one, where }L=\mathcal{L}(P)\cup\mathcal{L}(Q))$$
      $$\frac{P\xtworightarrow{l[m]}P'\quad Q\xtworightarrow{\overline{l}[m]}Q'}{(P[f]\parallel Q[f]\xtworightarrow{\tau}P'[f]\parallel Q'[f]}(\textrm{if } f\upharpoonright(L\cup\overline{L}) \textrm{ is one-to-one, where }L=\mathcal{L}(P)\cup\mathcal{L}(Q))$$

      So, with the assumptions $(P'\parallel Q)[f]\sim_s^{fr} P'[f]\parallel Q[f]$, $(P\parallel Q')[f]\sim_s^{fr} P[f]\parallel Q'[f]$ and $(P'\parallel Q')[f]\sim_s^{fr} P'[f]\parallel Q'[f]$, $(P\parallel Q)[f]\sim_s^{fr} P[f]\parallel Q[f]$, if $f\upharpoonright(L\cup\overline{L})$ is one-to-one, where $L=\mathcal{L}(P)\cup\mathcal{L}(Q)$, as desired.
\end{enumerate}
\end{proof}

\begin{proposition}[Static laws for strongly FR pomset bisimulation] \label{SLSPB}
The static laws for strongly FR pomset bisimulation are as follows.
\begin{enumerate}
  \item $P\parallel Q\sim_p^{fr} Q\parallel P$;
  \item $P\parallel(Q\parallel R)\sim_p^{fr} (P\parallel Q)\parallel R$;
  \item $P\parallel \textbf{nil}\sim_p^{fr} P$;
  \item $P\setminus L\sim_p^{fr} P$, if $\mathcal{L}(P)\cap(L\cup\overline{L})=\emptyset$;
  \item $P\setminus K\setminus L\sim_p^{fr} P\setminus(K\cup L)$;
  \item $P[f]\setminus L\sim_p^{fr} P\setminus f^{-1}(L)[f]$;
  \item $(P\parallel Q)\setminus L\sim_p^{fr} P\setminus L\parallel Q\setminus L$, if $\mathcal{L}(P)\cap\overline{\mathcal{L}(Q)}\cap(L\cup\overline{L})=\emptyset$;
  \item $P[Id]\sim_p^{fr} P$;
  \item $P[f]\sim_p^{fr} P[f']$, if $f\upharpoonright\mathcal{L}(P)=f'\upharpoonright\mathcal{L}(P)$;
  \item $P[f][f']\sim_p^{fr} P[f'\circ f]$;
  \item $(P\parallel Q)[f]\sim_p^{fr} P[f]\parallel Q[f]$, if $f\upharpoonright(L\cup\overline{L})$ is one-to-one, where $L=\mathcal{L}(P)\cup\mathcal{L}(Q)$.
\end{enumerate}
\end{proposition}

\begin{proof}
From the definition of strongly FR pomset bisimulation (see Definition \ref{FRPSB}), we know that strongly FR pomset bisimulation is defined by FR pomset transitions, which are labeled by pomsets. In an FR pomset transition, the events in the pomset are either within causality relations (defined by the prefix $.$) or in concurrency (implicitly defined by $.$ and $+$, and explicitly defined by $\parallel$), of course, they are pairwise consistent (without conflicts). In Proposition \ref{SLSSB}, we have already proven the case that all events are pairwise concurrent, so, we only need to prove the case of events in causality. Without loss of generality, we take a pomset of $p=\{\alpha,\beta:\alpha.\beta\}$. Then the FR pomset transition labeled by the above $p$ is just composed of one single event transition labeled by $\alpha$ succeeded by another single event transition labeled by $\beta$, that is, $\xrightarrow{p}=\xrightarrow{\alpha}\xrightarrow{\beta}$ and $\xtworightarrow{p[\mathcal{K}]}=\xtworightarrow{\beta[n]}\xtworightarrow{\alpha[m]}$.

Similarly to the proof of static laws for strongly FR step bisimulation (see Proposition \ref{SLSSB}), we can prove that the static laws hold for strongly FR pomset bisimulation, we omit them.
\end{proof}

\begin{proposition}[Static laws for strongly FR hp-bisimulation] \label{SLSHPB}
The static laws for strongly FR hp-bisimulation are as follows.
\begin{enumerate}
  \item $P\parallel Q\sim_{hp}^{fr} Q\parallel P$;
  \item $P\parallel(Q\parallel R)\sim_{hp}^{fr} (P\parallel Q)\parallel R$;
  \item $P\parallel \textbf{nil}\sim_{hp}^{fr} P$;
  \item $P\setminus L\sim_{hp}^{fr} P$, if $\mathcal{L}(P)\cap(L\cup\overline{L})=\emptyset$;
  \item $P\setminus K\setminus L\sim_{hp}^{fr} P\setminus(K\cup L)$;
  \item $P[f]\setminus L\sim_{hp}^{fr} P\setminus f^{-1}(L)[f]$;
  \item $(P\parallel Q)\setminus L\sim_{hp}^{fr} P\setminus L\parallel Q\setminus L$, if $\mathcal{L}(P)\cap\overline{\mathcal{L}(Q)}\cap(L\cup\overline{L})=\emptyset$;
  \item $P[Id]\sim_{hp}^{fr} P$;
  \item $P[f]\sim_{hp}^{fr} P[f']$, if $f\upharpoonright\mathcal{L}(P)=f'\upharpoonright\mathcal{L}(P)$;
  \item $P[f][f']\sim_{hp}^{fr} P[f'\circ f]$;
  \item $(P\parallel Q)[f]\sim_{hp}^{fr} P[f]\parallel Q[f]$, if $f\upharpoonright(L\cup\overline{L})$ is one-to-one, where $L=\mathcal{L}(P)\cup\mathcal{L}(Q)$.
\end{enumerate}
\end{proposition}

\begin{proof}
From the definition of strongly FR hp-bisimulation (see Definition \ref{FRHHPB}), we know that strongly FR hp-bisimulation is defined on the posetal product $(C_1,f,C_2),f:C_1\rightarrow C_2\textrm{ isomorphism}$. Two processes $P$ related to $C_1$ and $Q$ related to $C_2$, and $f:C_1\rightarrow C_2\textrm{ isomorphism}$. Initially, $(C_1,f,C_2)=(\emptyset,\emptyset,\emptyset)$, and $(\emptyset,\emptyset,\emptyset)\in\sim_{hp}$. When $P\xrightarrow{\alpha}P'$ ($C_1\xrightarrow{\alpha}C_1'$), there will be $Q\xrightarrow{\alpha}Q'$ ($C_2\xrightarrow{\alpha}C_2'$), and we define $f'=f[\alpha\mapsto \alpha]$. And when $P\xtworightarrow{\alpha[m]}P'$ ($C_1\xtworightarrow{\alpha[m]}C_1'$), there will be $Q\xtworightarrow{\alpha[m]}Q'$ ($C_2\xtworightarrow{\alpha[m]}C_2'$), and we define $f'=f[\alpha[m]\mapsto \alpha[m]]$. Then, if $(C_1,f,C_2)\in\sim_{hp}^{fr}$, then $(C_1',f',C_2')\in\sim_{hp}^{fr}$.

Similarly to the proof of static laws for strongly FR pomset bisimulation (see Proposition \ref{SLSPB}), we can prove that static laws hold for strongly FR hp-bisimulation, we just need additionally to check the above conditions on FR hp-bisimulation, we omit them.
\end{proof}

\begin{proposition}[Static laws for strongly FR hhp-bisimulation] \label{SLSHHPB}
The static laws for strongly FR hhp-bisimulation are as follows.
\begin{enumerate}
  \item $P\parallel Q\sim_{hhp}^{fr} Q\parallel P$;
  \item $P\parallel(Q\parallel R)\sim_{hhp}^{fr} (P\parallel Q)\parallel R$;
  \item $P\parallel \textbf{nil}\sim_{hhp}^{fr} P$;
  \item $P\setminus L\sim_{hhp}^{fr} P$, if $\mathcal{L}(P)\cap(L\cup\overline{L})=\emptyset$;
  \item $P\setminus K\setminus L\sim_{hhp}^{fr} P\setminus(K\cup L)$;
  \item $P[f]\setminus L\sim_{hhp}^{fr} P\setminus f^{-1}(L)[f]$;
  \item $(P\parallel Q)\setminus L\sim_{hhp}^{fr} P\setminus L\parallel Q\setminus L$, if $\mathcal{L}(P)\cap\overline{\mathcal{L}(Q)}\cap(L\cup\overline{L})=\emptyset$;
  \item $P[Id]\sim_{hhp}^{fr} P$;
  \item $P[f]\sim_{hhp}^{fr} P[f']$, if $f\upharpoonright\mathcal{L}(P)=f'\upharpoonright\mathcal{L}(P)$;
  \item $P[f][f']\sim_{hhp}^{fr} P[f'\circ f]$;
  \item $(P\parallel Q)[f]\sim_{hhp}^{fr} P[f]\parallel Q[f]$, if $f\upharpoonright(L\cup\overline{L})$ is one-to-one, where $L=\mathcal{L}(P)\cup\mathcal{L}(Q)$.
\end{enumerate}
\end{proposition}

\begin{proof}
From the definition of strongly FR hhp-bisimulation (see Definition \ref{FRHHPB}), we know that strongly FR hhp-bisimulation is downward closed for strongly FR hp-bisimulation.

Similarly to the proof of static laws for strongly FR hp-bisimulation (see Proposition \ref{SLSHPB}), we can prove that static laws hold for strongly FR hhp-bisimulation, that is, they are downward closed for strongly FR hp-bisimulation, we omit them.
\end{proof}

\begin{proposition}[Milner's expansion law for strongly FR truly concurrent bisimulations]
Milner's expansion law does not hold any more for any strongly FR truly concurrent bisimulation, that is,

\begin{enumerate}
  \item $\alpha\parallel\beta\nsim_p^{fr} \alpha.\beta+\beta.\alpha$;
  \item $\alpha\parallel\beta\nsim_s^{fr} \alpha.\beta+\beta.\alpha$;
  \item $\alpha\parallel\beta\nsim_{hp}^{fr} \alpha.\beta+\beta.\alpha$;
  \item $\alpha\parallel\beta\nsim_{hhp}^{fr} \alpha.\beta+\beta.\alpha$.
\end{enumerate}
\end{proposition}

\begin{proof}
In nature, it is caused by $\alpha\parallel \beta$ and $\alpha.\beta + \beta.\alpha$ having different causality structure. By the transition rules of Composition and Prefix, we have

$$\alpha\parallel \beta\xrightarrow{\{\alpha,\beta\}}\textbf{nil}$$

while

$$\alpha.\beta+ \beta.\alpha\nrightarrow^{\{\alpha,\beta\}}.$$

And

$$\alpha[m]\parallel \beta[m]\xtworightarrow{\{\alpha[m],\beta[m]\}}\textbf{nil}$$

while

$$\alpha[m].\beta[m]+ \beta[m].\alpha[m]\xntworightarrow{\{\alpha[m],\beta[m]\}}.$$
\end{proof}

\begin{proposition}[New expansion law for strongly FR step bisimulation]\label{NELSSB}
Let $P\equiv (P_1[f_1]\parallel\cdots\parallel P_n[f_n])\setminus L$, with $n\geq 1$. Then

\begin{eqnarray}
P\sim_s^{f} \{(f_1(\alpha_1)\parallel\cdots\parallel f_n(\alpha_n)).(P_1'[f_1]\parallel\cdots\parallel P_n'[f_n])\setminus L: \nonumber\\
P_i\xrightarrow{\alpha_i}P_i',i\in\{1,\cdots,n\},f_i(\alpha_i)\notin L\cup\overline{L}\} \nonumber\\
+\sum\{\tau.(P_1[f_1]\parallel\cdots\parallel P_i'[f_i]\parallel\cdots\parallel P_j'[f_j]\parallel\cdots\parallel P_n[f_n])\setminus L: \nonumber\\
P_i\xrightarrow{l_1}P_i',P_j\xrightarrow{l_2}P_j',f_i(l_1)=\overline{f_j(l_2)},i<j\}
\end{eqnarray}
\begin{eqnarray}
P\sim_s^{r} \{(P_1'[f_1]\parallel\cdots\parallel P_n'[f_n]).(f_1(\alpha_1[m])\parallel\cdots\parallel f_n(\alpha_n)[m])\setminus L: \nonumber\\
P_i\xtworightarrow{\alpha_i[m]}P_i',i\in\{1,\cdots,n\},f_i(\alpha_i)\notin L\cup\overline{L}\} \nonumber\\
+\sum\{(P_1[f_1]\parallel\cdots\parallel P_i'[f_i]\parallel\cdots\parallel P_j'[f_j]\parallel\cdots\parallel P_n[f_n]).\tau\setminus L: \nonumber\\
P_i\xtworightarrow{l_1[m]}P_i',P_j\xtworightarrow{l_2[m]}P_j',f_i(l_1)=\overline{f_j(l_2)},i<j\}
\end{eqnarray}
\end{proposition}

\begin{proof}
Though transition rules in Definition \ref{semantics} are defined in the flavor of single event, they can be modified into a step (a set of events within which each event is pairwise concurrent), we omit them. If we treat a single event as a step containing just one event, the proof of the new expansion law has not any problem, so we use this way and still use the transition rules in Definition \ref{semantics}.

(1) The case of strongly forward step bisimulation.

Firstly, we consider the case without Restriction and Relabeling. That is, we suffice to prove the following case by induction on the size $n$.

For $P\equiv P_1\parallel\cdots\parallel P_n$, with $n\geq 1$, we need to prove

\begin{eqnarray}
P\sim_s \{(\alpha_1\parallel\cdots\parallel \alpha_n).(P_1'\parallel\cdots\parallel P_n'): P_i\xrightarrow{\alpha_i}P_i',i\in\{1,\cdots,n\}\nonumber\\
+\sum\{\tau.(P_1\parallel\cdots\parallel P_i'\parallel\cdots\parallel P_j'\parallel\cdots\parallel P_n): P_i\xrightarrow{l}P_i',P_j\xrightarrow{\overline{l}}P_j',i<j\} \nonumber
\end{eqnarray}

For $n=1$, $P_1\sim_s^{f} \alpha_1.P_1':P_1\xrightarrow{\alpha_1}P_1'$ is obvious. Then with a hypothesis $n$, we consider $R\equiv P\parallel P_{n+1}$. By the forward transition rules of Composition, we can get

\begin{eqnarray}
R\sim_s^{f} \{(p\parallel \alpha_{n+1}).(P'\parallel P_{n+1}'): P\xrightarrow{p}P',P_{n+1}\xrightarrow{\alpha_{n+1}}P_{n+1}',p\subseteq P\}\nonumber\\
+\sum\{\tau.(P'\parallel P_{n+1}'): P\xrightarrow{l}P',P_{n+1}\xrightarrow{\overline{l}}P_{n+1}'\} \nonumber
\end{eqnarray}

Now with the induction assumption $P\equiv P_1\parallel\cdots\parallel P_n$, the right-hand side can be reformulated as follows.

\begin{eqnarray}
\{(\alpha_1\parallel\cdots\parallel \alpha_n\parallel \alpha_{n+1}).(P_1'\parallel\cdots\parallel P_n'\parallel P_{n+1}'): \nonumber\\
P_i\xrightarrow{\alpha_i}P_i',i\in\{1,\cdots,n+1\}\nonumber\\
+\sum\{\tau.(P_1\parallel\cdots\parallel P_i'\parallel\cdots\parallel P_j'\parallel\cdots\parallel P_n\parallel P_{n+1}): \nonumber\\
P_i\xrightarrow{l}P_i',P_j\xrightarrow{\overline{l}}P_j',i<j\} \nonumber\\
+\sum\{\tau.(P_1\parallel\cdots\parallel P_i'\parallel\cdots\parallel P_j\parallel\cdots\parallel P_n\parallel P_{n+1}'): \nonumber\\
P_i\xrightarrow{l}P_i',P_{n+1}\xrightarrow{\overline{l}}P_{n+1}',i\in\{1,\cdots, n\}\} \nonumber
\end{eqnarray}

So,

\begin{eqnarray}
R\sim_s^{f} \{(\alpha_1\parallel\cdots\parallel \alpha_n\parallel \alpha_{n+1}).(P_1'\parallel\cdots\parallel P_n'\parallel P_{n+1}'): \nonumber\\
P_i\xrightarrow{\alpha_i}P_i',i\in\{1,\cdots,n+1\}\nonumber\\
+\sum\{\tau.(P_1\parallel\cdots\parallel P_i'\parallel\cdots\parallel P_j'\parallel\cdots\parallel P_n): \nonumber\\
P_i\xrightarrow{l}P_i',P_j\xrightarrow{\overline{l}}P_j',1 \leq i<j\geq n+1\} \nonumber
\end{eqnarray}

Then, we can easily add the full conditions with Restriction and Relabeling.

(2) The case of strongly reverse step bisimulation.

Firstly, we consider the case without Restriction and Relabeling. That is, we suffice to prove the following case by induction on the size $n$.

For $P\equiv P_1\parallel\cdots\parallel P_n$, with $n\geq 1$, we need to prove

\begin{eqnarray}
P\sim_s^{r} \{(P_1'\parallel\cdots\parallel P_n').(\alpha_1[m]\parallel\cdots\parallel \alpha_n[m]): P_i\xtworightarrow{\alpha_i[m]}P_i',i\in\{1,\cdots,n\}\nonumber\\
+\sum\{(P_1\parallel\cdots\parallel P_i'\parallel\cdots\parallel P_j'\parallel\cdots\parallel P_n).\tau: P_i\xtworightarrow{l[m]}P_i',P_j\xtworightarrow{\overline{l}[m]}P_j',i<j\} \nonumber
\end{eqnarray}

For $n=1$, $P_1\sim_s^{r} P_1'.\alpha_1[m]:P_1\xtworightarrow{\alpha_1[m]}P_1'$ is obvious. Then with a hypothesis $n$, we consider $R\equiv P\parallel P_{n+1}$. By the reverse transition rules of Composition, we can get

\begin{eqnarray}
R\sim_s^{r} \{(P'\parallel P_{n+1}').(p[m]\parallel \alpha_{n+1}[m]): P\xtworightarrow{p[m]}P',P_{n+1}\xtworightarrow{\alpha_{n+1}[m]}P_{n+1}',p\subseteq P\}\nonumber\\
+\sum\{(P'\parallel P_{n+1}').\tau: P\xtworightarrow{l[m]}P',P_{n+1}\xtworightarrow{\overline{l}[m]}P_{n+1}'\} \nonumber
\end{eqnarray}

Now with the induction assumption $P\equiv P_1\parallel\cdots\parallel P_n$, the right-hand side can be reformulated as follows.

\begin{eqnarray}
\{(P_1'\parallel\cdots\parallel P_n'\parallel P_{n+1}).(\alpha_1[m]\parallel\cdots\parallel \alpha_n[m]\parallel \alpha_{n+1}[m]): \nonumber\\
P_i\xtworightarrow{\alpha_i[m]}P_i',i\in\{1,\cdots,n+1\}\nonumber\\
+\sum\{(P_1\parallel\cdots\parallel P_i'\parallel\cdots\parallel P_j'\parallel\cdots\parallel P_n\parallel P_{n+1}).\tau: \nonumber\\
P_i\xtworightarrow{l[m]}P_i',P_j\xtworightarrow{\overline{l}[m]}P_j',i<j\} \nonumber\\
+\sum\{(P_1\parallel\cdots\parallel P_i'\parallel\cdots\parallel P_j'\parallel\cdots\parallel P_n\parallel P_{n+1}).\tau: \nonumber\\
P_i\xtworightarrow{l[m]}P_i',P_{n+1}\xtworightarrow{\overline{l}[m]}P_{n+1}',i\in\{1,\cdots, n\}\} \nonumber
\end{eqnarray}

So,

\begin{eqnarray}
R\sim_s^{r} \{(P_1'\parallel\cdots\parallel P_n'\parallel P_{n+1}').(\alpha_1[m]\parallel\cdots\parallel \alpha_n[m]\parallel \alpha_{n+1}[m]): \nonumber\\
P_i\xtworightarrow{\alpha_i[m]}P_i',i\in\{1,\cdots,n+1\}\nonumber\\
+\sum\{(P_1\parallel\cdots\parallel P_i'\parallel\cdots\parallel P_j'\parallel\cdots\parallel P_n).\tau: \nonumber\\
P_i\xtworightarrow{l[m]}P_i',P_j\xtworightarrow{\overline{l}[m]}P_j',1 \leq i<j\geq n+1\} \nonumber
\end{eqnarray}

Then, we can easily add the full conditions with Restriction and Relabeling.
\end{proof}

\begin{proposition}[New expansion law for strong pomset bisimulation]\label{NELSPB}
Let $P\equiv (P_1[f_1]\parallel\cdots\parallel P_n[f_n])\setminus L$, with $n\geq 1$. Then

\begin{eqnarray}
P\sim_p^{f} \{(f_1(\alpha_1)\parallel\cdots\parallel f_n(\alpha_n)).(P_1'[f_1]\parallel\cdots\parallel P_n'[f_n])\setminus L: \nonumber\\
P_i\xrightarrow{\alpha_i}P_i',i\in\{1,\cdots,n\},f_i(\alpha_i)\notin L\cup\overline{L}\} \nonumber\\
+\sum\{\tau.(P_1[f_1]\parallel\cdots\parallel P_i'[f_i]\parallel\cdots\parallel P_j'[f_j]\parallel\cdots\parallel P_n[f_n])\setminus L: \nonumber\\
P_i\xrightarrow{l_1}P_i',P_j\xrightarrow{l_2}P_j',f_i(l_1)=\overline{f_j(l_2)},i<j\}
\end{eqnarray}
\begin{eqnarray}
P\sim_p^{r} \{(P_1'[f_1]\parallel\cdots\parallel P_n'[f_n]).(f_1(\alpha_1[m])\parallel\cdots\parallel f_n(\alpha_n)[m])\setminus L: \nonumber\\
P_i\xtworightarrow{\alpha_i[m]}P_i',i\in\{1,\cdots,n\},f_i(\alpha_i)\notin L\cup\overline{L}\} \nonumber\\
+\sum\{(P_1[f_1]\parallel\cdots\parallel P_i'[f_i]\parallel\cdots\parallel P_j'[f_j]\parallel\cdots\parallel P_n[f_n]).\tau\setminus L: \nonumber\\
P_i\xtworightarrow{l_1[m]}P_i',P_j\xtworightarrow{l_2[m]}P_j',f_i(l_1)=\overline{f_j(l_2)},i<j\}
\end{eqnarray}
\end{proposition}

\begin{proof}
From the definition of strongly FR pomset bisimulation (see Definition \ref{FRPSB}), we know that strongly FR pomset bisimulation is defined by FR pomset transitions, which are labeled by pomsets. In an FR pomset transition, the events in the pomset are either within causality relations (defined by the prefix $.$) or in concurrency (implicitly defined by $.$ and $+$, and explicitly defined by $\parallel$), of course, they are pairwise consistent (without conflicts). In Proposition \ref{NELSSB}, we have already proven the case that all events are pairwise concurrent, so, we only need to prove the case of events in causality. Without loss of generality, we take a pomset of $p=\{\alpha,\beta:\alpha.\beta\}$. Then the FR pomset transition labeled by the above $p$ is just composed of one single event transition labeled by $\alpha$ succeeded by another single event transition labeled by $\beta$, that is, $\xrightarrow{p}=\xrightarrow{\alpha}\xrightarrow{\beta}$ and $\xtworightarrow{p[\mathcal{K}]}=\xtworightarrow{\beta[n]}\xtworightarrow{\alpha[m]}$.

Similarly to the proof of new expansion law for strongly FR step bisimulation (see Proposition \ref{NELSSB}), we can prove that the new expansion law holds for strongly FR pomset bisimulation, we omit them.
\end{proof}

\begin{proposition}[New expansion law for strong hp-bisimulation]\label{NELSHPB}
Let $P\equiv (P_1[f_1]\parallel\cdots\parallel P_n[f_n])\setminus L$, with $n\geq 1$. Then

\begin{eqnarray}
P\sim_{hp}^{f} \{(f_1(\alpha_1)\parallel\cdots\parallel f_n(\alpha_n)).(P_1'[f_1]\parallel\cdots\parallel P_n'[f_n])\setminus L: \nonumber\\
P_i\xrightarrow{\alpha_i}P_i',i\in\{1,\cdots,n\},f_i(\alpha_i)\notin L\cup\overline{L}\} \nonumber\\
+\sum\{\tau.(P_1[f_1]\parallel\cdots\parallel P_i'[f_i]\parallel\cdots\parallel P_j'[f_j]\parallel\cdots\parallel P_n[f_n])\setminus L: \nonumber\\
P_i\xrightarrow{l_1}P_i',P_j\xrightarrow{l_2}P_j',f_i(l_1)=\overline{f_j(l_2)},i<j\}
\end{eqnarray}
\begin{eqnarray}
P\sim_{hp}^{r} \{(P_1'[f_1]\parallel\cdots\parallel P_n'[f_n]).(f_1(\alpha_1[m])\parallel\cdots\parallel f_n(\alpha_n)[m])\setminus L: \nonumber\\
P_i\xtworightarrow{\alpha_i[m]}P_i',i\in\{1,\cdots,n\},f_i(\alpha_i)\notin L\cup\overline{L}\} \nonumber\\
+\sum\{(P_1[f_1]\parallel\cdots\parallel P_i'[f_i]\parallel\cdots\parallel P_j'[f_j]\parallel\cdots\parallel P_n[f_n]).\tau\setminus L: \nonumber\\
P_i\xtworightarrow{l_1[m]}P_i',P_j\xtworightarrow{l_2[m]}P_j',f_i(l_1)=\overline{f_j(l_2)},i<j\}
\end{eqnarray}
\end{proposition}

\begin{proof}
From the definition of strongly FR hp-bisimulation (see Definition \ref{FRHHPB}), we know that strongly FR hp-bisimulation is defined on the posetal product $(C_1,f,C_2),f:C_1\rightarrow C_2\textrm{ isomorphism}$. Two processes $P$ related to $C_1$ and $Q$ related to $C_2$, and $f:C_1\rightarrow C_2\textrm{ isomorphism}$. Initially, $(C_1,f,C_2)=(\emptyset,\emptyset,\emptyset)$, and $(\emptyset,\emptyset,\emptyset)\in\sim_{hp}$. When $P\xrightarrow{\alpha}P'$ ($C_1\xrightarrow{\alpha}C_1'$), there will be $Q\xrightarrow{\alpha}Q'$ ($C_2\xrightarrow{\alpha}C_2'$), and we define $f'=f[\alpha\mapsto \alpha]$. And when $P\xtworightarrow{\alpha[m]}P'$ ($C_1\xtworightarrow{\alpha[m]}C_1'$), there will be $Q\xtworightarrow{\alpha[m]}Q'$ ($C_2\xtworightarrow{\alpha[m]}C_2'$), and we define $f'=f[\alpha[m]\mapsto \alpha[m]]$. Then, if $(C_1,f,C_2)\in\sim_{hp}^{fr}$, then $(C_1',f',C_2')\in\sim_{hp}^{fr}$.

Similarly to the proof of new expansion law for strongly FR pomset bisimulation (see Proposition \ref{NELSPB}), we can prove that new expansion law holds for strongly FR hp-bisimulation, we just need additionally to check the above conditions on FR hp-bisimulation, we omit them.
\end{proof}

\begin{proposition}[New expansion law for strongly hhp-bisimulation]\label{NELSHHPB}
Let $P\equiv (P_1[f_1]\parallel\cdots\parallel P_n[f_n])\setminus L$, with $n\geq 1$. Then

\begin{eqnarray}
P\sim_{hhp}^{f} \{(f_1(\alpha_1)\parallel\cdots\parallel f_n(\alpha_n)).(P_1'[f_1]\parallel\cdots\parallel P_n'[f_n])\setminus L: \nonumber\\
P_i\xrightarrow{\alpha_i}P_i',i\in\{1,\cdots,n\},f_i(\alpha_i)\notin L\cup\overline{L}\} \nonumber\\
+\sum\{\tau.(P_1[f_1]\parallel\cdots\parallel P_i'[f_i]\parallel\cdots\parallel P_j'[f_j]\parallel\cdots\parallel P_n[f_n])\setminus L: \nonumber\\
P_i\xrightarrow{l_1}P_i',P_j\xrightarrow{l_2}P_j',f_i(l_1)=\overline{f_j(l_2)},i<j\}
\end{eqnarray}
\begin{eqnarray}
P\sim_{hhp}^{r} \{(P_1'[f_1]\parallel\cdots\parallel P_n'[f_n]).(f_1(\alpha_1[m])\parallel\cdots\parallel f_n(\alpha_n)[m])\setminus L: \nonumber\\
P_i\xtworightarrow{\alpha_i[m]}P_i',i\in\{1,\cdots,n\},f_i(\alpha_i)\notin L\cup\overline{L}\} \nonumber\\
+\sum\{(P_1[f_1]\parallel\cdots\parallel P_i'[f_i]\parallel\cdots\parallel P_j'[f_j]\parallel\cdots\parallel P_n[f_n]).\tau\setminus L: \nonumber\\
P_i\xtworightarrow{l_1[m]}P_i',P_j\xtworightarrow{l_2[m]}P_j',f_i(l_1)=\overline{f_j(l_2)},i<j\}
\end{eqnarray}
\end{proposition}

\begin{proof}
From the definition of strongly FR hhp-bisimulation (see Definition \ref{FRHHPB}), we know that strongly FR hhp-bisimulation is downward closed for strongly FR hp-bisimulation.

Similarly to the proof of new expansion law for strongly FR hp-bisimulation (see Proposition \ref{NELSHPB}), we can prove that new expansion law holds for strongly FR hhp-bisimulation, that is, they are downward closed for strongly FR hp-bisimulation, we omit them.
\end{proof}

\begin{theorem}[Congruence for strongly FR step bisimulation] \label{CSSB}
We can enjoy the congruence for strongly FR step bisimulation as follows.
\begin{enumerate}
  \item If $A\overset{\text{def}}{=}P$, then $A\sim_s^{fr} P$;
  \item Let $P_1\sim_s^{fr} P_2$. Then
        \begin{enumerate}
           \item $\alpha.P_1\sim_s^f \alpha.P_2$;
           \item $(\alpha_1\parallel\cdots\parallel\alpha_n).P_1\sim_s^f (\alpha_1\parallel\cdots\parallel\alpha_n).P_2$;
           \item $P_1.\alpha[m]\sim_s^r P_2.\alpha[m]$;
           \item $P_1.(\alpha_1[m]\parallel\cdots\parallel\alpha_n[m])\sim_s^r P_2.(\alpha_1[m]\parallel\cdots\parallel\alpha_n[m]).$;
           \item $P_1+Q\sim_s^{fr} P_2 +Q$;
           \item $P_1\parallel Q\sim_s^{fr} P_2\parallel Q$;
           \item $P_1\setminus L\sim_s^{fr} P_2\setminus L$;
           \item $P_1[f]\sim_s^{fr} P_2[f]$.
         \end{enumerate}
\end{enumerate}
\end{theorem}

\begin{proof}
Though transition rules in Definition \ref{semantics} are defined in the flavor of single event, they can be modified into a step (a set of events within which each event is pairwise concurrent), we omit them. If we treat a single event as a step containing just one event, the proof of the congruence does not exist any problem, so we use this way and still use the transition rules in Definition \ref{semantics}.

\begin{enumerate}
  \item If $A\overset{\text{def}}{=}P$, then $A\sim_s^{fr} P$. It is obvious.
  \item Let $P_1\sim_s^{fr} P_2$. Then
        \begin{enumerate}
           \item $\alpha.P_1\sim_s^f \alpha.P_2$. By the forward transition rules of Prefix, we can get

           $$\alpha.P_1\xrightarrow{\alpha}\alpha[m]P_1$$

           $$\alpha.P_2\xrightarrow{\alpha}\alpha[m]P_2$$

           Since $P_1\sim_s^{fr} P_2$, we get $\alpha.P_1\sim_s^f \alpha.P_2$, as desired.
           \item $(\alpha_1\parallel\cdots\parallel\alpha_n).P_1\sim_s^f (\alpha_1\parallel\cdots\parallel\alpha_n).P_2$. By the transition rules of Prefix, we can get

           $$(\alpha_1\parallel\cdots\parallel\alpha_n).P_1\xrightarrow{\{\alpha_1,\cdots,\alpha_n\}}(\alpha_1[m]\parallel\cdots\parallel\alpha_n[m]).P_1$$

           $$(\alpha_1\parallel\cdots\parallel\alpha_n).P_2\xrightarrow{\{\alpha_1,\cdots,\alpha_n\}}(\alpha_1[m]\parallel\cdots\parallel\alpha_n[m]).P_2$$

           Since $P_1\sim_s^{fr} P_2$, we get $(\alpha_1\parallel\cdots\parallel\alpha_n).P_1\sim_s^f (\alpha_1\parallel\cdots\parallel\alpha_n).P_2$, as desired.
           \item $P_1.\alpha[m]\sim_s^r P_2.\alpha[m]$. By the reverse transition rules of Prefix, we can get

           $$P_1.\alpha[m]\xtworightarrow{\alpha[m]}P_1.\alpha$$

           $$P_2.\alpha[m]\xtworightarrow{\alpha[m]}P_2.\alpha$$

           Since $P_1\sim_s^{fr} P_2$, we get $P_1.\alpha[m]\sim_s^r P_2.\alpha[m]$, as desired.
           \item $P_1.(\alpha_1[m]\parallel\cdots\parallel\alpha_n[m])\sim_s^r P_2.(\alpha_1[m]\parallel\cdots\parallel\alpha_n[m])$. By the reverse transition rules of Prefix, we can get

           $$P_1.(\alpha_1[m]\parallel\cdots\parallel\alpha_n[m])\xtworightarrow{\{\alpha_1[m],\cdots,\alpha_n[m]\}}P_1.(\alpha_1\parallel\cdots\parallel\alpha_n)$$

           $$P_2.(\alpha_1[m]\parallel\cdots\parallel\alpha_n[m])\xtworightarrow{\{\alpha_1[m],\cdots,\alpha_n[m]\}}P_2.(\alpha_1\parallel\cdots\parallel\alpha_n)$$

           Since $P_1\sim_s^{fr} P_2$, we get $(\alpha_1\parallel\cdots\parallel\alpha_n).P_1\sim_s^r (\alpha_1\parallel\cdots\parallel\alpha_n).P_2$, as desired.
           \item $P_1+Q\sim_s^{fr} P_2 +Q$. There are several cases, we will not enumerate all. By the forward transition rules of Summation, we can get

           $$\frac{P_1\xrightarrow{\alpha}P_1'}{P_2\xrightarrow{\alpha}P_2'}(P_1'\sim_s^{fr} P_2')$$

           $$\frac{P_1\xrightarrow{\alpha}P_1'\quad\alpha\notin Q}{P_1+Q\xrightarrow{\alpha}P_1'+Q}
           \quad \frac{P_2\xrightarrow{\alpha}P_2'\quad\alpha\notin Q}{P_2+Q\xrightarrow{\alpha}P_2'+Q}$$

           $$\frac{Q\xrightarrow{\beta}Q'\quad\beta\notin P_1}{P_1+Q\xrightarrow{\beta}P_1+Q'}
           \quad \frac{Q\xrightarrow{\beta}Q'\quad\beta\notin P_2}{P_2+Q\xrightarrow{\beta}P_2+Q'}$$

           By the reverse transition rules of Summation, we can get

           $$\frac{P_1\xtworightarrow{\alpha[m]}P_1'}{P_2\xtworightarrow{\alpha[m]}P_2'}(P_1'\sim_s^{fr} P_2')$$

           $$\frac{P_1\xtworightarrow{\alpha[m]}P_1'\quad\alpha\notin Q}{P_1+Q\xtworightarrow{\alpha[m]}P_1'+Q}
           \quad \frac{P_2\xtworightarrow{\alpha[m]}P_2'\quad\alpha\notin Q}{P_2+Q\xtworightarrow{\alpha[m]}P_2'+Q}$$

           $$\frac{Q\xtworightarrow{\beta[n]}Q'\quad\beta\notin P_1}{P_1+Q\xtworightarrow{\beta[n]}P_1+Q'}
           \quad \frac{Q\xtworightarrow{\beta[n]}Q'\quad\beta\notin P_2}{P_2+Q\xtworightarrow{\beta[n]}P_2+Q'}$$

           With the assumptions $P_1'+\sim_s^{fr} P_2'+Q$ and $P_1+Q'\sim_s^{fr} P_2+Q'$, we get $P_1+Q\sim_s^{fr} P_2+Q$, as desired.
           \item $P_1\parallel Q\sim_s^{fr} P_2\parallel Q$. There are several cases, we will not enumerate all. By the forward transition rules of Composition, we can get

           $$\frac{P_1\xrightarrow{\alpha}P_1'}{P_2\xrightarrow{\alpha}P_2'}(P_1'\sim_s^{fr} P_2')$$

           $$\frac{P_1\xrightarrow{\alpha}P_1'\quad Q\nrightarrow}{P_1\parallel Q\xrightarrow{\alpha}P_1'\parallel Q}
           \quad \frac{P_2\xrightarrow{\alpha}P_2'\quad Q\nrightarrow}{P_2\parallel Q\xrightarrow{\alpha}P_2'\parallel Q}$$

           $$\frac{Q\xrightarrow{\beta}Q'\quad P_1\nrightarrow}{P_1\parallel Q\xrightarrow{\beta}P_1\parallel Q'}
           \quad \frac{Q\xrightarrow{\beta}P_2'\quad P_2\nrightarrow}{P_2\parallel Q\xrightarrow{\beta}P_2\parallel Q'}$$

           $$\frac{P_1\xrightarrow{\alpha}P_1'\quad Q\xrightarrow{\beta}Q'}{P_1\parallel Q\xrightarrow{\{\alpha,\beta\}}P_1'\parallel Q'}(\beta\neq\overline{\alpha})
           \quad \frac{P_2\xrightarrow{\alpha}P_2'\quad Q\xrightarrow{\beta}Q'}{P_2\parallel Q\xrightarrow{\{\alpha,\beta\}}P_2'\parallel Q'}(\beta\neq\overline{\alpha})$$

           $$\frac{P_1\xrightarrow{l}P_1'\quad Q\xrightarrow{\overline{l}}Q'}{P_1\parallel Q\xrightarrow{\tau}P_1'\parallel Q'}
           \quad \frac{P_2\xrightarrow{l}P_2'\quad Q\xrightarrow{\overline{l}}Q'}{P_2\parallel Q\xrightarrow{\tau}P_2'\parallel Q'}$$

           By the reverse transition rules of Composition, we can get

           $$\frac{P_1\xtworightarrow{\alpha[m]}P_1'}{P_2\xtworightarrow{\alpha[m]}P_2'}(P_1'\sim_s^{fr} P_2')$$

           $$\frac{P_1\xtworightarrow{\alpha[m]}P_1'\quad Q\xntworightarrow{}}{P_1\parallel Q\xtworightarrow{\alpha[m]}P_1'\parallel Q}
           \quad \frac{P_2\xtworightarrow{\alpha[m]}P_2'\quad Q\xntworightarrow{}}{P_2\parallel Q\xtworightarrow{\alpha[m]}P_2'\parallel Q}$$

           $$\frac{Q\xtworightarrow{\beta[n]}Q'\quad P_1\xntworightarrow{}}{P_1\parallel Q\xtworightarrow{\beta[n]}P_1\parallel Q'}
           \quad \frac{Q\xtworightarrow{\beta[n]}P_2'\quad P_2\xntworightarrow{}}{P_2\parallel Q\xtworightarrow{\beta[n]}P_2\parallel Q'}$$

           $$\frac{P_1\xtworightarrow{\alpha[m]}P_1'\quad Q\xtworightarrow{\beta[m]}Q'}{P_1\parallel Q\xtworightarrow{\{\alpha[m],\beta[m]\}}P_1'\parallel Q'}(\beta\neq\overline{\alpha})
           \quad \frac{P_2\xtworightarrow{\alpha[m]}P_2'\quad Q\xtworightarrow{\beta[m]}Q'}{P_2\parallel Q\xtworightarrow{\{\alpha[m],\beta[m]\}}P_2'\parallel Q'}(\beta\neq\overline{\alpha})$$

           $$\frac{P_1\xtworightarrow{l[m]}P_1'\quad Q\xtworightarrow{\overline{l}[m]}Q'}{P_1\parallel Q\xtworightarrow{\tau}P_1'\parallel Q'}
           \quad \frac{P_2\xtworightarrow{l[m]}P_2'\quad Q\xtworightarrow{\overline{l}[m]}Q'}{P_2\parallel Q\xtworightarrow{\tau}P_2'\parallel Q'}$$

           Since $P_1'\sim_s^{fr} P_2'$ and $Q'\sim_s^{fr} Q'$, and with the assumptions $P_1'\parallel Q\sim_s^{fr} P_2'\parallel Q$, $P_1\parallel Q'\sim_s^{fr} P_2\parallel Q'$ and $P_1'\parallel Q'\sim_s^{fr} P_2'\parallel Q'$, we get $P_1\parallel Q\sim_s^{fr} P_2\parallel Q$, as desired.
           \item $P_1\setminus L\sim_s^{fr} P_2\setminus L$. There are several cases, we will not enumerate all. By the forward transition rules of Restriction, we get

           $$\frac{P_1\xrightarrow{\alpha}P_1'}{P_2\xrightarrow{\alpha}P_2'}(P_1'\sim_s^{fr} P_2')$$

           $$\frac{P_1\xrightarrow{\alpha}P_1'}{P_1\setminus L\xrightarrow{\alpha}P_1'\setminus L}$$

           $$\frac{P_2\xrightarrow{\alpha}P_2'}{P_2\setminus L\xrightarrow{\alpha}P_2'\setminus L}$$

           By the reverse transition rules of Restriction, we get

           $$\frac{P_1\xtworightarrow{\alpha[m]}P_1'}{P_2\xtworightarrow{\alpha[m]}P_2'}(P_1'\sim_s^{fr} P_2')$$

           $$\frac{P_1\xtworightarrow{\alpha[m]}P_1'}{P_1\setminus L\xtworightarrow{\alpha[m]}P_1'\setminus L}$$

           $$\frac{P_2\xtworightarrow{\alpha[m]}P_2'}{P_2\setminus L\xtworightarrow{\alpha[m]}P_2'\setminus L}$$

           Since $P_1'\sim_s^{fr} P_2'$, and with the assumption $P_1'\setminus L\sim_s^{fr} P_2'\setminus L$, we get $P_1\setminus L\sim_s^{fr} P_2\setminus L$, as desired.
           \item $P_1[f]\sim_s^{fr} P_2[f]$. By the forward transition rules of Relabelling, we get

           $$\frac{P_1\xrightarrow{\alpha}P_1'}{P_2\xrightarrow{\alpha}P_2'}(P_1'\sim_s^{fr} P_2')$$

           $$\frac{P_1\xrightarrow{\alpha}P_1'}{P_1[f]\xrightarrow{f(\alpha)}P_1'[f]}$$

           $$\frac{P_2\xrightarrow{\alpha}P_2'}{P_2[f]\xrightarrow{f(\alpha)}P_2'[f]}$$

           By the reverse transition rules of Relabelling, we get

           $$\frac{P_1\xtworightarrow{\alpha[m]}P_1'}{P_2\xtworightarrow{\alpha[m]}P_2'}(P_1'\sim_s^{fr} P_2')$$

           $$\frac{P_1\xtworightarrow{\alpha[m]}P_1'}{P_1[f]\xtworightarrow{f(\alpha)[m]}P_1'[f]}$$

           $$\frac{P_2\xtworightarrow{\alpha[m]}P_2'}{P_2[f]\xtworightarrow{f(\alpha)[m]}P_2'[f]}$$

           Since $P_1'\sim_s^{fr} P_2'$, and with the assumption $P_1'[f]\sim_s^{fr} P_2'[f]$, we get $P_1[f]\sim_s^{fr} P_2[f]$, as desired.
         \end{enumerate}
\end{enumerate}
\end{proof}

\begin{theorem}[Congruence for strongly FR pomset bisimulation] \label{CSPB}
We can enjoy the congruence for strongly FR pomset bisimulation as follows.
\begin{enumerate}
  \item If $A\overset{\text{def}}{=}P$, then $A\sim_p^{fr} P$;
  \item Let $P_1\sim_p^{fr} P_2$. Then
        \begin{enumerate}
           \item $\alpha.P_1\sim_p^f \alpha.P_2$;
           \item $(\alpha_1\parallel\cdots\parallel\alpha_n).P_1\sim_p^f (\alpha_1\parallel\cdots\parallel\alpha_n).P_2$;
           \item $P_1.\alpha[m]\sim_p^r P_2.\alpha[m]$;
           \item $P_1.(\alpha_1[m]\parallel\cdots\parallel\alpha_n[m])\sim_p^r P_2.(\alpha_1[m]\parallel\cdots\parallel\alpha_n[m]).$;
           \item $P_1+Q\sim_p^{fr} P_2 +Q$;
           \item $P_1\parallel Q\sim_p^{fr} P_2\parallel Q$;
           \item $P_1\setminus L\sim_p^{fr} P_2\setminus L$;
           \item $P_1[f]\sim_p^{fr} P_2[f]$.
         \end{enumerate}
\end{enumerate}
\end{theorem}

\begin{proof}
From the definition of strongly FR pomset bisimulation (see Definition \ref{FRPSB}), we know that strongly FR pomset bisimulation is defined by FR pomset transitions, which are labeled by pomsets. In an FR pomset transition, the events in the pomset are either within causality relations (defined by the prefix $.$) or in concurrency (implicitly defined by $.$ and $+$, and explicitly defined by $\parallel$), of course, they are pairwise consistent (without conflicts). In Theorem \ref{CSSB}, we have already proven the case that all events are pairwise concurrent, so, we only need to prove the case of events in causality. Without loss of generality, we take a pomset of $p=\{\alpha,\beta:\alpha.\beta\}$. Then the FR pomset transition labeled by the above $p$ is just composed of one single event transition labeled by $\alpha$ succeeded by another single event transition labeled by $\beta$, that is, $\xrightarrow{p}=\xrightarrow{\alpha}\xrightarrow{\beta}$ and $\xtworightarrow{p[\mathcal{K}]}=\xtworightarrow{\beta[n]}\xtworightarrow{\alpha[m]}$.

Similarly to the proof of congruence for strongly FR step bisimulation (see Theorem \ref{CSSB}), we can prove that the congruence holds for strongly FR pomset bisimulation, we omit them.
\end{proof}

\begin{theorem}[Congruence for strongly FR hp-bisimulation] \label{CSHPB}
We can enjoy the congruence for strongly FR hp-bisimulation as follows.
\begin{enumerate}
  \item If $A\overset{\text{def}}{=}P$, then $A\sim_{hp}^{fr} P$;
  \item Let $P_1\sim_{hp}^{fr} P_2$. Then
        \begin{enumerate}
           \item $\alpha.P_1\sim_{hp}^f \alpha.P_2$;
           \item $(\alpha_1\parallel\cdots\parallel\alpha_n).P_1\sim_{hp}^f (\alpha_1\parallel\cdots\parallel\alpha_n).P_2$;
           \item $P_1.\alpha[m]\sim_{hp}^r P_2.\alpha[m]$;
           \item $P_1.(\alpha_1[m]\parallel\cdots\parallel\alpha_n[m])\sim_{hp}^r P_2.(\alpha_1[m]\parallel\cdots\parallel\alpha_n[m]).$;
           \item $P_1+Q\sim_{hp}^{fr} P_2 +Q$;
           \item $P_1\parallel Q\sim_{hp}^{fr} P_2\parallel Q$;
           \item $P_1\setminus L\sim_{hp}^{fr} P_2\setminus L$;
           \item $P_1[f]\sim_{hp}^{fr} P_2[f]$.
         \end{enumerate}
\end{enumerate}
\end{theorem}

\begin{proof}
From the definition of strongly FR hp-bisimulation (see Definition \ref{FRHHPB}), we know that strongly FR hp-bisimulation is defined on the posetal product $(C_1,f,C_2),f:C_1\rightarrow C_2\textrm{ isomorphism}$. Two processes $P$ related to $C_1$ and $Q$ related to $C_2$, and $f:C_1\rightarrow C_2\textrm{ isomorphism}$. Initially, $(C_1,f,C_2)=(\emptyset,\emptyset,\emptyset)$, and $(\emptyset,\emptyset,\emptyset)\in\sim_{hp}$. When $P\xrightarrow{\alpha}P'$ ($C_1\xrightarrow{\alpha}C_1'$), there will be $Q\xrightarrow{\alpha}Q'$ ($C_2\xrightarrow{\alpha}C_2'$), and we define $f'=f[\alpha\mapsto \alpha]$. And when $P\xtworightarrow{\alpha[m]}P'$ ($C_1\xtworightarrow{\alpha[m]}C_1'$), there will be $Q\xtworightarrow{\alpha[m]}Q'$ ($C_2\xtworightarrow{\alpha[m]}C_2'$), and we define $f'=f[\alpha[m]\mapsto \alpha[m]]$. Then, if $(C_1,f,C_2)\in\sim_{hp}^{fr}$, then $(C_1',f',C_2')\in\sim_{hp}^{fr}$.

Similarly to the proof of congruence for strongly FR pomset bisimulation (see Theorem \ref{CSPB}), we can prove that the congruence holds for strongly FR hp-bisimulation, we just need additionally to check the above conditions on FR hp-bisimulation, we omit them.
\end{proof}

\begin{theorem}[Congruence for strongly FR hhp-bisimulation] \label{CSHHPB}
We can enjoy the congruence for strongly FR hhp-bisimulation as follows.
\begin{enumerate}
  \item If $A\overset{\text{def}}{=}P$, then $A\sim_{hhp}^{fr} P$;
  \item Let $P_1\sim_{hhp}^{fr} P_2$. Then
        \begin{enumerate}
           \item $\alpha.P_1\sim_{hhp}^f \alpha.P_2$;
           \item $(\alpha_1\parallel\cdots\parallel\alpha_n).P_1\sim_{hhp}^f (\alpha_1\parallel\cdots\parallel\alpha_n).P_2$;
           \item $P_1.\alpha[m]\sim_{hhp}^r P_2.\alpha[m]$;
           \item $P_1.(\alpha_1[m]\parallel\cdots\parallel\alpha_n[m])\sim_{hhp}^r P_2.(\alpha_1[m]\parallel\cdots\parallel\alpha_n[m]).$;
           \item $P_1+Q\sim_{hhp}^{fr} P_2 +Q$;
           \item $P_1\parallel Q\sim_{hhp}^{fr} P_2\parallel Q$;
           \item $P_1\setminus L\sim_{hhp}^{fr} P_2\setminus L$;
           \item $P_1[f]\sim_{hhp}^{fr} P_2[f]$.
         \end{enumerate}
\end{enumerate}
\end{theorem}

\begin{proof}
From the definition of strongly FR hhp-bisimulation (see Definition \ref{FRHHPB}), we know that strongly FR hhp-bisimulation is downward closed for strongly FR hp-bisimulation.

Similarly to the proof of congruence for strongly FR hp-bisimulation (see Theorem \ref{CSHPB}), we can prove that the congruence holds for strongly FR hhp-bisimulation, we omit them.
\end{proof}

\begin{definition}[Weakly guarded recursive expression]
$X$ is weakly guarded in $E$ if each occurrence of $X$ is with some subexpression $\alpha.F$ or $(\alpha_1\parallel\cdots\parallel\alpha_n).F$ or $F.\alpha[m]$ or $F.(\alpha_1[m]\parallel\cdots\parallel\alpha_n[m])$ of $E$.
\end{definition}

\begin{proposition}\label{LUS}
If the variables $\widetilde{X}$ are weakly guarded in $E$, and $E\{\widetilde{P}/\widetilde{X}\}\xrightarrow{\{\alpha_1,\cdots,\alpha_n\}}P'$, or $E\{\widetilde{P}/\widetilde{X}\}\xtworightarrow{\{\alpha_1[m],\cdots,\alpha_n[m]\}}P'$, then $P'$ can not takes the form $E'\{\widetilde{P}/\widetilde{X}\}$ for some expression $E'$.
\end{proposition}
\begin{proof}
It needs to induct on the depth of the inference of $E\{\widetilde{P}/\widetilde{X}\}\xrightarrow{\{\alpha_1,\cdots,\alpha_n\}}P'$ or $E\{\widetilde{P}/\widetilde{X}\}\xtworightarrow{\{\alpha_1[m],\cdots,\alpha_n[m]\}}P'$. We consider $E\{\widetilde{P}/\widetilde{X}\}\xrightarrow{\{\alpha_1,\cdots,\alpha_n\}}P'$.

Case $E\equiv E_1+E_2$. We may have $E_1\xrightarrow{e_1}e_1[m]\cdot E_1'\quad e_1\notin E_2$, $E_1+E_2\xrightarrow{e_1}e_1[m]\cdot E_1'+E_2$, $e_1[m]\cdot E_1'+E_2$ can not takes the form $E'\{\widetilde{P}/\widetilde{X}\}$ for some expression $E'$.

So, there may be not recursive expression for strongly FR truly concurrent bisimulations. For the same reason, there also may be not recursive expression for weakly FR truly concurrent bisimulations.
\end{proof}

%% file: section6.tex
\section{Weakly Forward-reverse Truly Concurrent Bisimulations}\label{wftcb}

Remembering that $\tau$ can neither be restricted nor relabeled, we know that the monoid laws, the static laws and the new expansion law in section \ref{sftcb} still hold with respect to the corresponding weakly FR truly concurrent bisimulations. And also, we can enjoy the congruence of Prefix, Summation, Composition, Restriction, Relabelling and Constants with respect to corresponding weakly FR truly concurrent bisimulations. We will not retype these laws, and just give the $\tau$-specific laws. The forward and reverse transition rules of $\tau$ are shown in Table \ref{TRForTAU}, where $\xrightarrow{\tau}\surd$ is a predicate which represents a successful termination after execution of the silent step $\tau$.

\begin{center}
    \begin{table}
        $$\frac{}{\tau\xrightarrow{\tau}\surd}$$
        $$\frac{}{\tau\xtworightarrow{\tau}\surd}$$
        \caption{Forward and reverse transition rules of $\tau$}
        \label{TRForTAU}
    \end{table}
\end{center}

\begin{proposition}[$\tau$ laws for weakly FR step bisimulation]\label{TAUWSB}
The $\tau$ laws for weakly FR step bisimulation is as follows.
\begin{enumerate}
  \item $P\approx_s^f \tau.P$;
  \item $P\approx_s^r P.\tau$;
  \item $\alpha.\tau.P\approx_s^f \alpha.P$;
  \item $P.\tau.\alpha[m]\approx_s^r P.\alpha[m]$;
  \item $(\alpha_1\parallel\cdots\parallel\alpha_n).\tau.P\approx_s^f (\alpha_1\parallel\cdots\parallel\alpha_n).P$;
  \item $P.\tau.(\alpha_1[m]\parallel\cdots\parallel\alpha_n[m])\approx_s^r P.(\alpha_1[m]\parallel\cdots\parallel\alpha_n[m])$;
  \item $P+\tau.P\approx_s^f \tau.P$;
  \item $P+P.\tau\approx_s^r P.\tau$;
  \item $\alpha.(\tau.(P+Q)+P)\approx_s^f\alpha.(P+Q)$;
  \item $((P+Q).\tau+P).\alpha[m]\approx_s^r(P+Q).\alpha[m]$;
  \item $(\alpha_1\parallel\cdots\parallel\alpha_n).(\tau.(P+Q)+P)\approx_s^f(\alpha_1\parallel\cdots\parallel\alpha_n).(P+Q)$;
  \item $((P+Q).\tau+P).(\alpha_1[m]\parallel\cdots\parallel\alpha_n[m])\approx_s^r(P+Q).(\alpha_1[m]\parallel\cdots\parallel\alpha_n[m])$;
  \item $P\approx_s^{fr} \tau\parallel P$.
\end{enumerate}
\end{proposition}

\begin{proof}
Though transition rules in Definition \ref{semantics} are defined in the flavor of single event, they can be modified into a step (a set of events within which each event is pairwise concurrent), we omit them. If we treat a single event as a step containing just one event, the proof of $\tau$ laws does not exist any problem, so we use this way and still use the transition rules in Definition \ref{semantics}.

\begin{enumerate}
  \item $P\approx_s^f \tau.P$. By the forward transition rules of Prefix, we get

  $$\frac{P\xRightarrow{\alpha}P'}{P\xRightarrow{\alpha}P'}
  \quad \frac{P\xRightarrow{\alpha}P'}{\tau.P\xRightarrow{\alpha} P'}$$

  Since $P'\approx_s^f P'$, we get $P\approx_s^f \tau.P$, as desired.
  \item $P\approx_s^r P.\tau$. By the reverse transition rules of Prefix, we get

  $$\frac{P\xTworightarrow{\alpha[m]}P'}{P\xTworightarrow{\alpha[m]}P'}
  \quad \frac{P\xTworightarrow{\alpha[m]}P'}{\tau.P\xTworightarrow{\alpha[m]} P'}$$

  Since $P'\approx_s^r P'$, we get $P\approx_s^r P.\tau$, as desired.
  \item $\alpha.\tau.P\approx_s^f \alpha.P$. By the forward transition rules of Prefix, we get

  $$\frac{}{\alpha.\tau.P\xRightarrow{\alpha}\alpha[m].P}
  \quad \frac{}{\alpha.P\xRightarrow{\alpha}\alpha[m].P}$$

  Since $\alpha[m].P\approx_s^f \alpha[m].P$, we get $\alpha.\tau.P\approx_s^f \alpha.P$, as desired.
  \item $P.\tau.\alpha[m]\approx_s^r P.\alpha[m]$. By the reverse transition rules of Prefix, we get

  $$\frac{}{P.\tau.\alpha[m]\xTworightarrow{\alpha[m]}P.\alpha}
  \quad \frac{}{P.\alpha[m]\xTworightarrow{\alpha[m]}P.\alpha}$$

  Since $P.\alpha\approx_s^r P.\alpha$, we get $P.\tau.\alpha[m]\approx_s^r P.\alpha[m]$, as desired.
  \item $(\alpha_1\parallel\cdots\parallel\alpha_n).\tau.P\approx_s^f (\alpha_1\parallel\cdots\parallel\alpha_n).P$. By the forward transition rules of Prefix, we get

  $$\frac{}{(\alpha_1\parallel\cdots\parallel\alpha_n).\tau.P\xRightarrow{\{\alpha_1,\cdots,\alpha_n\}}(\alpha_1[m]\parallel\cdots\parallel\alpha_n[m]).P}
  \quad \frac{}{(\alpha_1\parallel\cdots\parallel\alpha_n).P\xRightarrow{\{\alpha_1,\cdots,\alpha_n\}}(\alpha_1[m]\parallel\cdots\parallel\alpha_n[m]).P}$$

  Since $(\alpha_1[m]\parallel\cdots\parallel\alpha_n[m]).P\approx_s^f (\alpha_1[m]\parallel\cdots\parallel\alpha_n[m]).P$, we get $(\alpha_1\parallel\cdots\parallel\alpha_n).\tau.P\approx_s^f (\alpha_1\parallel\cdots\parallel\alpha_n).P$, as desired.
  \item $P.\tau.(\alpha_1[m]\parallel\cdots\parallel\alpha_n[m])\approx_s^r P.(\alpha_1[m]\parallel\cdots\parallel\alpha_n[m])$. By the reverse transition rules of Prefix, we get

  $$\frac{}{P.\tau.(\alpha_1[m]\parallel\cdots\parallel\alpha_n[m])\xTworightarrow{\{\alpha_1[m],\cdots,\alpha_n[m]\}}(P.(\alpha_1\parallel\cdots\parallel\alpha_n)}$$
  $$\frac{}{P.(\alpha_1[m]\parallel\cdots\parallel\alpha_n[m])\xTworightarrow{\{\alpha_1[m],\cdots,\alpha_n[m]\}}(P.(\alpha_1\parallel\cdots\parallel\alpha_n)}$$

  Since $P.(\alpha_1\parallel\cdots\parallel\alpha_n)\approx_s^r P.(\alpha_1\parallel\cdots\parallel\alpha_n)$, we get $P.\tau.(\alpha_1[m]\parallel\cdots\parallel\alpha_n[m])\approx_s^r P.(\alpha_1[m]\parallel\cdots\parallel\alpha_n[m])$, as desired.
  \item $P+\tau.P\approx_s^f \tau.P$. By the forward transition rules of Summation, we get

  $$\frac{P\xRightarrow{\alpha}P'}{P+\tau.P\xRightarrow{\alpha}P'+P'}
  \quad \frac{P\xRightarrow{\alpha}P'}{\tau.P\xRightarrow{\alpha} P'}$$

  Since $P'+P'\approx_s^f P'$, we get $P+\tau.P\approx_s^f \tau.P$, as desired.
  \item $P+P.\tau\approx_s^r P.\tau$. By the reverse transition rules of Summation, we get

  $$\frac{P\xTworightarrow{\alpha[m]}P'}{P+P.\tau\xTworightarrow{\alpha[m]}P'+P'}
  \quad \frac{P\xTworightarrow{\alpha[m]}P'}{P.\tau\xTworightarrow{\alpha[m]} P'}$$

  Since $P'+P'\approx_s^r P'$, we get $P+P.\tau\approx_s^r P.\tau$, as desired.
  \item $\alpha.(\tau.(P+Q)+P)\approx_s^f\alpha.(P+Q)$. By the forward transition rules of Prefix and Summation, we get

  $$\frac{}{\alpha.(\tau.(P+Q)+P)\xRightarrow{\alpha}\alpha[m].(P+Q+P)}
  \quad \frac{}{\alpha.(P+Q)\xRightarrow{\alpha} \alpha[m].(P+Q)}$$

  Since $\alpha[m].(P+Q+P)\approx_s^f \alpha[m].(P+Q)$, we get $\alpha.(\tau.(P+Q)+P)\approx_s^f\alpha.(P+Q)$, as desired.
  \item $((P+Q).\tau+P).\alpha[m]\approx_s^r(P+Q).\alpha[m]$. By the reverse transition rules of Prefix and Summation, we get

  $$\frac{}{((P+Q).\tau+P).\alpha[m]\xTworightarrow{\alpha[m]}(P+Q+P).\alpha}
  \quad \frac{}{(P+Q).\alpha[m]\xTworightarrow{\alpha[m]} (P+Q).\alpha}$$

  Since $(P+Q+P).\alpha\approx_s^r (P+Q).\alpha$, we get $((P+Q).\tau+P).\alpha[m]\approx_s^r(P+Q).\alpha[m]$, as desired.
  \item $(\alpha_1\parallel\cdots\parallel\alpha_n).(\tau.(P+Q)+P)\approx_s^f(\alpha_1\parallel\cdots\parallel\alpha_n).(P+Q)$. By the forward transition rules of Prefix and Summation, we get

  $$\frac{}{(\alpha_1\parallel\cdots\parallel\alpha_n).(\tau.(P+Q)+P) \xRightarrow{\{\alpha_1,\cdots,\alpha_n\}}(\alpha_1[m]\parallel\cdots\parallel\alpha_n[m]).(P+Q+P)}$$
  $$\frac{}{(\alpha_1\parallel\cdots\parallel\alpha_n).(P+Q)\xRightarrow{\{\alpha_1,\cdots,\alpha_n\}} (\alpha_1[m]\parallel\cdots\parallel\alpha_n[m]).(P+Q)}$$

  Since $(\alpha_1[m]\parallel\cdots\parallel\alpha_n[m]).(P+Q+P)\approx_s^f (\alpha_1[m]\parallel\cdots\parallel\alpha_n[m]).(P+Q)$, we get $(\alpha_1\parallel\cdots\parallel\alpha_n).(\tau.(P+Q)+P)\approx_s^f(\alpha_1\parallel\cdots\parallel\alpha_n).(P+Q)$, as desired.
  \item $((P+Q).\tau+P).(\alpha_1[m]\parallel\cdots\parallel\alpha_n[m])\approx_s^r(P+Q).(\alpha_1[m]\parallel\cdots\parallel\alpha_n[m])$. By the reverse transition rules of Prefix and Summation, we get

  $$\frac{}{((P+Q).\tau+P).(\alpha_1[m]\parallel\cdots\parallel\alpha_n[m]) \xTworightarrow{\{\alpha_1[m],\cdots,\alpha_n[m]\}}(P+Q+P).(\alpha_1\parallel\cdots\parallel\alpha_n)}$$
  $$\frac{}{(P+Q).(\alpha_1[m]\parallel\cdots\parallel\alpha_n[m])\xTworightarrow{\{\alpha_1[m],\cdots,\alpha_n[m]\}}(P+Q).(\alpha_1\parallel\cdots\parallel\alpha_n)}$$

  Since $(P+Q+P).(\alpha_1\parallel\cdots\parallel\alpha_n)\approx_s^r (P+Q).(\alpha_1\parallel\cdots\parallel\alpha_n)$, we get $((P+Q).\tau+P).(\alpha_1[m]\parallel\cdots\parallel\alpha_n[m])\approx_s^r(P+Q).(\alpha_1[m]\parallel\cdots\parallel\alpha_n[m])$, as desired.
  \item $P\approx_s^{fr} \tau\parallel P$. By the forward transition rules of Composition, we get

  $$\frac{P\xRightarrow{\alpha}P'}{P\xRightarrow{\alpha}P'}
  \quad \frac{P\xRightarrow{\alpha}P'}{\tau\parallel P\xRightarrow{\alpha} P'}$$

  By the reverse transition rules of Composition, we get

  $$\frac{P\xTworightarrow{\alpha}P'}{P\xTworightarrow{\alpha}P'}
  \quad \frac{P\xTworightarrow{\alpha}P'}{\tau\parallel P\xTworightarrow{\alpha} P'}$$

  Since $P'\approx_s^{fr} P'$, we get $P\approx_s^{fr} \tau\parallel P$, as desired.
\end{enumerate}
\end{proof}

\begin{proposition}[$\tau$ laws for weakly FR pomset bisimulation]\label{TAUWPB}
The $\tau$ laws for weakly FR pomset bisimulation is as follows.
\begin{enumerate}
  \item $P\approx_p^f \tau.P$;
  \item $P\approx_p^r P.\tau$;
  \item $\alpha.\tau.P\approx_p^f \alpha.P$;
  \item $P.\tau.\alpha[m]\approx_p^r P.\alpha[m]$;
  \item $(\alpha_1\parallel\cdots\parallel\alpha_n).\tau.P\approx_p^f (\alpha_1\parallel\cdots\parallel\alpha_n).P$;
  \item $P.\tau.(\alpha_1[m]\parallel\cdots\parallel\alpha_n[m])\approx_p^r P.(\alpha_1[m]\parallel\cdots\parallel\alpha_n[m])$;
  \item $P+\tau.P\approx_p^f \tau.P$;
  \item $P+P.\tau\approx_p^r P.\tau$;
  \item $\alpha.(\tau.(P+Q)+P)\approx_s^f\alpha.(P+Q)$;
  \item $((P+Q).\tau+P).\alpha[m]\approx_s^r(P+Q).\alpha[m]$;
  \item $(\alpha_1\parallel\cdots\parallel\alpha_n).(\tau.(P+Q)+P)\approx_s^f(\alpha_1\parallel\cdots\parallel\alpha_n).(P+Q)$;
  \item $((P+Q).\tau+P).(\alpha_1[m]\parallel\cdots\parallel\alpha_n[m])\approx_s^r(P+Q).(\alpha_1[m]\parallel\cdots\parallel\alpha_n[m])$;
  \item $P\approx_p^{fr} \tau\parallel P$.
\end{enumerate}
\end{proposition}

\begin{proof}
From the definition of weakly FR pomset bisimulation $\approx_{p}^{fr}$ (see Definition \ref{FRWPSB}), we know that weakly FR pomset bisimulation $\approx_{p}^{fr}$ is defined by weakly FR pomset transitions, which are labeled by pomsets with $\tau$. In a weakly FR pomset transition, the events in the pomset are either within causality relations (defined by $.$) or in concurrency (implicitly defined by $.$ and $+$, and explicitly defined by $\parallel$), of course, they are pairwise consistent (without conflicts). In Proposition \ref{TAUWSB}, we have already proven the case that all events are pairwise concurrent, so, we only need to prove the case of events in causality. Without loss of generality, we take a pomset of $p=\{\alpha,\beta:\alpha.\beta\}$. Then the weakly forward pomset transition labeled by the above $p$ is just composed of one single event transition labeled by $\alpha$ succeeded by another single event transition labeled by $\beta$, that is, $\xRightarrow{p}=\xRightarrow{\alpha}\xRightarrow{\beta}$ and $\xTworightarrow{p[\mathcal{K}]}=\xTworightarrow{\beta[n]}\xTworightarrow{\alpha[m]}$.

Similarly to the proof of $\tau$ laws for weakly FR step bisimulation $\approx_{s}^{fr}$ (Proposition \ref{TAUWSB}), we can prove that $\tau$ laws hold for weakly FR pomset bisimulation $\approx_{p}^{fr}$, we omit them.
\end{proof}

\begin{proposition}[$\tau$ laws for weakly FR hp-bisimulation]\label{TAUWHPB}
The $\tau$ laws for weakly FR hp-bisimulation is as follows.
\begin{enumerate}
  \item $P\approx_{hp}^f \tau.P$;
  \item $P\approx_{hp}^r P.\tau$;
  \item $\alpha.\tau.P\approx_{hp}^f \alpha.P$;
  \item $P.\tau.\alpha[m]\approx_{hp}^r P.\alpha[m]$;
  \item $(\alpha_1\parallel\cdots\parallel\alpha_n).\tau.P\approx_{hp}^f (\alpha_1\parallel\cdots\parallel\alpha_n).P$;
  \item $P.\tau.(\alpha_1[m]\parallel\cdots\parallel\alpha_n[m])\approx_{hp}^r P.(\alpha_1[m]\parallel\cdots\parallel\alpha_n[m])$;
  \item $P+\tau.P\approx_{hp}^f \tau.P$;
  \item $P+P.\tau\approx_{hp}^r P.\tau$;
  \item $\alpha.(\tau.(P+Q)+P)\approx_s^f\alpha.(P+Q)$;
  \item $((P+Q).\tau+P).\alpha[m]\approx_s^r(P+Q).\alpha[m]$;
  \item $(\alpha_1\parallel\cdots\parallel\alpha_n).(\tau.(P+Q)+P)\approx_s^f(\alpha_1\parallel\cdots\parallel\alpha_n).(P+Q)$;
  \item $((P+Q).\tau+P).(\alpha_1[m]\parallel\cdots\parallel\alpha_n[m])\approx_s^r(P+Q).(\alpha_1[m]\parallel\cdots\parallel\alpha_n[m])$;
  \item $P\approx_{hp}^{fr} \tau\parallel P$.
\end{enumerate}
\end{proposition}

\begin{proof}
From the definition of weakly FR hp-bisimulation $\approx_{hp}^{fr}$ (see Definition \ref{FRWHHPB}), we know that weakly FR hp-bisimulation $\approx_{hp}^{fr}$ is defined on the weakly posetal product $(C_1,f,C_2),f:\hat{C_1}\rightarrow \hat{C_2}\textrm{ isomorphism}$. Two processes $P$ related to $C_1$ and $Q$ related to $C_2$, and $f:\hat{C_1}\rightarrow \hat{C_2}\textrm{ isomorphism}$. Initially, $(C_1,f,C_2)=(\emptyset,\emptyset,\emptyset)$, and $(\emptyset,\emptyset,\emptyset)\in\approx_{hp}^{fr}$. When $P\xRightarrow{\alpha}P'$ ($C_1\xRightarrow{\alpha}C_1'$), there will be $Q\xRightarrow{\alpha}Q'$ ($C_2\xRightarrow{\alpha}C_2'$), and we define $f'=f[\alpha\mapsto \alpha]$. And when $P\xTworightarrow{\alpha[m]}P'$ ($C_1\xTworightarrow{\alpha[m]}C_1'$), there will be $Q\xTworightarrow{\alpha[m]}Q'$ ($C_2\xTworightarrow{\alpha[m]}C_2'$), and we define $f'=f[\alpha[m]\mapsto \alpha[m]]$. Then, if $(C_1,f,C_2)\in\approx_{hp}^{fr}$, then $(C_1',f',C_2')\in\approx_{hp}^{fr}$.

Similarly to the proof of $\tau$ laws for weakly FR pomset bisimulation (Proposition \ref{TAUWPB}), we can prove that $\tau$ laws hold for weakly FR hp-bisimulation, we just need additionally to check the above conditions on weakly FR hp-bisimulation, we omit them.
\end{proof}

\begin{proposition}[$\tau$ laws for weakly FR hhp-bisimulation]\label{TAUWHHPB}
The $\tau$ laws for weakly FR hhp-bisimulation is as follows.
\begin{enumerate}
  \item $P\approx_{hhp}^f \tau.P$;
  \item $P\approx_{hhp}^r P.\tau$;
  \item $\alpha.\tau.P\approx_{hhp}^f \alpha.P$;
  \item $P.\tau.\alpha[m]\approx_{hhp}^r P.\alpha[m]$;
  \item $(\alpha_1\parallel\cdots\parallel\alpha_n).\tau.P\approx_{hhp}^f (\alpha_1\parallel\cdots\parallel\alpha_n).P$;
  \item $P.\tau.(\alpha_1[m]\parallel\cdots\parallel\alpha_n[m])\approx_{hhp}^r P.(\alpha_1[m]\parallel\cdots\parallel\alpha_n[m])$;
  \item $P+\tau.P\approx_{hhp}^f \tau.P$;
  \item $P+P.\tau\approx_{hhp}^r P.\tau$;
  \item $\alpha.(\tau.(P+Q)+P)\approx_s^f\alpha.(P+Q)$;
  \item $((P+Q).\tau+P).\alpha[m]\approx_s^r(P+Q).\alpha[m]$;
  \item $(\alpha_1\parallel\cdots\parallel\alpha_n).(\tau.(P+Q)+P)\approx_s^f(\alpha_1\parallel\cdots\parallel\alpha_n).(P+Q)$;
  \item $((P+Q).\tau+P).(\alpha_1[m]\parallel\cdots\parallel\alpha_n[m])\approx_s^r(P+Q).(\alpha_1[m]\parallel\cdots\parallel\alpha_n[m])$;
  \item $P\approx_{hhp}^{fr} \tau\parallel P$.
\end{enumerate}
\end{proposition}

\begin{proof}
From the definition of weakly FR hhp-bisimulation (see Definition \ref{FRWHHPB}), we know that weakly FR hhp-bisimulation is downward closed for weakly FR hp-bisimulation.

Similarly to the proof of $\tau$ laws for weakly FR hp-bisimulation (see Proposition \ref{TAUWHPB}), we can prove that the $\tau$ laws hold for weakly FR hhp-bisimulation, we omit them.
\end{proof}

%% file: section7.tex
\section{Conclusions}{\label{con}}

We design a reversible version of truly concurrent process algebra CTC \cite{CTC}. It has good properties modulo several kinds of strongly FR truly concurrent bisimulations and weakly FR truly concurrent bisimulations. These properties include monoid laws, static laws, new expansion law for strongly FR truly concurrent bisimulations, $\tau$ laws for weakly FR truly concurrent bisimulations, and congruences for strongly and weakly FR truly concurrent bisimulations. It can be used in verification of computer systems with a truly concurrent and reversible flavor.

%% file: Reversible_Truly_Concurrent_Process_Algebra.bbl
\begin{thebibliography}{Lam94}

\bibitem{ALNC}M. Hennessy and R. Milner. Algebraic laws for nondeterminism and concurrency. J. ACM, 1985, 32, 137-161.

\bibitem{CC}R. Milner. Communication and concurrency. Printice Hall, 1989.

\bibitem{CCS} R. Milner. A calculus of communicating systems. LNCS 92, Springer, 1980.

\bibitem{ACP} W. Fokkink. Introduction to process algebra 2nd ed. Springer-Verlag, 2007.

\bibitem{RCCS2} I. Phillips, I. Ulidowski.: \emph{Reversing algebraic process calculi.} The Journal of Logic and Algebraic Programming, 2007, 73 (2007): 70--96.

\bibitem{TCSR} I. Phillips, I. Ulidowski.: \emph{True Concurrency Semantics via Reversibility.} http://www.researchgate.net/publication/266891384, 2014.

\bibitem{CR} I. Ulidowski, I. Phillips, and S. Yuen.: \emph{Concurrency and reversibility.} In RC, volume 8507 of LNCS, pages 1--14. Springer, 2014.

\bibitem{CTC}Y. Wang. A Calculus for True Concurrency. Manuscript, 2016. arXiv: 1703.00159.

\end{thebibliography}
